\documentclass[journal]{IEEEtran}

\usepackage[margin=1in]{geometry}

\usepackage{amssymb,amsmath,amsthm,amsfonts}
\usepackage{mathtools}
\usepackage{enumitem}
\usepackage[numbers,comma,sort&compress]{natbib}
\usepackage{authblk}
\usepackage{graphicx}
\usepackage[font=small]{caption}

\usepackage[labelformat=simple]{subcaption}
\usepackage{float}
\usepackage[ruled,vlined,linesnumbered]{algorithm2e}
\usepackage{physics}
\usepackage{footnote}
\usepackage{xcolor}
\usepackage{bbding}

\usepackage{tikz}
\usetikzlibrary{calc,patterns,angles,quotes}
\usepackage{siunitx}
\usetikzlibrary{quantikz}

\usepackage{hyperref}

\newtheorem{theorem}{Theorem}
\newtheorem{definition}{Definition}
\newtheorem{lemma}{Lemma}
\newtheorem{proposition}{Proposition}

\newtheorem{corollary}{Corollary}
\newtheorem{remark}{Remark}

\newcommand{\eq}[1]{Eq.~(\ref{eq:#1})}

\newcommand{\thm}[1]{\hyperref[thm:#1]{Theorem~\ref*{thm:#1}}}
\newcommand{\cor}[1]{\hyperref[cor:#1]{Corollary~\ref*{cor:#1}}}
\newcommand{\defn}[1]{\hyperref[defn:#1]{Definition~\ref*{defn:#1}}}
\newcommand{\lem}[1]{\hyperref[lem:#1]{Lemma~\ref*{lem:#1}}}
\newcommand{\prop}[1]{\hyperref[prop:#1]{Proposition~\ref*{prop:#1}}}
\newcommand{\fig}[1]{\hyperref[fig:#1]{Figure~\ref*{fig:#1}}}
\newcommand{\tab}[1]{\hyperref[tab:#1]{Table~\ref*{tab:#1}}}
\newcommand{\algo}[1]{\hyperref[algo:#1]{Algorithm~\ref*{algo:#1}}}
\renewcommand{\sec}[1]{\hyperref[sec:#1]{Section~\ref*{sec:#1}}}
\newcommand{\append}[1]{\hyperref[append:#1]{Appendix~\ref*{append:#1}}}
\newcommand{\fac}[1]{\hyperref[fac:#1]{Fact~\ref*{fac:#1}}}
\newcommand{\lin}[1]{\hyperref[lin:#1]{Line~\ref*{lin:#1}}}

\def\>{\rangle}
\def\<{\langle}

\renewcommand{\ket}[1]{|#1\rangle}

\newcommand{\E}{\mathbb{E}}

\DeclareMathOperator{\poly}{poly}
\DeclareMathOperator{\spn}{Span}

\renewcommand{\emptyset}{\varnothing}
\def\Tr{\operatorname{Tr}}\def\:{\hbox{\bf:}}

\newcommand{\specialcell}[2][c]{%
  \begin{tabular}[#1]{@{}c@{}}#2\end{tabular}}

\makeatletter
\newcommand{\doubletilde}[1]{{%
  \mathpalette\double@tilde{#1}%
}}
\newcommand{\double@tilde}[2]{%
  \sbox\z@{$\m@th#1\tilde{#2}$}%
  \ht\z@=.9\ht\z@
  \tilde{\box\z@}%
}
\makeatother

\let\oldnl\nl
\newcommand{\nonl}{\renewcommand{\nl}{\let\nl\oldnl}}

\SetKwComment{Comment}{/* }{ */}
\SetKwInput{KwInput}{Input}                
\SetKwInput{KwOutput}{Output}              
\SetKwInput{KwNotation}{Notation}              
\SetKwInput{KwParameter}{Parameter}             

\begin{document}

\title{A Quantum Algorithm Framework for Discrete Probability Distributions with Applications to R{\'e}nyi Entropy Estimation}

\author{Xinzhao Wang, Shengyu Zhang, Tongyang Li
\thanks{Tongyang Li and Xinzhao Wang are with Center on Frontiers of Computing Studies, Peking University, and School of Computer Science, Peking University. Shengyu Zhang is with Tencent Quantum Laboratory.}
}

\maketitle

\begin{abstract}
Estimating statistical properties is fundamental in statistics and computer science. In this paper, we propose a unified quantum algorithm framework for estimating properties of discrete probability distributions, with estimating R{\'e}nyi entropies as specific examples. In particular, given a quantum oracle that prepares an $n$-dimensional quantum state $\sum_{i=1}^{n}\sqrt{p_{i}}|i\>$, for $\alpha>1$ and $0<\alpha<1$, our algorithm framework estimates $\alpha$-R{\'e}nyi entropy $H_{\alpha}(p)$ to within additive error $\epsilon$ with probability at least $2/3$ using $\widetilde{\mathcal{O}}(n^{1-\frac{1}{2\alpha}}/\epsilon + \sqrt{n}/\epsilon^{1+\frac{1}{2\alpha}})$ and $\widetilde{\mathcal{O}}(n^{\frac{1}{2\alpha}}/\epsilon^{1+\frac{1}{2\alpha}})$ queries, respectively. This improves the best known dependence in $\epsilon$ as well as the joint dependence between $n$ and $1/\epsilon$. Technically, our quantum algorithms combine quantum singular value transformation, quantum annealing, and variable-time amplitude estimation. We believe that our algorithm framework is of general interest and has wide applications.
\end{abstract}




\section{Introduction}
\noindent
\textbf{Motivations.}
For many problems, quantum algorithms can dramatically outperform their classical counterparts. Among those, an important category is quantum algorithms for linear algebraic problems. Recently, Gily{\'e}n, Low, Su, and Wiebe~\cite{gilyen2019singular} proposed a powerful framework for quantum matrix arithmetics, namely \emph{quantum singular value transformation (QSVT)}. QSVT encompasses quantum algorithms for various problems (see also~\cite{martyn2021grand}), and can recover the best-known or even optimal quantum algorithms for fixed-point amplitude amplification~\cite{grover2005fixed,tulsi2006new,aaronson2012quantum,yoder2014fixed}, solving linear systems~\cite{harrow2009quantum,ambainis2012variable,childs2017,chakraborty2019linear}, Hamiltonian simulation~\cite{low2017optimal,low2019hamiltonian}, etc.

In this paper, we study a fundamental problem in statistics, theoretical computer science, and machine learning: \emph{estimating statistical properties}, which aims to estimate properties of probability distributions using the least number of independent samples. On the one hand, statistical properties such as entropies, divergences, etc., characterize some key measures of randomness. On the other hand, relevant theoretical tools are rapidly developing in topics such as property testing~\cite{ron2010algorithmic}, statistical learning~\cite{valiant2011testing}, etc. Among statistical properties, the most basic one is the Shannon entropy~\cite{ShannonEntropy}. For a discrete distribution $\mathbf{p}=(p_{i})_{i=1}^{n}$ supported on $[n]$, it is defined as
\begin{align}
    H(\mathbf{p}):=-\sum_{i=1}^{n}p_{i}\log p_{i}.
\end{align}
A natural generalization of the Shannon entropy is the family of R{\'e}nyi entropies~\cite{renyi1961measures}. Specifically, the $\alpha$-R{\'e}nyi entropy is defined as
\begin{align}\label{eq:Renyi}
    H_\alpha(\mathbf{p}) := \frac{1}{1-\alpha} \log \sum_{i=1}^n p_i^\alpha.
\end{align}
For our convenience, the power sum in the logarithm is denoted by $P_\alpha(\mathbf{p})$, i.e., $P_\alpha(\mathbf{p}) := \sum_{i=1}^n p_i^\alpha$. When $\alpha\to 1$, $\lim_{\alpha\to 1} H_{\alpha}(\mathbf{p})=H(\mathbf{p})$. Classically, references~\cite{jiao2015minimax,wu2016minimax} proved the tight classical sample complexity bound
\begin{align}\label{eq:Shannon-bound-classical}
\Theta\left(\frac{\log^{2}n}{\epsilon^{2}}+\frac{n}{\epsilon \log n}\right)
\end{align}
for estimating Shannon entropy within precision $\epsilon$ with success probability at least $2/3$. For $\alpha$-R{\'e}nyi entropy estimation, reference~\cite{acharya2016estimating} proved that when $\alpha>1$ and $0<\alpha<1$ respectively, it takes $\mathcal{O}(n/\log n)$ and $\mathcal{O}(n^{1/\alpha}/\log n)$ independent samples from $\mathbf{p}$ respectively to estimate $H_{\alpha}(\mathbf{p})$ within constant additive error with probability at least $2/3$. In addition, for any constant $\eta>0$, the paper also established sample complexity lower bounds $\Omega(n^{1-\eta})$ and $\Omega(n^{1/\alpha-\eta})$ when $\alpha>1$ and $0<\alpha<1$, respectively.

There has also been literature on quantum algorithms for entropy estimation (see the paragraph on related works for more details). Among those, the state-of-the-art result on estimating Shannon entropy was given by Gily{\'e}n and Li~\cite{gilyen2019distributional}, which applies QSVT to estimate the Shannon entropy within additive error $\epsilon$ with success probability at least $2/3$ using $\widetilde{\mathcal{O}}(\sqrt{n}/\epsilon^{1.5})$ quantum queries. For $\alpha$-R{\'e}nyi entropy, Li and Wu~\cite{li2019entropy} gave algorithms with quantum query complexities $\widetilde{\mathcal{O}}(n^{1-1/2\alpha}/\epsilon^{2})$ and $\widetilde{\mathcal{O}}(n^{1/\alpha-1/2}/\epsilon^{2})$ when $\alpha>1$ and $0<\alpha<1$, respectively. Both papers used a common model proposed by Bravyi et al.~\cite{bravyi2011quantum} which encodes $\mathbf{p}$ as frequencies of $n$ symbols in a given input string and quantum algorithms can access the input string in superposition (\defn{discrete-quantum-query}), whereas~\cite{gilyen2019distributional} also adopted oracles preparing a superposed quantum state whose amplitude in the $i^{\text{th}}$ term is $p_{i}$ (\defn{pure-state-preparation} and \defn{purified-query-access}).

Nevertheless, it can be observed that although quantum algorithms for learning statistical properties have applied advanced algorithmic tools including quantum singular value transformation~\cite{gilyen2019singular}, and have achieved speedup in the cardinality $n$ and precision $\epsilon$ separately, the combined dependence on $n$ and $\epsilon$ is not yet as well understood as the classical counterparts, for instance the sample complexity of Shannon entropy in \eq{Shannon-bound-classical}. From a high-level perspective, even though quantum algorithms for linear algebraic problems have been systematically developed, we shall still endeavor to quantum algorithms with optimal or near-optimal dependence on \emph{all parameters}. In this paper, we shed light on this question for estimating statistical properties.


\vspace{3mm}
\noindent
\textbf{Contributions.}
In this paper, we introduce a unified quantum algorithm framework for estimating properties of discrete distributions. Our algorithm is stemmed from quantum singular value transformation~\cite{gilyen2019singular}, but we enhance the framework with quantum annealing and variable-time amplitude amplification and estimation. Specifically, we propose algorithms for estimating R{\'e}nyi entropies of discrete probability distributions with refined dependence on $n$ and $\epsilon$, assuming access to quantum oracle $U_{\mathrm{pure}}$ which maps $\ket{0}$ to $\sum_{i=1}^n \sqrt{p_i} \ket{i}$ (see the later ``related work'' paragraph for more discussions and comparisons of different oracles).\footnote{In fact, our quantum algorithm also applies to the purified quantum query-access in \defn{purified-query-access}. Please see \sec{main-algorithm} and \sec{applications} for more details.}

\begin{theorem}[Main theorem]\label{thm:main}
There are quantum algorithms that approximate the R{\'e}nyi entropy $H_\alpha(\mathbf{p})$ in \eq{Renyi} within an additive error $\epsilon>0$ with success probability at least $2/3$ using
\begin{itemize}[leftmargin=15pt]
    \item $\widetilde{\mathcal{O}}\left(\frac{n^{1-\frac{1}{2\alpha}}}{\epsilon} + \frac{\sqrt{n}}{\epsilon^{1+\frac{1}{2\alpha}}}\right)$
 quantum queries to $U_{\mathrm{pure}}$ and $U_{\mathrm{pure}}^\dagger$ in \defn{pure-state-preparation} when $\alpha>1$ (\thm{main-alpha-large}), and
    \item $\widetilde{\mathcal{O}}\left(\frac{n^{\frac{1}{2\alpha}}}{\epsilon^{\frac{1}{2\alpha}+1}}\right)$
 quantum queries to $U_{\mathrm{pure}}$ and $U_{\mathrm{pure}}^\dagger$ in \defn{pure-state-preparation} when $0<\alpha<1$ (\thm{main-alpha-small}).
\end{itemize}
\end{theorem}

Compared to the state-of-the-art result for estimating R{\'e}nyi entropies by Li and Wu~\cite{li2019entropy} which uses $\widetilde{\mathcal{O}}(n^{1/\alpha-1/2}/\epsilon^{2})$ quantum queries when $0<\alpha<1$ and $\widetilde{\mathcal{O}}(n^{1-1/2\alpha}/\epsilon^{2})$ quantum queries when $\alpha>1$ and $\alpha$ is not an integer, our result achieves a systematic improvement in both $n$ and $\epsilon$. This can be illustrated by \fig{main}.\footnote{\label{ft-low}The integral $\alpha$ cases are excluded in the figure because computing $H_\alpha(p)$ for integral $\alpha$ seems fundamentally easier. Classically, the best-known upper bound for integral $\alpha>2$ is $\Theta(n^{1-\frac{1}{\alpha}})$, smaller than that of $\Omega(n^{1-o(1)})$ for non-integral cases \cite{acharya2016estimating}. For quantum algorithms, Li and Wu~\cite{li2019entropy} made special designs for integer $\alpha$ cases, with query cost better than their non-integral $\alpha$ 
cases (and also ours), albeit using a stronger input oracle (\defn{discrete-quantum-query}).} 

\textcolor{black}{The $\epsilon$ dependence of our algorithm seems to be worse than that of Li and Wu \cite{li2019entropy} when $\alpha \in (0,\frac{1}{2})$. We suspect this is due to an error of the analysis of their Theorem 9 and we have fixed it in  \sec{eps-dependence}. The analysis of Theorem 14 in the arXiv version of \cite{acharya2016estimating} also seems to have an error, which analyzed the classical sample complexity of estimating R{\'e}nyi entropy for $\alpha \in (0,1)$. We note that Jiao et at.~\cite{jiao2015minimax} gave a R{\'e}nyi entropy estimation algorithm with different classical sample complexity for $\alpha \in (0,1)$, so we only compare our algorithms with that of Jiao et at.~\cite{jiao2015minimax}. We discuss these points also in \sec{eps-dependence}.}

\begin{figure*}[htbp]
\centering
\subcaptionbox{$\alpha > 1$, $\alpha \not\in \mathbb{N}$, $\epsilon = n^{-0.5}$}{
\includegraphics[width=5cm]{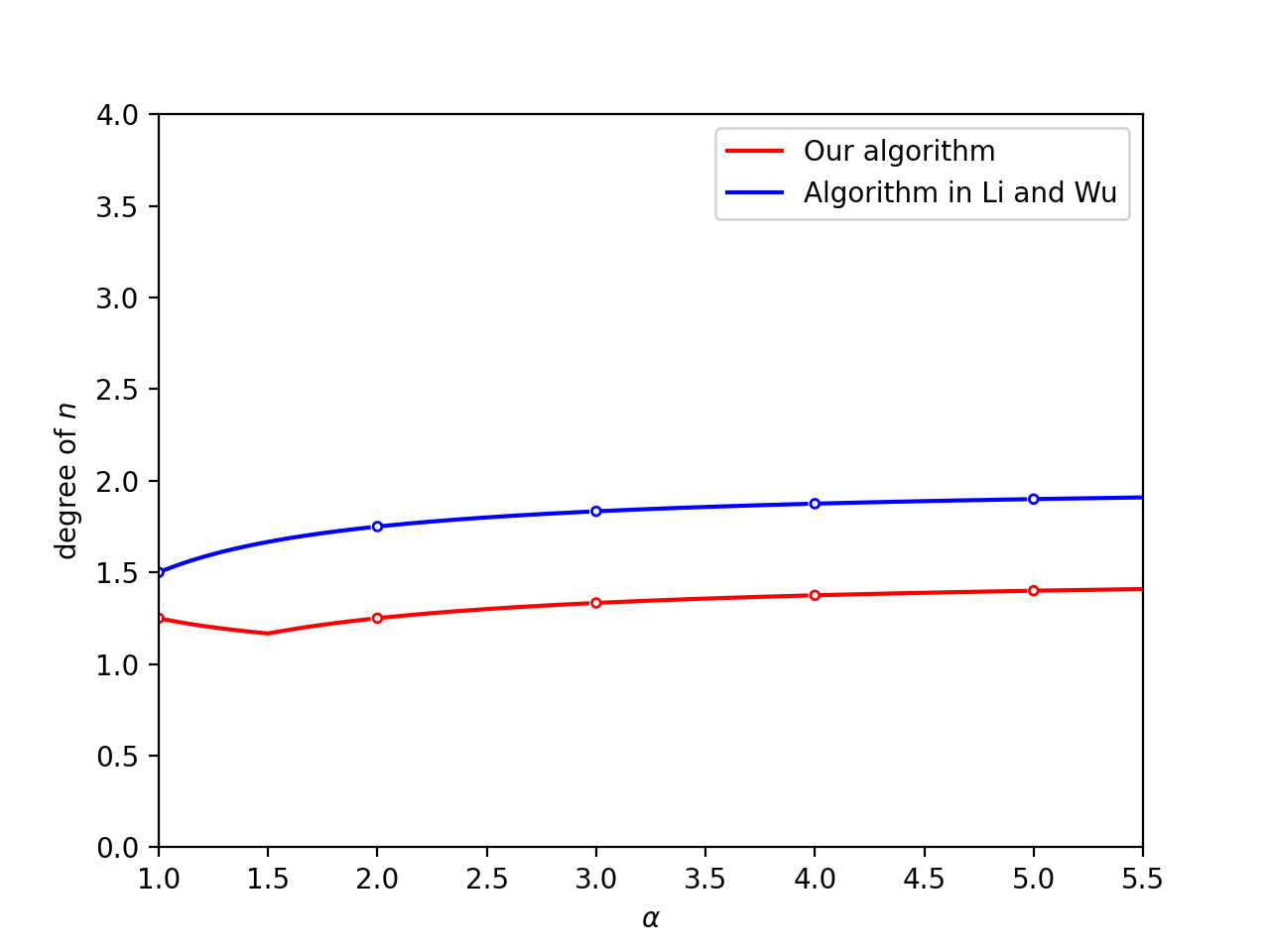}
}
\subcaptionbox{$\alpha > 1$, $\alpha \not\in \mathbb{N}$, $\epsilon = n^{-0.25}$}{
\includegraphics[width=5cm]{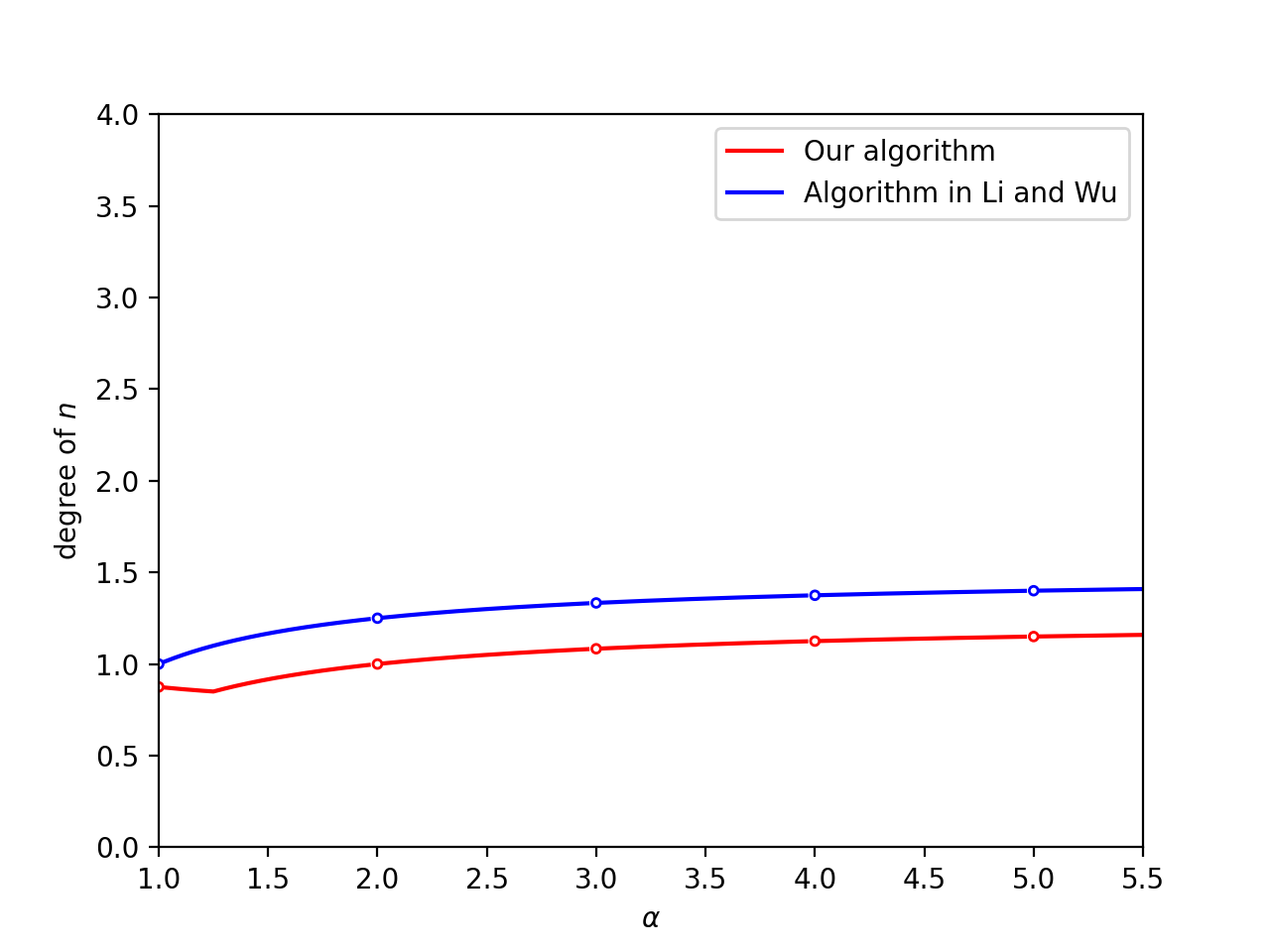}
}
\subcaptionbox{$0<\alpha<1$, $\epsilon = \Theta(1)$}{
\includegraphics[width=5cm]{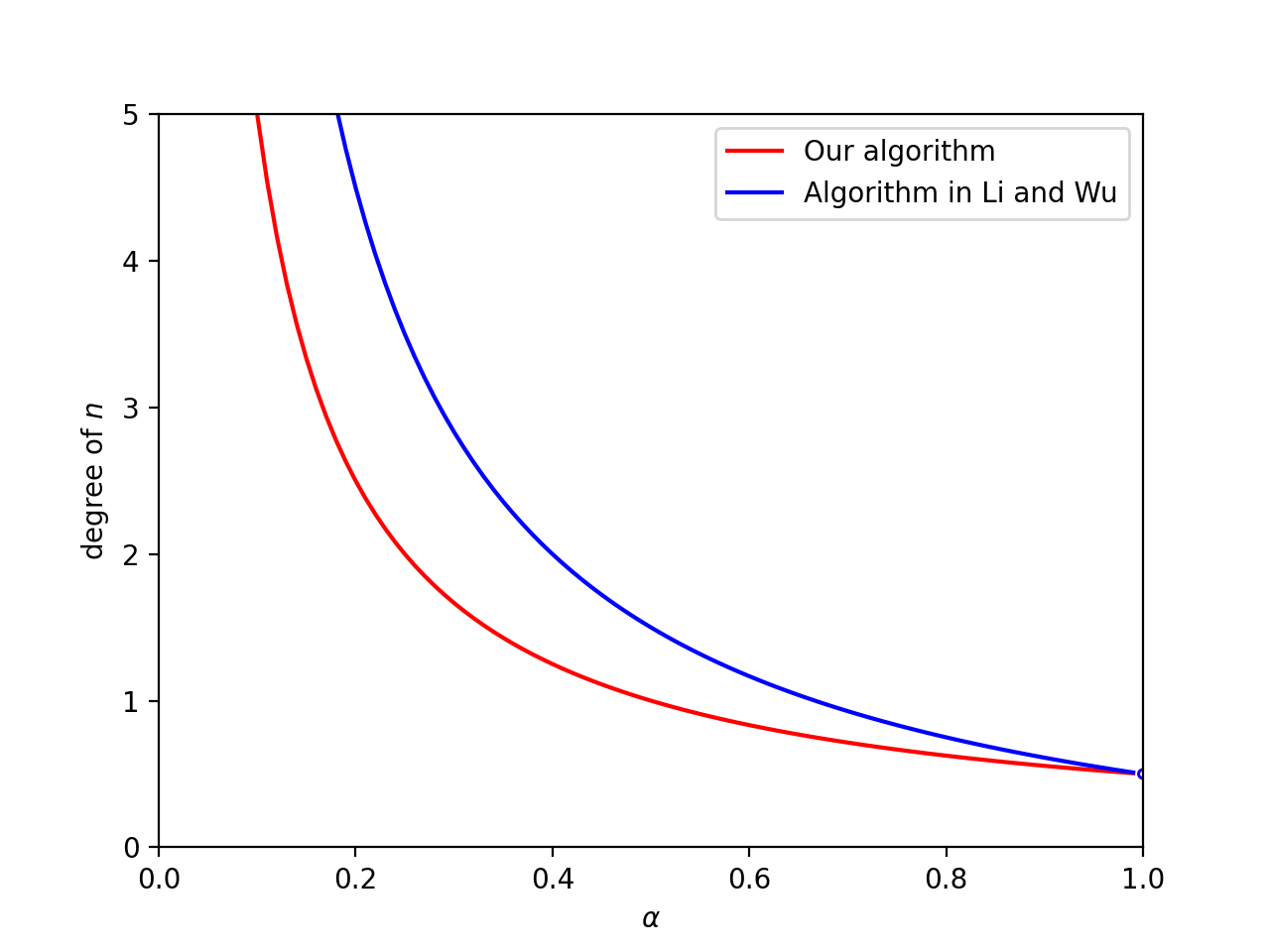}}
\caption{Comparison between our algorithms and the algorithm in Li and Wu~\cite{li2019entropy}.}
\label{fig:main}
\end{figure*}

Here we list current query complexity lower bounds to estimate R{\'e}nyi entropy with $U_{\mathrm{pure}}$ and $U_{\mathrm{pure}}^\dagger$. 
\begin{itemize}[leftmargin=15pt]
    \item \textcolor{black}{For $\alpha \in(0,1)$, we proved that $\Omega\left(\frac{n^{1/2\alpha-1/2}}{\epsilon^{1/2\alpha}}\right)$ queries to $U_{\mathrm{pure}}$ and $U_{\mathrm{pure}}^\dagger$ are necessary to estimate $H_{\alpha}(\mathbf{p})$ to error $\epsilon$ in \thm{lower}. This almost matches our upper bound when $\alpha$ tends to $0$.} 
    \item \textcolor{black}{For $\alpha \in [\frac{3}{7},3]$ and $\epsilon \in [\frac{1}{n},1]$, Li and Wu \cite{li2019entropy} proved that $\Omega\bigl(n^{\frac{1}{3}}/\epsilon^{\frac{1}{6}}\bigr)$ queries are necessary. For $\alpha =1$, Bun et al. \cite{bun2018polynomial} improved the lower bound to $\tilde{\Omega}(\sqrt{n})$.}
    \item \textcolor{black}{For $\alpha \in[3,\infty)$, Li and Wu \cite{li2019entropy} proved that $\Omega\Bigl(\frac{n^{\frac{1}{2 }-\frac{1}{2 \alpha}}}{\epsilon}\Bigr)$ queries are necessary, so our upper bound has an $\widetilde{\mathcal{O}}\bigl(n^{\frac{1}{2}-\frac{1}{2\alpha}}+(\frac{n}{\epsilon})^{\frac{1}{2\alpha}}\bigr)$ overhead. However, as mentioned in footnote \ref{ft-low}, estimating R{\'e}nyi entropy for integral and non-integral $\alpha$ have fundamental differences in the classical case, and the lower bound in \cite{li2019entropy} holds for all $\alpha \ge 3$, suggesting that it may not be tight for $\alpha \not \in \mathbb{N}$. }
\end{itemize}

\textcolor{black}{We also applied our algorithms to sparse or low-rank distributions. If a classical probability distribution $\mathbf{p}$ has at most $r$ elements $i$ such that $p_i > 0$ and we know the value of $r$ in advance, we give an algorithm using  $\widetilde{\mathcal{O}}\left(\frac{r^{1-\frac{1}{2\alpha}}}{\epsilon} + \frac{\sqrt{r}}{\epsilon^{1+\frac{1}{2\alpha}}}\right)$ calls to $U_{\mathrm{pure}}$ and $U_{\mathrm{pure}}^{\dagger}$ to estimate $H_{\alpha}(\mathbf{p})$ to within additive error $\epsilon$ when $\alpha >  1$, and an algorithm using $\widetilde{\mathcal{O}}\left(\frac{r^{\frac{1}{2\alpha}}}{\epsilon^{\frac{1}{2\alpha}+1}}\right)$ calls to $U_{\mathrm{pure}}$ and $U_{\mathrm{pure}}^{\dagger}$ to estimate $H_{\alpha}(\mathbf{p})$ to within additive error $\epsilon$ when $0<\alpha<1$. In addition, we also give a quantum algorithm in \cor{large-alpha-not-r} for $\alpha > 1$ when we do not know the value of $r$.}

\color{black}
Our quantum algorithms can be applied to estimate the R{\'e}nyi entropy
\begin{align}
\label{eq:quantum-renyi}
    H_{\alpha}(\rho) = \frac{1}{1-\alpha}\log(\Tr(\rho^{\alpha})).
\end{align}
of a quantum density matrix $\rho$.
\color{black}

\begin{corollary}
\label{cor:renyi-density}
There are quantum algorithms that approximate the R{\'e}nyi entropy of a density operator $H_\alpha(\rho)$ in \eq{quantum-renyi} within an additive error $\epsilon>0$ with success probability at least $2/3$ using
\begin{itemize}[leftmargin=15pt]
    \item $\widetilde{\mathcal{O}}\Bigl(\min\Bigl(\frac{n^{\frac{3}{2}-\frac{1}{2\alpha}}}{\epsilon} + \frac{n}{\epsilon^{1+\frac{1}{2\alpha}}}, \frac{n}{\epsilon^{\frac{1}{\alpha}+1}}\Bigr)\Bigr)$
 quantum queries to $U_{\rho}$ and $U_{\rho}^\dagger$ in \defn{purified-query-access} when $\alpha>1$ (\cor{main-alpha-large-purified}), and
    \item $\widetilde{\mathcal{O}}\left(\frac{n^{\frac{1}{2\alpha}+\frac{1}{2}}}{\epsilon^{\frac{1}{2\alpha}+1}}\right)$
 quantum queries to $U_{\rho}$ and $U_{\rho}^\dagger$ in \defn{purified-query-access} when $0<\alpha<1$ (\cor{main-alpha-small-purified}).
\end{itemize}
\end{corollary}

\textcolor{black}{Wang et al.~\cite{wang2022new} studied estimating R{\'e}nyi entropy of density operators in low-rank cases, we show in \sec{applications} that our techniques can improve their algorithms. Subramanian and Hsieh \cite{subramanian2021quantum} consider the same task as in \cor{renyi-density}. Their algorithm used sampling methods instead of Amplitude Estimation, so it has worse asymptotic query complexity bound but requires less stringent quantum resources.}

In addition, our quantum algorithms can also be extended to estimate quantum R{\'e}nyi divergence of density matrices. Please find details in \sec{applications}

\vspace{3mm}
\noindent
\textbf{Techniques.}
Our quantum algorithm follows the one in Gily{\'e}n and Li~\cite{gilyen2019distributional} for Shannon entropy estimation. In \sec{main_algo_standard}, we first construct a unitary operator, which has a matrix block encoding of the square root of the probability distribution we want to study. We then use quantum singular value transformation~\cite{gilyen2019singular} to compute a polynomial approximation of the function that we want to estimate, which is then encoded into an amplitude. Finally, we apply  amplitude estimation to obtain the estimate as a classical output.

However, this algorithm is sub-optimal in many cases because of the following two reasons. 
\begin{itemize}[leftmargin=15pt]
    \item If we need an estimate to within a given multiplicative error, the query complexity of the amplitude estimation process is inversely proportional to the square root of the quantity that we want to estimate, so it has poor dependence if the quantity is too small.
    
    \item Quantum singular value transformation leverages the computation of the same function of all singular values in parallel, but this brings restrictions that the polynomial should well-approximate the function within the entire interval $[0,1]$, where the singular values may range over.
\end{itemize}
For the first issue, we design an annealing process in \sec{annealing} to obtain an estimate of the quantity to within constant multiplicative error in advance. With this rough estimate, we can amplify the quantity with smaller overhead in each step. For the second issue, we exploit variable-time amplitude estimation in \sec{main_algo_2} by designing a variable-stopping-time algorithm which applies different transformation polynomials to singular values in different intervals. This give us more flexibility to construct polynomials for different singular values and the final complexity is related to the average degree of all polynomials. 

As a technical contribution, we also improve the bounded polynomial approximation of $x^\alpha$ for $\alpha > 1$ (\lem{scaled_poly_approx_with_small_value_at_zero}), which may be of independent interest. Our approximation polynomial is bounded by $2x^{\alpha}$ when $x$ is smaller than a threshold while the bounded polynomial approximation constructed in \cite{gilyen2019singular} is only guaranteed to be bounded by 1. 

We summarize and compare the techniques in previous literature on quantum algorithms for estimating entropies of discrete probability distributions in \tab{main}.

\begin{table*}[ht]
  \centering
  \resizebox{1.6\columnwidth}{!}{\begin{tabular}{ccccc}
    \hline
    Reference                       & Oracle               & Quantum SVT  & Annealing         & VTAE    \\ \hline\hline
    \cite{bravyi2011quantum}  & Discrete query-access (\defn{discrete-quantum-query}) & \XSolidBrush & \XSolidBrush & \XSolidBrush  \\ \hline
    \cite{li2019entropy} &  Discrete query-access (\defn{discrete-quantum-query})  & \XSolidBrush & \Checkmark   & \XSolidBrush \\ \hline
    \cite{gilyen2019distributional} & Purified query-access (\defn{purified-query-access}) & \Checkmark & \XSolidBrush & \XSolidBrush \\ \hline
    \cite{gur2021sublinear} & Purified query-access (\defn{purified-query-access}) & \Checkmark & \XSolidBrush & \XSolidBrush \\ \hline\hline
    \textbf{this paper}  & \specialcell{Pure-state preparation (\defn{pure-state-preparation})\\Purified query-access (\defn{purified-query-access})} & \Checkmark & \Checkmark & \Checkmark   \\ \hline
  \end{tabular}}
  \caption{Summary of quantum algorithms for estimating entropies of classical discrete distributions.}
  \label{tab:main}
\end{table*}

We give a brief explanation of the comparisons in the table here. Compared with quantum algorithms without using QSVT to estimate entropy, they first sample $i$ according to $p_i$, then estimate $p_i$ using amplitude estimation and compute the entropy accordingly, while using QSVT we can directly compute any polynomial function value of $p_i$ for all $i$ simultaneously in the block encoding. On the other hand, VTAE is an accelerated version of amplitude estimation that takes problem instances into account. In our problem, VTAE allows us to apply QSVT with different polynomials to $p_i$ in different intervals, which makes our quantum algorithm more flexible. Moreover, annealing is applied to handle the issue that that the estimated quantity is too small and makes the amplitude estimation costly. By using annealing, we can obtain a rough estimate of the estimated quantity, which enlarges the estimated quantity when using QSVT.

\vspace{3mm}
\noindent
\textbf{Related work.}
Previous literature investigated quantum algorithms for estimating statistical properties using different input models (see also the survey paper~\cite{montanaro2016survey}). First, if we want to utilize quantum algorithms to accelerate the solving of problems related to classical distributions, we need coherent access to classical distributions via quantum oracle. It is thus natural to consider a unitary oracle which can prepare a pure state encoding a classical distribution as follows:

\begin{definition}[Pure-state preparation access to classical distribution]\label{defn:pure-state-preparation}
    A classical distribution $\mathbf{p} = (p_i)_{i=1}^{n}$ is accessible via pure-state preparation access if we have access to a unitary oracle $U_{\mathrm{pure}}$ and its inverse, which satisfies
    \begin{align}\label{eq:pure-state-preparation}
        U_{\mathrm{pure}}|\mathbf{0}\rangle = \sum_{i=1}^n \sqrt{p_i} |i\rangle.
    \end{align}
\end{definition}

\textcolor{black}{This oracle can be traced back to the \textit{quantum example oracle} proposed by \cite{bshouty1995learning}.}

\textcolor{black}{Another common model, originally proposed by Bravyi et al.~\cite{bravyi2011quantum}, encodes the classical probability distribution as frequencies of $n$ symbols in a given input string, and quantum algorithms can query the input string in superposition. Note that amplitude estimation in \cite{brassard2002quantum} can be regarded as estimating the mean of a random variable encoded in this way.}


\begin{definition}[Discrete quantum query-access to classical distribution]\label{defn:discrete-quantum-query}
    A classical distribution $\mathbf{p} = (p_i)_{i=1}^{n}$ is accessible via discrete quantum query-access if we have quantum access to a function $f\colon S\to [n]$ such that for all $i\in[n]$, $p_i = |\{s \in S\mid f(s) = i\}|/|S|$, which means we have access to a unitary oracle $O$ and its inverse acting on $\mathbb{C}^{|S|}\otimes \mathbb{C}^{n}$ such that
    \begin{align}\label{eq:discrete-quantum-query}
        O|s\rangle|\mathbf{0}\rangle = |s\rangle|f(s)\rangle \text{ for all } s\in S.
    \end{align}
\end{definition}

In this model, Bravyi et al.~\cite{bravyi2011quantum} gave a quantum algorithm to estimate the $\ell_{1}$-norm distance of two distributions $p$ and $q$ with support cardinality $n$ and with constant precision using $\mathcal{O}(\sqrt{n})$ queries, and gave quantum algorithms for testing uniformity and orthogonality with query complexity $\mathcal{O}(n^{1/3})$. This was later generalized to identity testing, i.e., testing whether a distribution is identical or $\epsilon$-far in $\ell_{1}$-norm from a given distribution, in $\widetilde{\mathcal{O}}(n^{1/3})$ queries by~\cite{chakraborty2010testing}. Li and Wu~\cite{li2019entropy} gave a quantum algorithm for estimating the Shannon entropy within additive error $\epsilon$ with high success probability using $\widetilde{\mathcal{O}}(n^{1/2}/\epsilon^{2})$ queries, and this paper also studied the query complexity of R{\'e}nyi entropy estimation (see the paragraph of ``contributions"). To complement the algorithm results, Bun et al.~\cite{bun2018polynomial} proved that Shannon entropy estimation with a certain constant $\epsilon$ requires $\widetilde{\Omega}(\sqrt{n})$ quantum queries to the oracle in \eq{discrete-quantum-query}.

Beyond classical distributions, it is natural to extend to statistical problems of genuine quantum systems. The quantum counterpart of a classical discrete distribution is a \emph{density matrix}. Density matrices can be regarded as the (possibly random) outcome of some physical process, and if we can access this physical process by calling it as a black box, we can generate quantum samples ourselves. If the physical process is reversible, which is common in a quantum scenario, we can also access the inverse process of it. For example, if a quantum computer produces the state $\rho$ without measurements, we can easily reverse this process. We can define the following input model to characterize the situations mentioned above.
\begin{definition}[Purified quantum query-access]
\label{defn:purified-query-access}
    A density operators $\rho\in \mathbb{C}^{n\times n}$ has purified quantum query-access if we have access to a unitary oracle $U_{\rho}$ and its inverse, which satisfies
    \begin{align}\label{eq:purified-query-access}
        U_{\rho}|\mathbf{0}\rangle_A|\mathbf{0}\rangle_B =|\psi_{\rho}\rangle= \sum_{i=1}^n\sqrt{p_i} |\phi_i\rangle_A|\psi_i\rangle_B
    \end{align}
    such that $\Tr_A(|\psi_{\rho}\rangle\langle \psi_{\rho}|) = \rho$, where $\langle\phi_i|\phi_j\rangle = \langle \psi_i|\psi_j\rangle = \delta_{ij}$. If $|\psi_i\rangle = |i\rangle$, $\rho = \sum_{i=1}^n p_i |i\rangle\langle i|$ is a diagonal density operator which can be seen as a classical distribution $\mathbf{p} = (p_i)_{i=1}^n$, and we write $U_p$ in this case instead of $U_\rho$.
\end{definition}

We note that for encoding classical distributions, \defn{purified-query-access} is weaker than \defn{discrete-quantum-query} since we can apply $O$ to a uniform superposition over $S$ in \eq{discrete-quantum-query}, and this is equivalent to applying a purified quantum query-access encoding a classical distribution to $|\mathbf{0}\rangle$. (Furthermore, \defn{discrete-quantum-query} essentially assumes that all probabilities $p_{i}$ are rational, whereas \defn{purified-query-access} does not have this requirement.) 
In addition, \defn{purified-query-access} is also weaker than~\defn{pure-state-preparation} since we can use one query to $U_{\mathrm{pure}}$ to prepare $\sum_{i=1}^n \sqrt{p_i}|i\rangle$, and then apply CNOT gates to produce the state $\sum_{i=1}^n \sqrt{p_i}|i\rangle|i\rangle$, which satisfies the condition in \defn{purified-query-access}. Our results are established with \defn{purified-query-access} being the input oracle.

For classical distributions encoded by \defn{purified-query-access}, Gily{\'e}n and Li~\cite{gilyen2019distributional} systematically studied different oracle access of distributional property testing, and proved that it takes $\widetilde{\mathcal{O}}(n^{1/2}/\epsilon^{1.5})$ queries to the purified query access for estimating Shannon entropy to within additive error $\epsilon$ with high success probability. This work also studied closeness testing, where we are given purified query access to distributions $\mathbf{p}$ and $\mathbf{q}$ and the goal is to distinguish between $\mathbf{p}=\mathbf{q}$ and $\|\mathbf{p}-\mathbf{q}\|\geq\epsilon$. For $\ell_{1}$-norm and $\ell_{2}$-norm distances,~\cite{gilyen2019distributional} proved that the quantum query complexities are $\widetilde{\mathcal{O}}(\sqrt{n}/\epsilon)$ and $\widetilde{\Theta}(1/\epsilon)$, respectively.  Belovs~\cite{belovs2019quantum} proved that distinguishing between $\mathbf{p}$ and $\mathbf{q}$ takes $\Theta(1/d_{\textrm{H}}(\mathbf{p},\mathbf{q}))$ queries (see also \sec{small-alpha-lower}), where $d_{\textrm{H}}(\mathbf{p},\mathbf{q})$ is the Hellinger distance between $\mathbf{p}$ and $\mathbf{q}$, and this tight bound applies to all oracles in \defn{pure-state-preparation}, \defn{discrete-quantum-query}, and  \defn{purified-query-access}.

\textcolor{black}{For quantum density matrix, Watrous \cite{watrous2002limits} used this oracle to access a mixed state implicitly. }\defn{purified-query-access} is also widely used among quantum algorithms for estimating properties of quantum density operators. The results in~\cite{gilyen2019distributional} about Shannon entropy estimation and $\ell_{1}$-norm and $\ell_{2}$-norm closeness testing  can be generalized to those of quantum density matrices with purification with an overhead of $\sqrt{n}$. \textcolor{black}{Chowdhury et al.~\cite{chowdhury2020variational} estimates the von Neumann entropy of quantum density matrices to within an additive error}. Gur et al.~\cite{gur2021sublinear} estimates the von Neumann entropy of quantum density matrices to within a certain multiplicative error, and under appropriate choices of parameters the query complexity to the purified query access can be sublinear in $n$. Regarding the estimation of quantum R{\'e}nyi entropy in general, Subramanian and Hsieh~\cite{subramanian2021quantum} used $\widetilde{\mathcal{O}}(\kappa n^{\max\{2\alpha,2\}}/\epsilon^{2})$ queries to estimate the $\alpha$-R{\'e}nyi entropy of a density matrix $\rho$ satisfying $I/\kappa\preceq\rho\preceq I$ to within additive error $\epsilon$. When $\rho$ has rank at most $r$, Wang et al.~\cite{wang2022new} gave quantum algorithms taking $\poly(r,1/\epsilon)$ queries for estimating von Neumann entropy, quantum R{\'e}nyi entropy, and trace distance and fidelity between two density matrices. \textcolor{black}{Fidelity estimation \cite{wang2022fidelity,gilyen2022improved}, trace distance estimation \cite{wang2023fast}, and quantum state tomography \cite{van2023quantum} using \defn{purified-query-access} are also studied.}

Finally, since classical algorithms for estimating distribution properties takes independent samples, it is natural to consider quantum samples of density operators defined as follows.




\begin{definition}[Quantum sampling]
    \label{defn:q-sample-def}
    A quantum distribution $\rho\in\mathbb{C}^{n\times n}$ is accessible via quantum sampling if we can request independent copies of the state $\rho$.
\end{definition}

\textcolor{black}{Childs et al.~\cite{childs2007weak} studied sample complexity of the quantum collision problem in this model and proved weak Fourier-Schur sampling fails to identify the hidden subgroup in HSP problem.} A series of papers by O'Donnell and Wright~\cite{OW15,OW16,OW17} (see also their survey paper~\cite{odonnell2017guest}) studied the sample complexity of various problems, including quantum state tomography, maximally mixedness testing, rankness testing, spectrum estimation, learning eigenvalues, learning top-$k$ eigenvalues, and learning optimal rank-$k$ approximation. Subsequently, B{\u{a}}descu, O'Donnell, and Wright~\cite{badescu2019quantum} studied the sample complexity of testing whether $\rho$ is equal to some known density matrix or $\epsilon$-far from it, which is $O(n/\epsilon)$ with respect to fidelity and $O(n/\epsilon^2)$ with respect to trace distance; both results are optimal up to constant factors. Regarding von Neumann and quantum R{\'e}nyi entropies, Acharya et al.~\cite{acharya2019measuring} proved that estimation with additive error $\epsilon$ of von Neumann entropy, quantum R{\'e}nyi entropy with $\alpha>1$, and quantum R{\'e}nyi entropy with $0<\alpha<1$ have sample complexity bounds $\mathcal{O}(n^{2}/\epsilon^{2})$ and $\Omega(n^{2}/\epsilon)$, $\mathcal{O}(n^{2}/\epsilon^{2})$ and $\Omega(n^{2}/\epsilon)$, and $\mathcal{O}(n^{2/\alpha}/\epsilon^{2/\alpha})$ and $\Omega(n^{1+1/\alpha}/\epsilon^{1/\alpha})$, respectively. Given an additional assumption that all nonzero eigenvalues of $\rho$ are at least $1/\kappa$, Wang et al.~\cite{wang2022quantum} gave a quantum algorithm for estimating its von Neumann entropy using $\widetilde{\mathcal{O}}(\kappa^{2}/\epsilon^{5})$ samples, and bounds under the same assumption were also proved for estimating quantum R{\'e}nyi entropy.

\vspace{3mm}
\noindent
\textbf{Open questions.}
Our work raises several natural questions for future investigation:
\begin{itemize}[leftmargin=*]
\item When $\alpha>1$, can we achieve quadratic quantum speedup in $n$ compared to the classical algorithm in~\cite{acharya2016estimating} for estimating $\alpha$-R{\'e}nyi entropy with $\mathcal{O}(n/\log n)$ queries? A natural goal is to give a quantum algorithm with query complexity $\widetilde{\mathcal{O}}(\sqrt{n})$ for constant $\epsilon$, but our current bound in \thm{main-alpha-large} has complexity $\widetilde{\mathcal{O}}(n^{1-\frac{1}{2\alpha}})$.
This may be related to our estimation paradigm. A classical analogy to our algorithm is to draw samples independently from the probability distribution $\mathbf{p}$ on $[n]$, estimate $(p_i)^{\alpha-1}$ for each sample $i$, and output the mean value of all estimates. Such algorithms are called empirical estimators, but they can be sub-optimal classically. 


\item Can we apply our quantum algorithm framework to other statistical problems? One possibility is the estimation of partition functions -- it is another prominent type of statistical properties, and many previous quantum algorithms including~\cite{wocjan2009quantum,montanaro2015quantum,chakrabarti2023quantum,harrow2020adaptive,arunachalam2021simpler} had applied annealing on the system's temperature to estimating partition functions. It would be of general interest to achieve further quantum speedup by our algorithm framework.

\item For other quantum linear algebraic problems, can we elaborate on the dependence on all parameters? Decent efforts had been conducted for Hamiltonian simulation~\cite{low2017optimal,low2019hamiltonian,gilyen2019singular} and linear system solving~\cite{harrow2009quantum,ambainis2012variable,childs2017,chakraborty2019linear}, and this work investigates the estimation of statistical properties. It would be natural to leverage refined analyses for more problems, for instance the applications in quantum machine learning.
\end{itemize}

\vspace{3mm}
\noindent
\textbf{Organization.}
The rest of the paper is organized as follows. We review necessary background in \sec{prelim}. We introduce our main technical contribution, our quantum algorithm framework, in \sec{main-algorithm}. We prove our results about the quantum query complexity of $\alpha$-R{\'e}nyi entropy estimation with $\alpha>1$ and $0<\alpha<1$ in \sec{alpha-large} and \sec{alpha-small}, respectively. In \sec{applications}, we describe further applications of our quantum algorithm framework in estimating statistical properties.

\vspace{3mm}
\noindent
\textbf{Notation.}
Throughout the paper, $\widetilde{\mathcal{O}}$ omits poly-logarithmic factors in the big-$\mathcal{O}$ notation, i.e., $\widetilde{\mathcal{O}}(g)=g\poly(\log g)$. Unless otherwise stated, all vector norms $\|\cdot \|$ in this paper are $\ell_2$-norm. We use $\log$ to represent $\log_2$ and $\ln$ to represent $\log_e$. We use $\Delta^n$ to represent the set of all probability distributions on $[n]$. For a set $A$, we use $|A|$ to represent the size of $A$. In description of quantum algorithms, the corresponding Hilbert space of a quantum register $X$ is denoted by $\mathcal{H}_X$. We write operator $A$ acting on Hilbert space $\mathcal{H}_X$ as $A_X$. We use $I$ to represent the identity oprator and $|\mathbf{0}\rangle$ to represent the all-0 state. 


\section{Preliminaries}\label{sec:prelim}
We summarize necessary tools used in our quantum algorithm framework as follows.

\subsection{Amplitude amplification and estimation}
\noindent
\textbf{Fixed-point amplitude amplification.}
Classically, for a Bernoulli random variable $X$ with $\E[X] = p$, we need $\Theta(1/p)$ i.i.d.~samples in expectation to observe the first 1. In the quantum case, this can be improved by {\it amplitude amplification}~\cite{brassard2002quantum}, a quantum algorithm in which the number of iterations depends on $p$.  This was later strengthened to a fixed-point version, where the algorithm only needs to know a lower bound of $p$. There are a number of implementations~\cite{hoyer2000arbitrary,grover2005fixed,tulsi2006new,aaronson2012quantum,yoder2014fixed}, and here we use a version given in \cite{gilyen2019singular}.
Let $|\mathbf{0}\rangle$ denote the all-0 initial state. 
Consider a unitary $U$ such that
\begin{equation}
    \label{eq:quantum_bernoulli}
    U|0\rangle|\mathbf{0}\rangle = \sqrt{p}|1\rangle|\phi\rangle + \sqrt{1-p}|0\rangle|\psi\rangle.
\end{equation}

The following theorem says that we can obtain an approximation of $|\phi\rangle$ using $\Theta(\frac{1}{\sqrt{p}})$ calls to $U$ and $U^{\dagger}$, achieving a quadratic quantum speedup over its classical counterpart. 

\begin{theorem}[Fixed-point amplitude amplification {\cite[Theorem 27]{gilyen2019singular}}]
    \label{thm:amplitude_amplification}
  Let $\mathcal{A}$ be a quantum algorithm on space $\mathcal{H}_{\mathcal{A}} = \mathcal{H}_F\otimes \mathcal{H}_W$ such that
    \begin{align}
    \mathcal{A}|\mathbf{0}\rangle_{\mathcal{H}_{\mathcal{A}}}=&\sqrt{p_{\mathrm{succ}}}|1\rangle_{\mathcal{H}_F}|\phi\rangle_{\mathcal{H}_W}\nonumber\\&+\sqrt{1-p_{\mathrm{succ}}}|0\rangle_{\mathcal{H}_F}|\psi\rangle_{\mathcal{H}_W},
    \end{align}
    where $\||\phi\rangle\| = 1$.

    For any $0 < \delta < 1,0 < \epsilon < 1$, there is a quantum algorithm $\mathcal{A}'$ using a single ancilla qubit and
    $\mathcal{O}(\frac{\log(1/\epsilon)}{\delta})$ calls to $\mathcal{A}$ and $\mathcal{A}^{\dagger}$, such that $\|\mathcal{A}'|0\rangle_{\mathcal{H}_{\mathcal{A}}}-|1\rangle_{\mathcal{H}_F}|\phi\rangle_{\mathcal{H}_W}\|\le \epsilon$ as long as $\sqrt{p_{\mathrm{succ}}} > \delta$. 
\end{theorem}

\vspace{3mm}
\noindent
\textbf{Amplitude estimation.}
Classically, if we like to estimate the expectation of the Bernoulli random variable $X$ to within additive error $\epsilon$, we need $\Theta(1/\epsilon^2)$ i.i.d. samples of $X$. Given access to $U$ in \eq{quantum_bernoulli}, we can also estimate $p$ with a quadratic quantum speedup:
\begin{theorem} [Amplitude estimation {\cite[Theorem 12]{brassard2002quantum}}]
  \label{thm:amplitude_estimation}
    Let $\mathcal{A}$ be a quantum algorithm on space $\mathcal{H}_{\mathcal{A}} = \mathcal{H}_F\otimes \mathcal{H}_W$ such that
    \begin{align}
    \mathcal{A}|\mathbf{0}\rangle_{\mathcal{H}_{\mathcal{A}}}=&\sqrt{p_{\mathrm{succ}}}|1\rangle_{\mathcal{H}_F}|\phi\rangle_{\mathcal{H}_W}\nonumber\\&+\sqrt{1-p_{\mathrm{succ}}}|0\rangle_{\mathcal{H}_F}|\psi\rangle_{\mathcal{H}_W},
    \end{align}
    where $\||\phi\rangle \|=1$, the amplitude estimation algorithm outputs a $\tilde{p}_{\mathrm{succ}} \in[0,1]$ satisfying
  \begin{align}
    |\tilde{p}_{\mathrm{succ}}-p_{\mathrm{succ}}| \leq \frac{2 \pi \sqrt{p_{\mathrm{succ}}(1-p_{\mathrm{succ}})}}{M}+\frac{\pi^{2}}{M^{2}}
  \end{align}
  with success probability at least $8 / \pi^{2}$, using $M$ calls to $\mathcal{A}$ and $\mathcal{A}^{\dagger}$.
\end{theorem}

In application, we often need to estimate $p_{\text{succ}}$ to within multiplicative error $\epsilon$. Then we can set
\begin{align}
\label{eq:M_mul}
M = \frac{3\pi}{\epsilon\sqrt{p_{\text{succ}}}}
\end{align}
in \thm{amplitude_estimation} such that
\begin{equation}
  |\tilde{p}_{\text{succ}}-p_{\text{succ}}| \le \frac{2}{3}\epsilon p_{\mathrm{succ}} \sqrt{1-p_{\mathrm{succ}}} + \frac{1}{9}\epsilon^2 p_{\mathrm{succ}} \le \epsilon p_{\mathrm{succ}} .
\end{equation}


\subsection{Projected unitary encoding}
To manipulate general matrices $A$ by quantum circuits, we need a tool called \textit{projected unitary encoding} introduced by \cite{gilyen2019singular}. We say that a unitary $U$ and two orthogonal projections $\Pi, \widetilde{\Pi}$ form a projected unitary encoding of a matrix $A$ if $A=\widetilde{\Pi}U\Pi$.

An important special projected unitary encoding is the \textit{block-encoding} where $\widetilde{\Pi} = \Pi = |
0^k\rangle \langle 0^k 
|\otimes I$. In this case, {all nonzero entries of} $A$ only appears in the $2^k\times 2^k$ top-left corner of $U$. Sometimes the convention also refers to this corner as $A$, and call a unitary $U$ a block-encoding of $A$ if 
\begin{align}
&U=\left[\begin{array}{cc}
A  & \cdot \\
\cdot & \cdot
\end{array}\right] \nonumber ,
\end{align}
denoted by $A=(\langle \mathbf{0}| \otimes I) U(|\mathbf{0}\rangle \otimes I)$.

Here we list some useful projected unitary encoding and block-encoding from previous work and used in ours.

\begin{itemize}
    \item For $U_{\mathrm{pure}}$ in \defn{pure-state-preparation}, take $\widetilde{\Pi }=\sum_{i=1}^{n} |i\rangle\langle i| \otimes|i\rangle\langle i|$, $\Pi=|\mathbf{0}\rangle\langle \mathbf{0}|\otimes I$, and $U = U_{\mathrm{pure}} \otimes I$, then we have 
    \begin{align}
      \label{eq:block-1}
        \widetilde{\Pi}U\Pi =\sum_{i=1}^{n}\sqrt{p_{i}}|i\rangle \langle \mathbf{0}|\otimes |i\rangle \langle i|.
    \end{align}
    \item For $U_p$ in \defn{purified-query-access}, take $\widetilde{\Pi}=\sum_{i=1}^{n} I \otimes|i\rangle\langle i|\otimes| i\rangle\langle i|$, $\Pi=|\mathbf{0}\rangle\langle \mathbf{0}|\otimes| \mathbf{0}\rangle\langle \mathbf{0}| \otimes I$, and $U=U_{p} \otimes I$, then we have
    \begin{align}
      \label{eq:block-2}
     \widetilde{\Pi} U \Pi=\sum_{i=1}^{n} \sqrt{p_{i}}\left|\phi_{i}\right\rangle\langle \mathbf{0}|\otimes| i\rangle\langle \mathbf{0}|\otimes| i\rangle\langle i|.
    \end{align}
\item Let $U_\rho$ be the oracle in \defn{purified-query-access} which satisfies $U_{\rho}|\mathbf{0}\rangle_A|\mathbf{0}\rangle_B = \sum_{i=1}^n\sqrt{p_i} |\phi_i\rangle_A|\psi_i\rangle_B$.
    Let $W$ be a unitary that maps $|\mathbf{0}\rangle|\mathbf{0}\rangle$ to $\sum_{j=1}^n \frac{|j\rangle|j\rangle}{\sqrt{n}}$ and $|\phi_{j}^*\rangle $ be the conjugate of $|\phi_j\rangle$. Take $\widetilde{\Pi}=I \otimes|\mathbf{0}\rangle\langle \mathbf{0}|\otimes| \mathbf{0}\rangle\langle \mathbf{0}|$, $\Pi=|\mathbf{0}\rangle\langle \mathbf{0}|\otimes| \mathbf{0}\rangle\langle \mathbf{0}| \otimes I$, and $U=\left(I \otimes U_{\rho}^{\dagger}\right)\left(W\otimes I\right)$, then we have
    \begin{align}
      \label{eq:block-3}
      \widetilde{\Pi}U\Pi=\sum_{i=1}^n \sqrt{\frac{p_i}{n}}\left|\phi_i^{*}\right\rangle\langle \mathbf{0}|\otimes| \mathbf{0}\rangle\langle \mathbf{0}|\otimes| \mathbf{0}\rangle\left\langle\psi_i\right|.
    \end{align} 
    \item Let $A,B,C$ be three $\lceil \log n\rceil$-qubit registers. For $U_\rho$ in \defn{purified-query-access}, let $S$ be the swap operator, and $U= (U_{\rho}^{\dagger}\otimes I_C)(I_A\otimes S_{B,C})(U_{\rho}\otimes I_C)$, then we have
    \begin{align}
      \label{eq:block-4}
      (\langle\mathbf{0}|_{A,B}\otimes I_C)U(|\mathbf{0}\rangle_{A,B}\otimes I_C) &= \sum_{i=1}^n p_i|\psi_i\rangle\langle \psi_i|_C \nonumber\\&= \rho.
    \end{align}
\end{itemize}

The first three projected unitary encodings are proposed by \cite{gilyen2019distributional} and the last one is proposed by \cite{low2019hamiltonian} in its Lemma 7.


\subsection{Quantum singular value transformation}

In \cite{gilyen2019singular}, a general quantum algorithm framework called \textit{quantum singular value transformation (QSVT)} is proposed, which is useful in many computational tasks including property estimation. Before introducing this framework, we first give the definition of singular value transformation.
\begin{definition}[Singular value transformation {\cite[Definition 16]{gilyen2019singular}}]
   Let $f\colon \mathbb{R} \rightarrow \mathbb{C}$ be an even or odd function. Suppose that $A \in \mathbb{C}^{\tilde{d} \times d}$ has the following singular value decomposition
  \begin{align}
    A=\sum_{i=1}^{d_{\min }} \sigma_{i}|\tilde{\psi}_{i}\rangle\left\langle\psi_{i}\right|,
  \end{align}
  where $d_{\min }:=\min (d, \tilde{d})$. For the function $f$ we define the singular value transform of $A$ as
  \begin{align}f^{(S V)}(A):= \begin{cases}\sum_{i=1}^{d_{\min }} f\left(\sigma_{i}\right)|\tilde{\psi}_{i}\rangle\left\langle\psi_{i}\right| & \text { if } f \text { is odd, and } \\ \sum_{i=1}^{d} f\left(\sigma_{i}\right)\left|\psi_{i}\right\rangle\left\langle\psi_{i}\right| & \text { if } f \text { is even}, \end{cases}\end{align}
  where for $i \in[d] \backslash\left[d_{\min }\right]$ we define $\sigma_{i}:=0$.
\end{definition}

Given a matrix $A$ block-encoded in a unitary, polynomial singular value transformation of $A$ can be efficiently implemented as follows:
\begin{theorem}[{\cite[Corollary 18]{gilyen2019singular}}] Let $\mathcal{H}_{U}$ be a finite-dimensional Hilbert space and let $U, \Pi, \widetilde{\Pi} \in \operatorname{End}\left(\mathcal{H}_{U}\right)$ be linear operators on $\mathcal{H}_{U}$ such that $U$ is a unitary, and $\widetilde{\Pi}, \Pi$ are orthogonal projectors. Suppose that $P=\sum_{k=0}^{n} a_{k} x^{k} \in \mathbb{R}[x]$ is a degree-n polynomial such that
  \begin{align}
  \label{eq:qsvt_condition}
    &a_{k} \neq 0 \text{ only if }k \equiv n \bmod 2\text{, and }\nonumber \\&\text{for all }x \in[-1,1]:|P(x)| \leq 1. 
  \end{align}
  Then there exists a vector $\Phi  = (\phi_1, \phi_2, \ldots, \phi_n)  \in \mathbb{R}^{n}$, such that
  \begin{align}
    &P^{(S V)}\big(\widetilde{\Pi} U \Pi\big) \nonumber \\
    = & \begin{cases}
    \big(\langle+| \otimes \widetilde{\Pi}\big) U^{(SV)}_P \big(|+\rangle \otimes \Pi\big) & \text { if } n \text { is odd,} \\
    \big(\langle+| \otimes \Pi\big) U^{(SV)}_P \big(|+\rangle \otimes \Pi\big) & \text { if } n \text { is even, }\end{cases}
  \end{align}
  where $U^{(SV)}_P := |0\rangle\langle 0|\otimes U_{\Phi}+| 1\rangle\langle 1| \otimes U_{-\Phi}$ with
  \begin{equation}
  U_{\Phi}:=\begin{cases}
    e^{i \phi_{1}(2 \widetilde{\Pi}-I)} U \prod_{j=1}^{(n-1) / 2}(e^{i \phi_{2 j}(2 \Pi-I)} U^{\dagger}\\ \cdot e^{i \phi_{2 j+1}(2 \widetilde{\Pi}-I)} U) \qquad \text {if } n \text { is odd,} \\
    \prod_{j=1}^{n / 2}(e^{i \phi_{2 j-1}(2 \Pi-I)} U^{\dagger} \\\cdot e^{i \phi_{2 j}(2 \widetilde{\Pi}-I)} U) \qquad \  \text { if } n \text { is even. }
    \end{cases}
    \end{equation} 
  \label{thm:qsvt}
\end{theorem}
Note that $P^{(S V)}\big(\widetilde{\Pi} U \Pi\big)$ acts on the same space as $\widetilde{\Pi} U \Pi$, while $U_P^{(SV)}$ acts on an enlarged space with one ancillary qubit added. This theorem tells us that for a polynomial $P$ of degree $d$ which satisfies \eq{qsvt_condition} in \thm{qsvt}, we can implement $P^{(SV)}(\widetilde{\Pi}U{\Pi})$ with $\Theta(d)$ uses of $U, U^{\dagger}$ and controlled reflections $I-2\Pi, I-2\widetilde{\Pi}$.

\subsection{Variable-stopping-time algorithms}
\noindent
\textbf{Variable-stopping-time quantum algorithm.}
\label{sec:vtaa} 
In \cite{ambainis2012variable}, \textit{variable-stopping-time quantum algorithms} are proposed to characterize those having different branches of computations stopping at different time. We follow the definition in \cite{childs2017} and \cite{chakraborty2019linear}.
\begin{definition}[Variable-stopping-time quantum algorithm {\cite[Definition 13]{chakraborty2019linear}}]
    We say that $\mathcal{A}=\mathcal{A}_{m} \cdots \mathcal{A}_{1}$ is a variable-stopping-time quantum algorithm if $\mathcal{A}$ acts on $\mathcal{H}=\mathcal{H}_{C} \otimes \mathcal{H}_{\mathcal{A}}$, where $\mathcal{H}_{C}=\otimes_{i=1}^{m} \mathcal{H}_{C_{i}}$ with $\mathcal{H}_{C_{i}}=\operatorname{Span}(|0\rangle,|1\rangle)$, and each unitary $\mathcal{A}_{j}$ acts on $\mathcal{H}_{C_{j}} \otimes \mathcal{H}_{\mathcal{A}}$ controlled on the first $j-1$ qubits being $|0\rangle^{\otimes (j-1)} \in \otimes_{i=1}^{j-1} \mathcal{H}_{C_{i}}$.
\end{definition}
The algorithm $\mathcal{A}$ is divided into $m$ stages $\mathcal{A}_{1},\ldots  ,\mathcal{A}_{m}$ according to the $m$ possible stopping times $t_1,\ldots,t_m$. In property estimation problem, we focus more on query complexity, so the $t_j$ we refer to is the query complexity of $\mathcal{A}_j \cdots \mathcal{A}_{1}$.

In any stage $j$, the unitary $\mathcal{A}_j$ can set the state in $\mathcal{H}_{C_j}$ to $|1\rangle$. This indicates that the computation has stopped on this branch, since any $\mathcal{A}_k,k>j$ is controlled on the state in $\otimes_{i=1}^{k-1} \mathcal{H}_{C_{i}}$ being all-0 state and does not alter the state on this branch since the state in $\mathcal{H}_{C_j}$ is set to $|1\rangle$.

In order to analyze $\mathcal{A}$, we give the definition of the probability of the algorithm stopping by time $t_j$ as follows:
\begin{definition}[Probability of stopping by time $t_j$ {\cite[Definition 14]{chakraborty2019linear}}]
    We define the orthogonal projector
    \begin{align}
    \Pi_{\mathrm{stop} \leq t}:=\sum_{j: t_{j} \leq t}|1 \rangle\langle 1 |_{C_{j}} \otimes I_{\mathcal{H}_{\mathcal{A}}},
    \end{align}
    where by $|1\rangle\langle 1|_{C_{j}}$ we denote the orthogonal projector on $\mathcal{H}_{C}$ which projects onto the state
    \begin{align}
    &|0\rangle_{\mathcal{H}_{C_{1}}} \otimes \cdots \otimes|0\rangle_{\mathcal{H}_{C_{j-1}}} \otimes|1\rangle_{\mathcal{H}_{C_{j}}} \otimes|0\rangle_{\mathcal{H}_{C_{j+1}}} \nonumber \\&\otimes \cdots \otimes|0\rangle_{\mathcal{H}_{C_{m}}}.
    \end{align}
    Then we define $p_{\mathrm{stop} \leq t}:=\| \Pi_{\mathrm{stop} \leq t} \mathcal{A}|\mathbf{0}\rangle \|^{2}$, and similarly $p_{\mathrm{stop} \geq t}$ and $p_{\mathrm{stop}=t_j}$.
\end{definition}
It is also worth mentioning that in our applications, it always holds that $p_{\mathrm{stop}\le t_m} = 1$. Let $ p_{\mathrm{stop}=t_j} := p_{\mathrm{stop} \leq t_j}-p_{\mathrm{stop} \leq t_{j-1}}$. We define the average complexity of $\mathcal{A}$ in a way similar to that in \cite{ambainis2012variable} by
\begin{align}
    \label{eq:def_avg}
    T_{\text{avg}} := \sqrt{\sum_{j=1}^m t_j^2\cdot p_{\text{stop}=t_j}},
\end{align}
and the maximum complexity of $\mathcal{A}$
\begin{align}
    \label{eq:def_max}
    T_{\max} := t_m.
\end{align}

\vspace{3mm}
\noindent
\textbf{Variable-time amplitude amplification and estimation.}
Suppose we have access to a variable-stopping-time quantum algorithm $\mathcal{A}$ acting on $\mathcal{H} = \mathcal{H}_C\otimes \mathcal{H}_{\mathcal{A}} $ such that
\begin{align}
    \mathcal{A}|\mathbf{0}\rangle=&\sqrt{p_{\text{succ}}}|1\rangle_{\mathcal{H}_F}|\phi\rangle_{\mathcal{H}_W,\mathcal{H}_C}\nonumber\\&+\sqrt{1-p_{\text{succ}}}|0\rangle_{\mathcal{H}_F}|\psi\rangle_{\mathcal{H}_W,\mathcal{H}_C},
\end{align}
where $\||\phi\rangle \| = \||\psi\rangle \|= 1$, $\mathcal{H}_{\mathcal{A}} = \mathcal{H}_{F}\otimes \mathcal{H}_{W}$, and $\mathcal{H}_F = \spn(|0\rangle,|1\rangle)$ indicates ``good" and ``bad" outcomes. If we want to obtain the ``good" outcome $|\phi\rangle$, we can use amplitude amplification algorithm in \thm{amplitude_amplification} with $\Theta(\frac{1}{\sqrt{p_{\text{succ}}}})$ calls to $\mathcal{A}$ and $\mathcal{A}^{\dagger}$, so the total complexity is $\frac{T_{\max}}{\sqrt{p_{\text{succ}}}}$. However, we can do better for variable-stopping-time algorithm. In \cite{ambainis2012variable}, the following variable-time amplitude amplification algorithm is proposed with lower complexity:
\begin{theorem}[Variable-time amplitude amplification (VTAA) {\cite[Theorem 1]{ambainis2012variable}}]
      Given a variable-stopping-time quantum algorithm $\mathcal{A}$ acting on $\mathcal{H} = \mathcal{H}_C\otimes \mathcal{H}_{\mathcal{A}} $ such that
    \begin{align}
        \mathcal{A}|\mathbf{0}\rangle=&\sqrt{p_{\mathrm{succ}}}|1\rangle_{\mathcal{H}_F}|\phi\rangle_{\mathcal{H}_W,\mathcal{H}_C}\nonumber \\&+\sqrt{1-p_{\mathrm{succ}}}|0\rangle_{\mathcal{H}_F}|\psi\rangle_{\mathcal{H}_W,\mathcal{H}_C},
    \end{align}
    where $\||\phi\rangle \|=1$, $\mathcal{H}_{\mathcal{A}} = \mathcal{H}_{F}\otimes \mathcal{H}_{W}$ and $\mathcal{H}_F = \spn(|0\rangle,|1\rangle)$. Let $T_{\mathrm{avg}}$ and $T_{\max}$ be the parameters defined in \eq{def_avg} and \eq{def_max}. We can construct a quantum algorithm $\mathcal{A}^{\prime}$ invoking $\mathcal{A}$ several times, for total time
    \begin{align}
    \mathcal{O}\left(T_{\max } \sqrt{\log T_{\max }}+\frac{T_{\mathrm{avg}}}{\sqrt{p_{\mathrm{succ }}}} \log ^{1.5} T_{\max }\right)
    \end{align}
    that produces a state $\alpha|1\rangle_{\mathcal{H}_F}|\phi\rangle_{\mathcal{H}_W,\mathcal{H}_C}+\beta|0\rangle_{\mathcal{H}_F}|\psi\rangle_{\mathcal{H}_W,\mathcal{H}_C}$ with probability $|\alpha|^{2} \geq 1 / 2$ as the output.
\end{theorem}

In \cite{chakraborty2019linear}, standard amplitude estimation algorithm in \thm{amplitude_estimation} was also generalized to variable-time scenarios. \textcolor{black}{Our algorithms only need the following VTAE algorithm, but we include VTAA for completeness since VTAE is built upon VTAA (just similar to that Amplitude Estimation is built upon Amplitude Amplification).}

\begin{theorem}[Variable-time amplitude estimation (VTAE) {\cite[Theorem 23]{chakraborty2019linear}}]
\label{thm:vtae}
  Let $\mathcal{A}$ be a variable-stopping-time quantum algorithm acting on $\mathcal{H} = \mathcal{H}_C\otimes \mathcal{H}_{\mathcal{A}} $ such that
    \begin{align}
        \mathcal{A}|\mathbf{0}\rangle=&\sqrt{p_{\mathrm{succ}}}|1\rangle_{\mathcal{H}_F}|\phi\rangle_{\mathcal{H}_W,\mathcal{H}_C}\nonumber \\ &+\sqrt{1-p_{\mathrm{succ}}}|0\rangle_{\mathcal{H}_F}|\psi\rangle_{\mathcal{H}_W,\mathcal{H}_C},
    \end{align}
    where $\||\phi\rangle \|=1$, $\mathcal{H}_{\mathcal{A}} = \mathcal{H}_{F}\otimes \mathcal{H}_{W}$, and $\mathcal{H}_F = \spn(|0\rangle,|1\rangle)$. Let $T_{\mathrm{avg}},T_{\max}$ be the parameters defined in \eq{def_avg} and \eq{def_max}, respectively,  $T_{\max}' := 2T_{\max}/t_1$ and
    \begin{align}
        Q = T_{\max}\sqrt{\log(T_{\max}')} + \frac{T_{\mathrm{avg}}\log(T_{\max}')}{\sqrt{p_{\mathrm{succ}}}}.
    \end{align}
    Suppose that we know a lower bound $p_{\mathrm{succ}}'$  of $p_{\mathrm{succ}}$. Then for any $\epsilon,\delta \in (0,1)$, we can estimate $p_{\mathrm{succ}}$ to within multiplicative error $\epsilon$ and success probability at least $1-\delta$ with complexity
    \begin{align}
        \mathcal{O}\biggl(&\frac{Q}{\epsilon}\log^2(T_{\max}')\log(\frac{\log(T_{\max}')}{\delta}) \nonumber \biggr. \\ &\left.+ Q\log(T_{\max}')\log(\frac{1}{\delta}\log(\frac{T_{\max}'}{p_{\mathrm{succ}}'}))\right) = \widetilde{\mathcal{O}}(Q/\epsilon).
    \end{align}
\end{theorem}

Note that the total complexity of estimating $p_{\mathrm{succ}}$ to within multiplicative error $\epsilon$ using the standard amplitude estimation algorithm is $MT_{\max} = \mathcal{O}(\frac{T_{\max}}{\epsilon\sqrt{p_{\mathrm{succ}}}})$, where $M$ is determined in \eq{M_mul}. Therefore, if $T_{\mathrm{avg}}$ is much smaller than $T_{\max}$, we can achieve a great acceleration by replacing the standard amplitude estimation algorithm with variable-time amplitude estimation.

\section{Main Algorithm}\label{sec:main-algorithm}
\label{sec:main_algo}


\subsection{Estimating properties of probability distributions by QSVT}
\label{sec:main_algo_standard}
\noindent
In this section, we introduce a quantum algorithm for estimating properties of a probability distribution $\mathbf{p} = (p_i)_{i=1}^n$ on $[n]$ to within a certain error. Our algorithm is based on QSVT and amplitude estimation, which is similar to the entropy estimation algorithm in \cite{gilyen2019distributional}. In this paper, we mostly focus on the pure state preparation oracle in \defn{pure-state-preparation}, and we will show in \sec{applications} that our algorithm also works well with purified quantum query-access oracle in \defn{purified-query-access}. Here we give a brief explanation. Compared to the pure-state preparation oracle, the purified quantum query-access oracle just adds an unknown state in the right-hand side of \eq{purified-query-access}. However, in our algorithm, we produce a quantum state $|\psi\rangle$ such that the module square of the amplitude of the projection of $|\psi\rangle$ onto a subspace, $\|\Pi|\psi\rangle\|^2$, encodes the quantity we want to estimate, where the projector $\Pi$ acts as an identity in the added space. As a result, the module square $\|\Pi |\psi\rangle \|^2$ is independent of the unknown added state. 


In this section, we use the block-encoding in \eq{block-1} to encode the probability distribution $\mathbf{p} = (p_i)_{i=1}^n$ on $[n]$ and denote it by 
\begin{equation}
  \label{eq:def_E}
E := \widetilde{\Pi } U \Pi=\sum_{i=1}^{n}\sqrt{p_{i}}|i\rangle \langle \mathbf{0}|\otimes |i\rangle \langle i|,
\end{equation}
which has singular values $\sqrt{p_{1}}, \ldots, \sqrt{p_{n}}$.

Suppose $S$ is a polynomial satisfying \eq{qsvt_condition} in \thm{qsvt}.
We propose an algorithm to estimate $\sum_{i=1}^n p_i S(\sqrt{p_i})^2$ to within a certain multiplicative error.

Before describe out main algorithm, we first give a rough version of the amplitude estimation algorithm and it only needs a lower bound of the module of the amplitude.
\begin{lemma}
\label{lem:rough_amplitude_estimation}
Let $\mathcal{A}$ be a quantum algorithm on space $\mathcal{H} = \mathcal{H}_F\otimes \mathcal{H}_W$ such that
\begin{align}
\mathcal{A}|\mathbf{0}\rangle=\sqrt{p_{\mathrm{succ}}}|1\rangle_{\mathcal{H}_F}|\phi\rangle_{\mathcal{H}_W}\nonumber+\sqrt{1-p_{\mathrm{succ}}}|0\rangle_{\mathcal{H}_F}|\psi\rangle_{\mathcal{H}_W},
 \end{align} where $\||\phi\rangle \|=1$.
Given $L$ such that $L \le p_{\mathrm{succ}}$, there exists an algorithm which outputs an estimate $\tilde{p}\in[0,1]$ satisfying $\tilde{p}/p_{\mathrm{succ}}\in[\frac{1}{2},2]$ with success probability at least $1-\delta$, using $\mathcal{O}\left(\frac{1}{\sqrt{p_{\mathrm{succ}}}}\log(\frac{1}{\sqrt{p_{\mathrm{succ}}}})\log(\frac{\log(\frac{1}{L})}{\delta})\right) = \widetilde{\mathcal{O}}\left(\frac{1}{\sqrt{p_{\mathrm{succ}}}}\right)$ calls to $\mathcal{A} $ and $\mathcal{A}^{\dagger}$. 
\end{lemma}
The proof of \lem{rough_amplitude_estimation} is deferred to \append{other_proof}. With this estimate at hand, we now describe our main algorithm in \algo{main_algo}, and prove its correctness in \lem{main_algo}. 
\begin{algorithm*}[ht]
  \KwInput{Multiplicative error $\epsilon$, and quantum registers $(A,B,Q,F)$ initialized to $|\mathbf{0}\rangle$.}
  \KwOutput{$\tilde{p}$, an estimation of the value $\sum_{i=1}^n p_i S(\sqrt{p_i})^2$.}
  \KwParameter{$(L,S)$, where $L$ is a lower bound of $\sum_{i=1}^n p_i S(\sqrt{p_i})^2$, and $S$ is an odd or even polynomial.} 
  \KwNotation{
    \begin{itemize}[leftmargin=*]
    \setlength{\itemsep}{0pt}
    \setlength{\parsep}{0pt}
    \setlength{\parskip}{0pt}
        \item $A,B$ are two $\lceil \log n\rceil$-qubit input registers;
        \item $Q$ is a single-qubit register, to be used as an ancilla register for QSVT;
        \item $F$ is a single-qubit flag register indicating ``good" components whose amplitude we estimate;
        \item $U^{(SV)}_S$ is the unitary $U^{(SV)}_P$ in \thm{qsvt} when the transformation polynomial $P:=S$, and $U,\widetilde{\Pi},\Pi$ in \thm{qsvt} are set to $U,\widetilde{\Pi},\Pi$ in \eq{block-1}, \eq{block-1} and \eq{def_E}.
    \end{itemize}
}
  \SetKwProg{myproc}{Regard the following subroutine as $\mathcal{A}$:}{}{}
 \LinesNumberedHidden{
    \nonl\myproc{}{
      Apply $U_{\mathrm{pure}}$ to register $B$ and Hadamard gate $H$ to register $Q$ \label{lin:sample_from_p}\;
      Apply unitary $U^{(SV)}_S$ to registers $(A,B,Q)$ using $Q$ as the ancilla register \label{lin:state_after_U_S}\;
      \If{$S$ is an odd function}{Apply unitary $\text{C}_{\widetilde{\Pi }\otimes|+\rangle\langle+|}\text{NOT} = \widetilde{\Pi }\otimes |+\rangle\langle+|_Q\otimes X_F + (I-\widetilde{\Pi }\otimes |+\rangle\langle+|_Q)\otimes I_F$ to registers $(A,B,Q,F)$\;}
      \Else{Apply unitary $\text{C}_{\Pi\otimes|+\rangle\langle+|}\text{NOT} = \Pi \otimes |+\rangle\langle+|_Q\otimes X_F + (I-\Pi \otimes |+\rangle\langle+|_Q)\otimes I_F$ to registers $(A,B,Q,F)$\label{lin:state_after_cnot}\;}

    }
  }
    Run the algorithm in \lem{rough_amplitude_estimation} with $\delta = \frac{1}{8}$ and lower bound $L$ to obtain a rough estimate of $\|(I\otimes|1\rangle\langle1|_F)|\mathcal{A}|\mathbf{0}\rangle \|^2$. Denote this estimate by $P$\label{lin:rough_estimate}\;
    Use the amplitude estimation algorithm in \thm{amplitude_estimation} with $M = \lceil \frac{5\pi}{\sqrt{P}\epsilon}\rceil$ to estimate  $\|(I\otimes|1\rangle\langle1|_F)|\mathcal{A}|\mathbf{0}\rangle \|^2$ and denote the result by $\tilde{p}$. Output $\tilde{p}$\;
\caption{Estimate the value $\sum_{i=1}^n p_i S(\sqrt{p_i})^2$ for $\mathbf{p} =(p_i)_{i=1}^n$. }
  \label{algo:main_algo}
\end{algorithm*}

\begin{lemma}
\label{lem:main_algo}
  Suppose $S$ is a polynomial satisfying \eq{qsvt_condition} in \thm{qsvt}, $\mathbf{p} = (p_i)_{i=1}^n$ is a probability distribution on $[n]$ and we know a lower bound $L \le \sum_{i=1}^n p_i S(\sqrt{p_i})^2$. Then for any $\epsilon\in(0,1)$, \algo{main_algo} with parameters $L:=L,S:=S$ and input $\epsilon$ outputs an estimate of $\sum_{i=1}^n p_i S(\sqrt{p_i})^2$ to within multiplicative error $\epsilon$ with success probability at least $\frac{2}{3}$ using $\widetilde{\mathcal{O}}\left(\frac{1}{\epsilon}\frac{\deg(S)}{\sqrt{\sum_{i=1}^n p_i S(\sqrt{p_i})^2}}\right)$ calls to $U_{\mathrm{pure}}$ and $U_{\mathrm{pure}}^\dagger$.
\end{lemma}
\begin{proof}
    We first analyze the error of \algo{main_algo}.
    Let the state in registers $(A,B,Q)$ after \lin{state_after_U_S} of \algo{main_algo} be $|\Phi_p\rangle$. Assuming $S$ is an odd function for the moment for simplicity, the output $|\Phi_{p}\rangle$ satisfies that

  \begin{align}
    \label{eq:cal_qsvt}
    &(\widetilde{\Pi } \otimes \langle+|_Q)|\Phi_{p}\rangle\nonumber\\
    =& (\widetilde{\Pi } \otimes \langle+|_Q)U^{(SV)}_S(|\mathbf{0}\rangle_A|\Psi_{\mathbf{p}}\rangle_B|+\rangle_Q)\nonumber\\
    =& (\widetilde{\Pi } \otimes \langle+|_Q)(|0\rangle\langle 0|_Q\otimes U_{\Phi_S}+| 1\rangle\langle 1|_Q \otimes U_{-\Phi_S})\nonumber \\ &\cdot(|\mathbf{0}\rangle_A|\Psi_{\mathbf{p}}\rangle_B |+\rangle_Q)\nonumber\\
    =& \bigl((\widetilde{\Pi } \otimes \langle+|_Q)(|0\rangle\langle 0|_Q\otimes U_{\Phi_S}+| 1\rangle\langle 1|_Q \otimes U_{-\Phi_S})\nonumber \\ &\cdot(\Pi\otimes |+\rangle_Q)\bigr)(|\mathbf{0}\rangle_A|\Psi_{\mathbf{p}}\rangle_B )\nonumber\\
     =& S^{(SV)}(\tilde \Pi U \Pi)\Big(\sum_{i=1}^{n} \sqrt{p_{i}}|\mathbf{0}\rangle_A|i\rangle_B\Big) \nonumber\\
     =& \Big(\sum_{i=1}^{n}S(\sqrt{p_i}) |i\rangle \langle \mathbf{0}|_A\otimes |i\rangle \langle i|_B\Big)\Big(\sum_{i=1}^{n} \sqrt{p_{i}}|\mathbf{0}\rangle_A|i\rangle_B\Big)\nonumber \\
     =& \sum_{i=1}^{n} \sqrt{p_{i}} S\left(\sqrt{p_{i}}\right)|i\rangle_A|i\rangle_B,
  \end{align}
where the third equation comes from $(\Pi\otimes | + \rangle_Q)(|\mathbf{0}\rangle_A|\Psi_{\mathbf{p}}\rangle_B ) = |\mathbf{0}\rangle_A|\Psi_{\mathbf{p}}\rangle_B |+\rangle_Q$. 

Let the state in registers $(A,B,Q,F)$ after \lin{state_after_cnot} be $|\Psi_{\mathbf{p}}\rangle$. Then we have
\begin{align}
    \label{eq:ctrl_x}
    |\Psi_{p}\rangle
    & = \text{C}_{\widetilde{\Pi }\otimes|+\rangle\langle+|}\text{NOT}|\Phi_{p}\rangle|0\rangle_F \nonumber\\
    & = (\widetilde{\Pi }\otimes |+\rangle\langle+|_Q\otimes X_F)|\Phi_{p}\rangle|0\rangle_F \nonumber \\ &\quad+ ((I-\widetilde{\Pi }\otimes |+\rangle\langle+|_Q)\otimes I_F)|\Phi_{p}\rangle|0\rangle_F \nonumber\\
    & = \sum_{i=1}^{n} \sqrt{p_{i}} S\left(\sqrt{p_{i}}\right)|i\rangle_A|i\rangle_B|+\rangle_Q|1\rangle_F + |\psi_{\mathrm{garbage}}\rangle|0\rangle_F,
\end{align}
where the third equation comes from \eq{cal_qsvt} and $|\psi_{\mathrm{garbage}}\rangle$ is an unnormalized state which we do not care about.

Although we suppose $S$ to be an odd function, it is easily verified that \eq{cal_qsvt} and \eq{ctrl_x} hold for all even functions $S$ as well, if we replace $\widetilde{\Pi}$ with $\Pi$.

Note that if we measure register $F$ after the subroutine $\mathcal A$, the success probability is
\begin{align}
  \begin{aligned}
    &\|(I\otimes|1\rangle\langle1|_F)|\Psi_{p}\rangle \|^2\nonumber \\
=& \|\sum_{i=1}^{n} \sqrt{p_{i}} S\left(\sqrt{p_{i}}\right)|i\rangle_A|i\rangle_B|+\rangle_Q|1\rangle_F\|^2\nonumber \\ =& \sum_{i=1}^n p_iS(\sqrt{p_i})^2,
  \end{aligned}
\end{align}
the quantity we would like to estimate.

From \lem{rough_amplitude_estimation}, we can infer that the output $P$ in \lin{rough_estimate} satisfies $P/\sum_{i=1}^n p_iS(\sqrt{p_i})^2
\in[\frac{1}{2},2]$ with probability at least $\frac{7}{8}$, so according to \thm{amplitude_estimation}, with success probability at least $\frac{8}{\pi^2}\frac{7}{8}\ge \frac{2}{3}$, the output $\tilde{p}$ satisfies 
\begin{align}
    |\tilde{p}- \sum_{i=1}^n p_iS(\sqrt{p_i})^2| &\le \frac{2}{5}\sqrt{P}\epsilon \sqrt{\sum_{i=1}^n p_iS(\sqrt{p_i})^2} + \frac{1}{25} P\epsilon^2 \nonumber \\
    &\le \Big(\frac{2\sqrt{2}}{5} + \frac{2}{25}\Big) \epsilon\sum_{i=1}^n p_iS(\sqrt{p_i})^2 \nonumber \\ &\le  \epsilon\sum_{i=1}^n p_iS(\sqrt{p_i})^2,
\end{align}
where the second inequality is because $P \le 2\sum_{i=1}^n p_i S(\sqrt{p_i})^2$. 



We now calculate the complexity of the algorithm.
In $\mathcal{A}$, we call $U_{\mathrm{pure}}$ and $U_{\mathrm{pure}}^{\dagger}$ $\widetilde{\mathcal{O}}(\deg(S))$ times. In the main algorithm, the step using \lem{rough_amplitude_estimation} calls subroutine $\mathcal{A}$  $\tilde O(1/\sqrt{\sum_i p_i S(\sqrt{p_i})^2})$ times, and the step using \thm{amplitude_estimation} calls subroutine $\mathcal{A}$  $M$ times. Overall, the query complexity of \algo{main_algo}  is \begin{align}&\widetilde{\mathcal{O}}\left(\left(M+\frac{1}{\sqrt{\sum_{i=1}^n p_iS(\sqrt{p_i})^2}}\right)\deg(S)\right) \nonumber \\ =& \widetilde{\mathcal{O}}\left(\frac{\deg(S)}{\epsilon\sqrt{\sum_{i=1}^n p_iS(\sqrt{p_i})^2}}\right),\end{align}
as claimed.
\end{proof}

This algorithm is also gate efficient. If $n$ is a power of 2, $\text{C}_{\widetilde{\Pi }\otimes|+\rangle\langle+|}\text{NOT}$ can be efficiently implemented following \fig{ctrl_x}, and it can be easily generalized to any integer $n > 0
$ using additional quantum circuit which can determine whether $|i\rangle \in \mathcal{H},\ \dim(\mathcal{H}) = 2^{\lceil\log n\rceil}$ satisfies $i\le n$. 

\begin{figure}
    \centering
    \begin{quantikz}
    &[4mm] \ctrl{1}\qwbundle{n_A} & \qw & \qw  & \qw & \ctrl{1} &\qw \\
    & \targ{} \qwbundle{n_B} & \qw  & \octrl{1} & \qw & \targ{} &\qw \\
    & \qw \qwbundle{n_Q} & \gate{H} & \octrl{1}  &\gate{H} & \qw & \qw   \\
    & \qwbundle{n_F} & \qw & \targ{} & \qw & \qw &\qw
    \end{quantikz}
    \caption{Circuit of $\text{C}_{\widetilde{\Pi }\otimes|+\rangle\langle+|}\text{NOT} = \widetilde{\Pi }\otimes |+\rangle\langle+|_Q\otimes X_F + (I-\widetilde{\Pi }\otimes |+\rangle\langle+|_Q)\otimes I_F$ if $n$ is a power of 2. \\
    $n_A,n_B,n_C,n_F$ are the sizes of the registers $A,B,C,F$, respectively.
    The CNOT gate between two registers with the same size is an abbreviation of a sequence of CNOT gates between qubits in different registers with the same index and the CNOT gate targeting a qubit conditional on a register will flip the qubit when the regisiter is an all-0/all-1 state. }
    \label{fig:ctrl_x}
\end{figure}
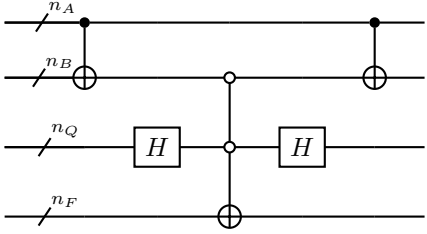

\subsection{Improvements based on annealing}
\label{sec:annealing}
\noindent
In \sec{main_algo_standard}, we have shown how to estimate the quantity 
$\sum_{i=1}^n p_i S(\sqrt{p_i})^2$ to within a certain multiplicative error, where $S$ is a polynomial satisfying \eq{qsvt_condition} in \thm{qsvt}. However, in many cases the quantity we want to estimate cannot be written in this form. Here we consider a more general quantity of discrete distributions $\mathbf{p} = (p_i)_{i=1}^n$ on $[n]$:
\begin{align}
    f(\mathbf{p}):=\sum_{i=1}^n p_i g(p_i),
\end{align}
where $g(x)$ is a monotonically increasing function on $[0,1]$ such that $g(0) = 0$, $g(1) = 1$ and $xg(x)$ is a convex function. 

An observation is that for any probability distribution $\mathbf{p}$ on $[n]$, given a lower bound $L$ of $f(\mathbf{p})$, if we can find a polynomial $S$ satisfying \eq{qsvt_condition} in \thm{qsvt} and
\begin{align}
    \Big|\sum_{i=1}^n q_i S(\sqrt{q_i})^2 - f(\mathbf{q})\Big| \le \epsilon f(\mathbf{q})
\end{align}
for any distribution $\mathbf{q}\in \Delta^n$, then by \lem{main_algo}, \algo{main_algo} can estimate $f(\mathbf{p})$ to within  multiplicative error $\epsilon$ using
\begin{align}
    \widetilde{\mathcal{O}}\left(\left(\frac{1}{\epsilon}\right)\frac{\deg(S)}{\sqrt{f(\mathbf{p})}}\right)
\end{align}
calls to $U_{\mathrm{pure}}$ and $U_{\mathrm{pure}}^{\dagger}$ for any distribution $\mathbf{p}$.

However, this algorithm can be sub-optimal in many cases. To  give an intuitive explanation, we compare this algorithm with the one in \cite{hamoudi2019Chebyshev}. In \cite{hamoudi2019Chebyshev}, they develop a quantum algorithm which can estimate the mean of a random variable $X$ to within multiplicative error $\epsilon$ using $\widetilde{\mathcal{O}}\left(\frac{\sqrt{\E[X^2]}}{\epsilon E[X]}\right)$ quantum samples. {\algo{main_algo}} is somewhat similar to this algorithm, and it can also be seen as estimating the mean of a random variable $X$ with 
$\Pr[X = g(p_i)] = p_i$. Note a catch 
that the value of $X$ in our problem is not given, and we need to estimate it ourselves. 
In \algo{main_algo}, we takes one sample from $\mathbf{p} = (p_i)_{i=1}^n$ in \lin{sample_from_p} and estimate the value of this sample by QSVT in \lin{state_after_U_S}. Therefore, for a fair comparison, we remove the cost of estimating the value of the sample in \algo{main_algo},  and
the remaining query cost is $\widetilde{\mathcal{O}}\left(\frac{1}{\epsilon \sqrt{f(\mathbf{p}})}\right) = \widetilde{\mathcal{O}}\left(\frac{1}{\epsilon \sqrt{\E[X]}}\right)$. This is higher than the cost of $\widetilde{\mathcal{O}}\left(\frac{\sqrt{\E[X^2]}}{\epsilon \E[X]}\right)$ in \cite{hamoudi2019Chebyshev} because $\E[X^2] \le \E[X]$ due to $X = g(p_i) \le g(1) = 1$.


Nevertheless, we can improve \algo{main_algo} if we are given an estimate of $f(\mathbf{p})$ to within constant multiplicative error. The idea is that we can estimate $\sum_{i=1}^n \Phi p_i g(p_i)$ for an amplification factor $\Phi \gg 1$ instead using \algo{main_algo}, for which we need to find a polynomial $S$ such that $\sum_{i=1}^n p_i S(\sqrt{p_i})^2/(\sum_{i=1}^n \Phi p_i g(p_i))\in [1-\epsilon,1+\epsilon]$. This brings the benefit that the quantity we estimate is much bigger, so the query complexity becomes smaller since it is inversely proportional to the square root of the quantity we estimate. Nevertheless, the amplification factor $\Phi$ should not be too large, since we need to guarantee the existence of the polynomial $S(x)$ which satisfies \eq{qsvt_condition} in \thm{qsvt} and is also an approximation to $\sqrt{\Phi g(x^2)}$ when $x = \sqrt{p_i}$. \thm{qsvt} requires $|S(x)|\le 1$ for all $x\in[-1,1]$, and this requires $|\Phi g(x^2)|\le 1$ for $x = \sqrt{p_i}$. Since $g(x^2)$ is a monotonically increasing function, we only need to ensure that $\Phi g(\max_{i\in[n]}p_i) \le 1$. Estimating the maximum $p_i$ is not simple, and an alternative method is to obtain an upper bound of $p_i$ from a rough estimate of $f(\mathbf{p})$. This is possible because $g(x)$ is monotonically increasing and positive, so we have $p_ig(p_i) \le f(\mathbf{p})$ and then $p_i\le (xg(x))^{-1}(f(\mathbf{p}))$ for all $i\in [n]$, where $(xg(x))^{-1}$ is the inverse function of $xg(x)$ on $[0,1]$. Detailed analysis is conducted in the following lemma.
\begin{lemma}
\label{lem:annealing_pre}
  For any convex and monotonically increasing function $g(x)$ on $[0,1]$ such that $g(0) = 0$ and $g(1) = 1$ and probability distribution $\mathbf{p} = (p_i)_{i=1}^n$ on $[n]$, let $f(\mathbf{p}):=\sum_{i=1}^n p_i g(p_i)$. Suppose that we are given constants $a,b$ and $P$ such that $af(\mathbf{p}) \le P \le b f(\mathbf{p})$ and let $p^{*}\in[0,1]$ be such that $p^{*}g(p^{*})  = \min(\frac{1}{a}P,1)$. For any $\epsilon\in(0,1)$, let $\epsilon_0 = \frac{\epsilon}{2}$. Then if we can construct a polynomial $S$ satisfying \eq{qsvt_condition} in \thm{qsvt} and
  \begin{align}
  \label{eq:annealing_pre}
      \Big|\sum_{i=1}^n q_i S(\sqrt{q_i})^2 &- c\frac{1}{g(p^*)}f(\mathbf{q})\Big| \le \frac{ca}{2b} p^{*}\epsilon_0 
  \end{align}
  for all $\mathbf{q}$ satisfying $\forall i \in[n], q_i\le p^*$, where $c> 0$ is an arbitrary positive constant, \algo{main_algo} with parameters to be $P:=\frac{P}{b}$, $S:=S$, and input $\epsilon_0$ outputs $\tilde{p}$ satisfying that $\frac{g(p^*)}{c}\tilde{p}$ is an estimate of $f(\mathbf{p})$ to within  multiplicative error $\epsilon$. This call to \algo{main_algo} uses $\widetilde{\mathcal{O}}\left(\frac{\deg(S)}{\epsilon\sqrt{p^{*}}}\right)$ calls to $U_{\mathrm{pure}}$ and $U_{\mathrm{pure}}^{\dagger}$ in \defn{pure-state-preparation}.
\end{lemma}

\begin{proof}
From \lem{main_algo}, $\tilde{p}$ is an estimate of $\sum_{i=1}^n p_i S(\sqrt{p_i})^2$ to within multiplicative error $\epsilon_0$, and this call to \algo{main_algo} uses $\widetilde{\mathcal{O}}\left(\frac{1}{\epsilon_0}\frac{\deg(S)}{\sqrt{p_i S(\sqrt{p_i})^2}}\right)$ calls to $U_{\mathrm{pure}}$ and $U_{\mathrm{pure}}^{\dagger}$. From \eq{annealing_pre}, we have
\begin{align}
\label{eq:bound_q_S_q}
\hspace{-3mm}
    \sum_{i=1}^n p_i S(\sqrt{p_i})^2 \ge& c\frac{1}{g(p^*)} f(\mathbf{p}) - \frac{ca}{2b} p^* \epsilon_0   \nonumber \\=& c\frac{p^*}{\min(\frac{1}{a}P,1)} f(\mathbf{p}) -\frac{ca}{2b} p^* \epsilon_0 \nonumber \\  \ge& \frac{ca}{b}p^*- \frac{ca}{2b} p^*\epsilon_0\ge \frac{a}{2b}cp^*
\end{align}
where the first equation comes from $p^* g(p^*) = \min(\frac{1}{a} P, 1)$ and the second inequality comes from $P \le b f(\mathbf{p})$.
Then by \lem{main_algo}, the query complexity is $\widetilde{\mathcal{O}}\left(\frac{1}{\epsilon_0}\frac{\deg(S)}{\sqrt{p_i S(\sqrt{p_i})^2}}\right) = \widetilde{\mathcal{O}}\left(\frac{1}{\epsilon} \frac{\deg(S)}{\sqrt{p^*}}\right)$.

Let $\tilde{f}(\mathbf{p}):=\frac{g(p^*)}{c}\tilde{p}$. We can infer that $\tilde{f}(\mathbf{p})$ is an approximation of $\frac{g(p^*)}{c}\sum_{i=1}^n p_i S(\sqrt{p_i})^2$ within $\epsilon_0$ multiplicative error. We now prove that $\tilde{f}(\mathbf{p})$ is an estimate of $f(\mathbf{p})$ to within multiplicative error $2\epsilon_0$.

For any $i\in [n]$, we have $p_i g(p_i) \le f(\mathbf{p}) \le \frac{1}{a} P$
and $p_i g(p_i) \le 1$, which implies that $p_i g(p_i) \le p^* g(p^*)$ for any $i\in[n]$. Since $xg(x)$ is a monotonically increasing function on $[0,1]$, we have $p_i \le p^*$ for all $i\in[n]$.
Therefore, from \eq{annealing_pre}, $\sum_{i=1}^n p_i S(\sqrt{p_i})^2$ is an approximation of $c\frac{1}{g(p^*)}f(\mathbf{p})$ within $\frac{ca}{2b}q^* \epsilon_0$ additive error. Since $\frac{c}{g(p^*)}f(\mathbf{p}) = c\frac{p^*}{\min(\frac{1}{a}P,1)}f(\mathbf{p}) \ge \frac{ca}{b}p^*$, we can infer that $\sum_{i=1}^n p_i S(\sqrt{p_i})^2$ is an approximation of $c\frac{1}{g(p^*)}f(\mathbf{p})$ within $\frac{\frac{ca}{2b}p^* \epsilon_0}{c\frac{1}{g(p^*)}f(\mathbf{p})} \le \frac{\frac{ca}{2b} p^* \epsilon_0}{\frac{ca}{b} p^*}= \frac{\epsilon_0}{2}$ multiplicative error, and because  $\tilde{f}(\mathbf{p})$ is an approximation of $\frac{g(p^*)}{c}\sum_{i=1}^n p_i S(\sqrt{p_i})^2$ within $\epsilon_0$ multiplicative error, we can infer that $\tilde{f}(\mathbf{p})$ is an approximation of $f(\mathbf{p})$ within $(1+\epsilon_0)(1+\frac{\epsilon_0}{2})-1 \le 2\epsilon_0 = \epsilon$ multiplicative error.
\end{proof}

We now show that with the additional information $P,a,b$ in \lem{annealing_pre}, the new query complexity bound of $\widetilde{\mathcal{O}}\left(\frac{\deg(S)}{\epsilon\sqrt{p^{*}}}\right)$ improves the $\frac{1}{\sqrt{\E[X]}}$ term in the bound in \lem{main_algo} to $\frac{\sqrt{\E[X^2]}}{\E[X]}$.

Since $xg(x)$ is a convex function, from Jensen's inequality, we have $f(\mathbf{p}) = \sum_{i=1}^n p_i g(p_i) \ge g(\frac{1}{n})$, and because $xg(x)$ is a monotonically increasing function, we have $p^* \ge (xg(x))^{-1}(\frac{1}{a}g(\frac{1}{n}))$. Therefore, the complexity bound in \lem{annealing_pre} becomes  $\mathcal{O}\left(\frac{\deg(S)}{\epsilon \sqrt{(xg(x))^{-1}(\frac{1}{a}g(\frac{1}{n}))}}\right)$ in the worst case and 
we prove that this bound is equivalent to the aforementioned  $\widetilde{\mathcal{O}}\left(\frac{\deg(S)\sqrt{\E[X^2]}}{\epsilon \E[X]}\right)$ bound in the following lemma, whose proof is deferred to \append{other_proof}.
\begin{lemma}
    \label{lem:bound_of_annealing}
 Let $\mathbf{p}= (p_i)_{i=1}^n$ be a probability distribution on $[n]$, and $g(x)$ be a monotonically increasing function on $[0,1]$ such that $g(0) = 1,g(1) = 1$ and $xg(x)$ is a convex function. Then we have $\max_{\mathbf{p}\in \Delta^n}\frac{\sqrt{\sum_{i=1}^n p_i g(p_i)^2}}{\sum_{i=1}^n p_i g(p_i)} = \Theta\left(\frac{1}{\sqrt{(xg(x))^{-1}(\frac{1}{a}g(\frac{1}{n}))}}\right)$ as $n\to \infty$, where $a > 0$ is  a constant satisfying the conditions in \lem{annealing_pre}. 
\end{lemma}

The remained problem is to get an estimate of $f(\mathbf{p})$ to within constant multiplicative error. We propose a framework based on annealing as follows to solve it.

\begin{proposition}
\label{prop:annealing}
Let $f(\mathbf{q})$ be any positive function on $\Delta^n$. Suppose there exists a sequence of functions $f_1(\mathbf{q}),\ldots,f_l(\mathbf{q}) = f(\mathbf{q})$ satisfying
\begin{align}
    \begin{cases}\max_{\mathbf{r} \in \Delta^n}\frac{\max_{\mathbf{q}\in \Delta^n, f_k(\mathbf{r})/f_k(\mathbf{q})\in [1-\epsilon_k,1+\epsilon_k]} f_{k+1}(\mathbf{q}) }{\min_{\mathbf{q}\in \Delta^n, f_k(\mathbf{r})/f_k(\mathbf{q})\in [1-\epsilon_k,1+\epsilon_k]} f_{k+1}(\mathbf{q})} \le c  \\ \qquad\qquad\qquad\qquad\ \text{ for } k=1,\ldots,l-1\text{, }\\
    \frac{\max_{\mathbf{q}\in \Delta^n} f_{1}(\mathbf{q}) }{\min_{\mathbf{q}\in \Delta^n} f_{1}(\mathbf{q})} \le c,\label{eq:annealing_condition}
    \end{cases}
\end{align}
for some constant $c$. If for any $k = 1,\ldots,l$, there exists a quantum algorithm $\mathcal{A}_k$ which can estimate $f_k(\mathbf{q})$ to within multiplicative error $\epsilon_k$ with success probability at least $1-\frac{\delta}{l}$ using $Q_k$ calls to $U_{\mathrm{pure}}$ and $U_{\mathrm{pure}}^{\dagger}$ given two constants $a,b$, and $P$ satisfying $a f_k(\mathbf{q})\le P\le b f_k(\mathbf{q})$, then there exists a quantum algorithm which can estimate $f(\mathbf{p})$ to within multiplicative error $\epsilon$ with success probability at least $1-\delta$ using $\sum_{k=1}^l Q_k$ calls to $U_{\mathrm{pure}}$ and $U_{\mathrm{pure}}^{\dagger}$ in \defn{pure-state-preparation}.
\end{proposition}

Intuitively, this is a framework based on annealing due to  \eq{annealing_condition}. The first condition in \eq{annealing_condition} ensures that we can get an estimate of $f_{k+1}(\mathbf{q})$ to within constant multiplicative error given a good estimate of $f_k(\mathbf{q})$ to within multiplicative error $\epsilon_k$. This rough estimate of $f_{k+1}(\mathbf{q})$ is used to construct the parameters of the next-stage algorithm estimating $f_{k+2}(\mathbf{q})$. A common construction to meet this condition is to choose $f_{k-1}$ close to $f_k$. To meet the second condition we need $f_1$ to be nearly a constant. If we consider $f_k$ as energy functions, this function sequence from $f_l$ to $f_1$ 
resembles an annealing process which slowly lowers the temperature. 

\begin{proof}[Proof of \prop{annealing}]
From~\eq{annealing_condition}, we have $\max_{\mathbf{q}\in \Delta^n} f_1(\mathbf{q}) \le c\min_{\mathbf{q}\in \Delta^n} f_1(\mathbf{q}) \le cf_1(\mathbf{p})$, so $P: = \max_{\mathbf{q}\in \Delta^n} f_1(\mathbf{q})$, $a := 1$, and $b := c$ are valid parameters for algorithm $\mathcal{A}_1$. By our assumption, we can get an estimate of $f_1(\mathbf{p})$, denoted by $\tilde{f}_1(\mathbf{p})$, to within multiplicative error $\epsilon_1$ with success probability at least $1-\frac{\delta}{l}$ using $\mathcal{A}_1$ with $Q_1$ calls to $U_{\mathrm{pure}}$ and $U_{\mathrm{pure}}^{\dagger}$.

For each $k$ from $2$ to $l$, we let
\begin{align}
    \label{eq:parameters_k}
\hspace{-2mm}   P_k :=\max_{\mathbf{q}\in \Delta^n,\tilde{f}_{k-1}(\mathbf{p})/f_{k-1}(\mathbf{q})\in [1-\epsilon_{k-1},1+\epsilon_{k-1}]}f_{k}(\mathbf{q}),
\end{align}
set the parameters of $\mathcal{A}_k$ as $P:=P_k, a:=1, b:=c$, run $\mathcal{A}_k$ to estimate $f_k(\mathbf{p})$, and denote the output by $\tilde{f}_{k}(\mathbf{p})$.

We prove that the output $\tilde{f}_{k}(\mathbf{p})$ is an estimate of $f_k(\mathbf{p})$ to within multiplicative error $\epsilon_k$ with success probability at least $1-\frac{\delta k}{l}$ by induction. The statement is true for $k = 1$ by our assumption. If the statement is true for $k-1$, which means $\tilde{f}_{k-1}(\mathbf{p})/f_{k-1}(\mathbf{p})\in[1-\epsilon_{k-1},1+\epsilon_{k-1}]$ with probability at least $1-\frac{\delta(k-1)}{l}$, from \eq{parameters_k}, we can infer that
\begin{align}
\hspace{-2.5mm}    f_k(\mathbf{p})&\le  \max_{\mathbf{q}\in \Delta^n,\frac{\tilde{f}_{k-1}(\mathbf{p})}{f_{k-1}(\mathbf{q})}\in [1-\epsilon_{k-1},1+\epsilon_{k-1}]}f_{k}(\mathbf{q}) = P_k \\
\hspace{-2.5mm}    f_k(\mathbf{p}) &\ge \min_{\mathbf{q}\in \Delta^n,\frac{\tilde{f}_{k-1}(\mathbf{p})}{f_{k-1}(\mathbf{q})}\in [1-\epsilon_{k-1},1+\epsilon_{k-1}]}f_{k}(\mathbf{q}) \ge \frac{P_k}{c} 
\end{align}
where the last inequality comes from \eq{annealing_condition}. Therefore, $(P_k,1,c)$ are valid parameters for $\mathcal{A}_k$, and the output $\tilde{f}_k(\mathbf{p})$ is an estimate of $f_k(\mathbf{p})$ to within multiplicative error $\epsilon_k$ with success probability at least $(1-\frac{\delta(k-1)}{l})(1-\frac{\delta}{l})\ge 1-\frac{\delta k}{l}$, which completes the induction proof.

In conclusion, $\tilde{f}_l(\mathbf{p})$ is an estimate of $f_l(\mathbf{p}) = f(\mathbf{p})$ to within multiplicative error $\epsilon_l = \epsilon$, and the query complexity of the whole algorithm is $\sum_{k=1}^l Q_k$.
\end{proof}

Our framework generalizes the annealing technique used in \cite{li2019entropy} to a family of functions and make it compatible with \algo{main_algo} based on QSVT. Although our annealing scheme is similar to \cite{li2019entropy}, we use quite different estimation algorithms, so the way we combine it with the annealing scheme is also different. In fact, the main reason why \cite{li2019entropy} needs annealing is that their algorithm used an estimation subroutine in \cite{hamoudi2019Chebyshev}, which requires a rough estimate of the mean by the annealing. However, in a follow-up work \cite{hamoudi2021quantum}, this requirement is removed, so the annealing  becomes unnecessary for the algorithm in \cite{li2019entropy}.

\subsection{Improvements based on variable-time amplitude estimation}
\label{sec:main_algo_2}\noindent
In \sec{main_algo_standard}, we apply QSVT to all singular values of $E$ with the same transformation polynomial $S$. \textcolor{black}{For functions which is not smooth at $0$ such as $\frac{1}{x}$ or $x^{\alpha}$ for irrational $\alpha >0$, the complexity of applying them to singular values is proportional to the ratio of the largest possible singular value to the smallest possible singular value $\frac{\sigma_{\max}}{\sigma_{\min}}$. Improvements in this section can be summarized as dividing the algorithm into multiple phases and applying QSVT to a narrower range of singular values in each phase.}

\textcolor{black}{This idea comes from \cite{ambainis2012variable} and \cite{childs2017} which improved the complexity of Quantum Linear System Solver from $\tilde{O}(\kappa^2)$ to $\tilde{O}(\kappa)$ by VTAA. Although they do not use QSVT, solving a linear system is equivalent to applying the function $\frac{1}{x}$ to eigenvalues of the matrix, and hence similar ideas still work in our setting.}


\textcolor{black}{Basically, we replace the QSVT subroutine $\mathcal{A}$ in \algo{main_algo} with a variable-stopping-time quantum algorithm and replace the standard amplitude estimation with variable-time amplitude estimation.} In the variable-stopping-time algorithm, we only apply singular value transformation to the singular values in a small pre-defined interval in each stage. In this way, those branches stopping in an early stage make $T_{\mathrm{avg}}$ smaller than $T_{\max}$.


Before describing the improved main algorithm, we give an algorithm to separate singular values.


\vspace{3mm}
\noindent
\textbf{Singular values separation.}
In order to transform different singular values in different stages of $\mathcal{A}$, we need to decompose a state into several components and each of them is a linear combination of singular vectors of $E$ whose singular values fall into a small interval.

In \cite{childs2017}, they use a gapped phase estimation  
algorithm and Hamiltonian simulation algorithm to separate eigenvalues in different intervals. We extend their algorithm to the following one which can deal with singular values.

\begin{lemma}
\label{lem:separate_singular_value}
Let $U$ be a unitary, and $\widetilde{\Pi}, \Pi$ orthogonal projectors with the same rank $d$ acting on $\mathcal{H}_I$. Suppose $A = \widetilde{\Pi}U\Pi$ has a singular value decomposition $A = \sum_{i = 1}^d\sigma_i|\tilde{\psi}_i\rangle\langle \psi_i|_I$. Let $\varphi \in(0,1]$ and $\epsilon>0$. Then there is a unitary $W(\varphi,\epsilon)$ using $\mathcal{O}\left(\frac{1}{\varphi} \log \frac{1}{\epsilon}\right)$ calls to $U$ and $U^{\dagger}$ such that
\begin{align}
    &W(\varphi,\epsilon)|0\rangle_{C}|0\rangle_{P}|\psi_i\rangle_I \nonumber \\ =& \beta_{0}|0\rangle_C|\gamma\rangle_{P,I}+\beta_{1}|1\rangle_{C}|+\rangle_{P}|\psi_i\rangle_I
\end{align}
where $|\beta_{0}|^{2}+|\beta_{1}|^{2}=1$, such that
  \begin{itemize}
    \item if $0 \leq \sigma_i \leq \varphi$ then $\left|\beta_{1}\right| \leq \epsilon$ and
    \item if $2 \varphi \leq \sigma_i \leq 1$ then $\left|\beta_{0}\right| \leq \epsilon$.
          \label{lem:gpe}
  \end{itemize}
  Here $C$ and $P$ are two single-qubit registers, and $I$ is the register that $A$ acts on. 
\end{lemma}

The proof of \lem{separate_singular_value} is deferred to \append{other_proof}.

\begin{algorithm*}[ht]
  \KwInput{Multiplicative error $\epsilon$, and quantum registers $(F,C,A,B,Q,P,I)$ initialized to $|\mathbf{0}\rangle$.}
  \KwOutput{Quantum state in registers $(F,C,A,B,Q,P,I)$}
  \KwParameter{$(S,\beta,L,m,\{S_j\mid j\in[m]\})$, where $S$ is the transform polynomial we want to approximate, $\beta$ is an upper bound of $\sqrt{p_i}$ for all $i\in[n]$, $L$ is a lower bound of $\sum_{i = 1}^np_iS(\sqrt{p_i})^2$, $m$ is the number of stages of $\mathcal{A}$, and $S_j$ is the transformation polynomial in $\mathcal{A}_j$, which satisfies \eq{qsvt_condition} in \thm{qsvt} and \eq{S_j_condition}.

  }
  \KwNotation{
  \begin{itemize}[leftmargin=*]
  \setlength{\itemsep}{0pt}
    \setlength{\parsep}{0pt}
    \setlength{\parskip}{0pt}
\item  $F$ is a single-qubit flag register indicating ``good" components whose amplitude we estimate;

\item  $C = (C_{1}, \ldots, C_{m})$ is an $m$-qubit clock register determining the interval which the singular value belongs to;

\item  $A,B$ are two $\lceil \log n\rceil$-qubit input registers;

\item    $Q$ is a single-qubit register used as ancilla register for QSVT;

\item  $P = (P_{1}, P_{2}, \ldots, P_{m})$ and $I = (I_{1},I_{2},\ldots,I_{m})$ are two registers used as ancilla registers in \lem{separate_singular_value}. Each $I_{j}$ is a $2\lceil \log n\rceil$-qubit register and each $P_j$ is a single-qubit register; and

\item  $U^{(SV)}_{S_j}$ is the unitary $U^{(SV)}_P$ in \thm{qsvt} when the transformation polynomial $P:=S_j$ and $U,\widetilde{\Pi},\Pi$ in \thm{qsvt} are set to $U,\widetilde{\Pi},\Pi$ in \eq{block-1}, \eq{block-1} and \eq{def_E}.
\end{itemize}
  }
  Set $\varphi_{j}:=\beta2^{-j}$ for $j = 0,\ldots,m-1$, and $\varphi_m := 0$\;
  \SetKwProg{myproc}{Regard the following subroutine as $\mathcal{A}_j$:}{}{}
  \LinesNumberedHidden{
    \nonl\myproc{}{
      Conditional on first $j-1$ qubits in register $C$ being $|\mathbf{0}\rangle$, apply CNOT gates controlled by qubits in register $B$ to flip the last $\lceil \log n\rceil$ qubits in register $I_j$\;
      Conditional on first $j-1$ qubits in register $C$ being $|\mathbf{0}\rangle$, if $j < m$, apply $W(\varphi_j,L\epsilon/m)$ in \lem{separate_singular_value} with $\widetilde{\Pi},\Pi,U$ defined in \sec{main_algo_standard} to the state in register $I_j$ using $C_{j}$ as the output register $C$ and $P_j$ as the ancilla register $P$,
      else apply $X$ gate to register $C_m$\;

      Conditional on $C_{j}$ being $|1\rangle_{C_{j}}$, apply unitary $\text{C}_{\widetilde{\Pi }\otimes|+\rangle\langle+|}\text{NOT}\cdot U^{(SV)}_{S_{j}}$ or $\text{C}_{\Pi\otimes|+\rangle\langle+|}\text{NOT}\cdot U^{(SV)}_{S_{j}}$ in the same way as in \lin{state_after_U_S} to \lin{state_after_cnot} of \algo{main_algo} to the state in $(A,B,Q,F)$\;
    }
  }

    Apply $U_{\mathrm{pure}}$ to register $B$, Hadamard gate $H$ to $Q$\;
    Apply $\mathcal{A}:=\mathcal{A}_m\cdots\mathcal{A}_1$ to registers $(F,C,A,B,Q,P,I)$\;
    \caption{A variable-stopping-time subroutine when estimating $\sum_{i=1}^n p_i S(\sqrt{p_i})^2$.}
  \label{algo:improved_main_algo_sub}
\end{algorithm*}
\vspace{3mm}
\noindent
\textbf{Variable-stopping-time subroutine $\mathcal{A}$.}
We now describe the $m$-stage variable-stopping-time quantum algorithm $\mathcal{A}=\mathcal{A}_{m} \cdot \cdots \cdot \mathcal{A}_{1}$.

To construct $\mathcal{A}$, we suppose that we are given $\beta \in (0,1]$ such that $\sqrt{p_i} \in [0,\beta)$ for all $i = 1,\ldots,n$ and a lower bound $L > 0$ such that $\sum_{i = 1}^n p_i S(\sqrt{p_i})^2 \ge L$.

Let $\varphi_{j}:=\beta2^{-j}$ for $j = 0,\ldots,m-1$ and $\varphi_m := 0$. We first divide $[0,\beta)$ into $m$ intervals $[\varphi_m,\varphi_{m-1}),[\varphi_{m-1},\varphi_{m-2}),\ldots,$ $[\varphi_{1},\varphi_0)$. Then we transform singular values in these intervals in different stages of $\mathcal{A}$. Specifically,
\begin{itemize}
    \item for $j = 2,\ldots,m$, we transform singular values in $[\beta 2^{-j},\beta 2^{-j+2})$ in $\mathcal{A}_j$, and
    \item for $j = 1$, we transform singular values in $[\frac{1}{2}\beta,\beta)$ in $\mathcal{A}_1$.
\end{itemize}

Then we construct the $j$-th stage of $\mathcal{A}$. First, we need to determine the transformation polynomial, $S_j$ in this stage. Since we like $S_{j}$ to perform a transformation similar to $S$, we need to construct polynomials $S_{j}$ for $j = 1,\ldots, m$ such that $S_{j}$ satisfies \eq{qsvt_condition} in \thm{qsvt} and
\begin{align}
\label{eq:S_j_condition}
    |S_{1}(x)-S(x)| \le L \epsilon &\text{ for all } x \in [\beta/2,\beta), \nonumber\\
    |S_{j}(x)-S(x)|\le L \epsilon &\text{ for all } x\in[\beta 2^{-j},\beta 2^{-j+2}),\nonumber \\ &\text{ and }j=2,\ldots,m.
\end{align}

Note that for any $S_j$, we only require it to be a good approximation of $S(x)$ in a small interval, so we may construct such polynomial with lower degree than $S$. Since the complexity of variable-time amplitude estimation is proportional to the average time of all stages, which is the average degree of all transformation polynomial $S_j$, this variable-stopping-time algorithm can improve our vanilla algorithm in \sec{main_algo_standard}. 

Assuming that we have constructed such $S_j$ satisfying \eq{S_j_condition}, we give a detailed description of $\mathcal{A}$ in \algo{improved_main_algo_sub}.

\vspace{3mm}
\noindent
\textbf{Final algorithm.} We now describe our final algorithm in \algo{main_algo_improved}.

\begin{algorithm*}[ht]
  \KwInput{$(\epsilon,\delta)$, where $\epsilon$ is the multiplicative error, and $\delta$ is the failure probability; quantum registers $(F,C,A,B,Q,P,I)$ initialized to $|\mathbf{0}\rangle$.}
  \KwOutput{$\tilde{p}$.}
  \KwNotation{Notations is the same as that of \algo{improved_main_algo_sub}.
    }
    \KwParameter{$(S, \beta,L,m,\{S_j\mid j\in[m]\})$.
  }
    Use variable-time amplitude estimation algorithm in \thm{vtae} with the parameters $(\epsilon,\delta)$ to be $\epsilon:=\epsilon,\delta:=\delta$ and registers $W:=(A,B,Q,P,I),C:=C,F:=F$ to estimate $\|(|1\rangle\langle1|_F\otimes I)\widetilde{\mathcal{A}}(|\mathbf{0}\rangle\|^2$, where $\widetilde{\mathcal{A}}$ is \algo{improved_main_algo_sub} with the same parameters $(S,\beta,L,m,\{S_j\mid j\in[m]\})$ and input $\epsilon$. Denote the output of variable-time amplitude estimation algorithm by $\tilde{p}$. Output $\tilde{p}$\;
    \caption{Improved algorithm for estimating $\sum_{i=1}^n p_i S(\sqrt{p_i})^2$ of $\mathbf{p} = (p_i)_{i=1}^n$. }
  \label{algo:main_algo_improved}
\end{algorithm*}

We prove the output of \algo{main_algo_improved} is an estimate of $\sum_{i=1}^n p_i S(\sqrt{p_i})^2$ to within multiplicative error $\epsilon$ with high probability in the following proposition:
\begin{proposition}
\label{prop:correctness_final}
  Let $\epsilon, \delta, \beta \in(0,1)$, $\mathbf{p} = (p_i)_{i = 1}^n$ be a probability distribution such that $\sqrt{p_i}\le \beta$ for all $i = 1,\ldots,n$ and $S$ be a polynomial which satisfies \eq{qsvt_condition} in \thm{qsvt}. Suppose that we are given $\beta$, $L > 0$ such that $\sum_{i = 1}^n p_i S(\sqrt{p_i})^2 \ge L$ and a sequence of polynomials $\{S_j\mid j = 1,\ldots,m\}$ which satisfy \eq{qsvt_condition} in \thm{qsvt} and \eq{S_j_condition}. \algo{main_algo_improved} with input $(\epsilon,\delta)$ and parameters $(S, \beta,L,m,\{S_j\mid j\in[m]\})$ outputs an estimate of $\sum_{i=1}^n p_i S(\sqrt{p_i})^2$ to within multiplicative error $\epsilon$ with success probability at least $1-\delta$.

  Let $t_j = \frac{2^{j}}{\beta}\log(\frac{m}{\epsilon L}) + \sum_{k = 1}^{j}\deg(S_k)$ for all $j = 1,\ldots,m$ and $t_{m+1} = t_m$, the query complexity of \algo{main_algo_improved} is
  \begin{align}
    \widetilde{\mathcal{O}}\left(t_m + \sqrt{\epsilon} \sum_{j=1}^m t_j + \frac{\sqrt{\sum_{j = 1}^m\sum_{i:\sqrt{p_i}\in[\varphi_j,\varphi_{j-1})}p_i t_{j+1}^2}}{\sqrt{\sum_{i = 1}^np_iS(\sqrt{p_i})^2}}\right),
  \end{align}
  where $\varphi_j = 2^{-j}\beta$ for $j = 1,\ldots,m-1$ and $\varphi_m = 0$. We omit $\mathrm{polylog}$ terms of $L$, $\delta$ and $t_m$ in this bound. 
\end{proposition}
The proof of \prop{correctness_final} is deferred to \append{correctness_final}.

\begin{remark}
The $\sqrt{\epsilon}\sum_{j=1}^{m} t_j$ term in query complexity can be eliminated by a more detailed analysis mentioned in \cite{childs2017}, but it does not improve the complexity bounds in our applications.
\end{remark}

\section{R{\'e}nyi Entropy Estimation ($\alpha>1$)}\label{sec:alpha-large}
\noindent
In this section, we propose a quantum algorithm to estimate $H_{\alpha}(\mathbf{p}) = \frac{1}{1-\alpha}\sum_{i=1}^np_i^{\alpha}$ for $\alpha > 1$ to within additive error $\epsilon$. This is equivalent to estimating $P_{\alpha}(\mathbf{p}) = \sum_{i=1}^n p_i^{\alpha}$ to within multiplicative error $\mathcal{O}(\epsilon)$.

Let $g(x) := x^{\alpha-1}$, and then we have $P_{\alpha}(\mathbf{p}) = \sum_{i=1}^n p_i g(p_i)$. Since $g(x) = x^{\alpha-1}$ for $\alpha > 1 $ is monotonically increasing function on $[0,1]$ such that $g(0)= 0, g(1) = 1$ and $xg(x) = x^{\alpha}$ is a convex function, we can use the framework in \sec{annealing} to construct our algorithm.

\subsection{Estimate $P_{\alpha}(\mathbf{p})$ given a rough bound}
\noindent
We first construct a quantum algorithm which can estimate $P_{\alpha}(\mathbf{p})$ to within multiplicative error given $P,a,b$ such that $a P_{\alpha}(\mathbf{p}) \le P \le b P_{\alpha}(\mathbf{p})$ following \lem{annealing_pre}. Let $p^* = (xg(x))^{-1}(\min(\frac{1}{a}P,1)) = (\min(\frac{1}{a}P,1))^{\frac{1}{\alpha}}$. Like \lem{annealing_pre}, we need to construct a polynomial $S(x)$ which satisfies \eq{qsvt_condition} in \thm{qsvt} and \eq{annealing_pre}, which means
\begin{align}
  \Big|\sum_{i=1}^n q_i S(\sqrt{q_i})^2 &- c\frac{1}{g(p^*)}f(\mathbf{q})\Big| \le \frac{ca}{2b} p^{*}\epsilon_0
\end{align}
for all $\mathbf{q}$ such that $q_i\le p^*$ for all $i\in[n]$, where $c> 0$ is an arbitrary constant.

Before constructing such a polynomial, we first construct a class of polynomials which satisfies \eq{qsvt_condition} in \thm{qsvt} and is also a good approximation to $2^{-c-1}\beta^{-c}x^{c}$ in $[\nu,\beta]$ for any $\beta \in (0,1]$, $\nu \in (0,\beta)$, and $c > 0$.
\begin{lemma}
  \label{lem:scaled_poly_approx_with_small_value_at_zero}
  For any $c > 0$, $\beta \in (0,1]$, $\nu \in(0,\beta)$, and $\eta\in(0,\frac{1}{2})$, let $f(x) = 2^{-c-1}\beta^{-c}x^c$, there is an efficiently computable even or odd polynomial $S\in \mathbb{R}[x]$ of degree $\mathcal{O}\left(\frac{c}{\nu}\log(\frac{1}{\beta\nu\eta})\right)$ such that
  \begin{align}
      \forall x & \in [0,\nu]: |S(x)| \le 2f(x)         \nonumber \\
      \forall x & \in [\nu, \beta]: |f(x)-S(x)| \le \eta \nonumber \\
      \forall x & \in [-1,1]: |S(x)| \le 1               
  \end{align}
\end{lemma}
\begin{proof}
  Let $d = \lceil c \rceil -c$.
  We first introduce a lemma to construct polynomial approximation of $kx^{-d}$ where $k$ is a constant.
  \begin{lemma}[{\cite[Corollary 67, Polynomial approximations of negative power functions]{gilyen2019singular}}] Let $\delta, \varepsilon \in\left(0, \frac{1}{2}\right], c>0$, and let $f(x):=\frac{\delta^{c}}{2} x^{-c}$, then there exist even/odd polynomials $P, P^{\prime} \in \mathbb{R}[x]$ such that $\|P-f\|_{[\delta, 1]} \leq \varepsilon$, $\|P\|_{[-1,1]} \leq 1$ and similarly $\left\|P^{\prime}-f\right\|_{[\delta, 1]} \leq \varepsilon,\left\|P^{\prime}\right\|_{[-1,1]} \leq 1$. In addition, the degree of the polynomials
  are $\mathcal{O}\left(\frac{\max [1, c]}{\delta} \log \left(\frac{1}{\varepsilon}\right)\right)$.
  \label{lem:poly-approx-nega-power}
\end{lemma}

Setting the parameters $(\epsilon,\delta,c)$ in \lem{poly-approx-nega-power} to $\delta \coloneqq \nu, c \coloneqq d, \epsilon \coloneqq \frac{\beta^{c}\nu^{d}\eta}{2}$, we can construct an even polynomial $\tilde{Q}\coloneqq P$ in \lem{poly-approx-nega-power} with $\deg(\tilde{Q}) = \mathcal{O}\left(\frac{1}{\nu}\log(\frac{1}{\beta\nu\eta})\right)$ such that  
\begin{align}
      & \forall x \in [-1,1]:        |\tilde{Q}(x)| \le 1                                                      \nonumber\\
      & \forall x \in [\nu,1]:       |\tilde{Q}(x)-\frac{\nu^{d}}{2}x^{-d}|\le \frac{\beta^{c}\nu^{d}\eta}{2}.
\end{align}

Then we need to construct a polynomial approximation to the rectangle function according to the following lemma:
\begin{lemma}[Polynomial approximations of the rectangle function {\cite[Lemma 29]{gilyen2019singular}}]\label{lem:rectangle_approximation}
  Let $\delta^{\prime}, \varepsilon^{\prime} \in\left(0, \frac{1}{2}\right)$ and $t$ satisfying $\delta' \le t \le 1$. There exists an even polynomial $P^{\prime} \in \mathbb{R}[x]$ of degree $\mathcal{O}\left(\log \left(\frac{1}{\varepsilon^{\prime}}\right) / \delta^{\prime}\right)$, such that $\left|P^{\prime}(x)\right| \leq 1$ for all $x \in[-1,1]$, and
  \begin{align}
    \begin{cases}P^{\prime}(x) \in \left[0, \varepsilon^{\prime}\right] & \forall x \in\left[-1,-t-\delta^{\prime}\right] \cup\left[t+\delta^{\prime}, 1\right], \\ P^{\prime}(x) \in\left[1-\varepsilon^{\prime}, 1\right] & \forall x \in\left[-t+\delta^{\prime}, t-\delta^{\prime}\right].\end{cases}
  \end{align}
\end{lemma}

  Setting $\delta'\coloneqq \frac{\beta}{2}, t \coloneqq \frac{3\beta}{2}, \epsilon' \coloneqq{\frac{\beta^{c}\nu^{d}\eta}{2}}$ in \lem{rectangle_approximation}, we can construct an even polynomial $P$ with $\deg(P) = \mathcal{O}\left(\frac{1}{\beta}\log(\frac{1}{\beta\nu\eta})\right)$ such that
  \begin{align}
       & \forall x \in[-1,1]: |P(x)| \le 1                                            \nonumber \\
       & \forall x \in [2\beta,1]:|P(x)| \le \frac{\beta^{c}\nu^{d}\eta}{2}           \nonumber \\
       & \forall x \in [-\beta,\beta]:1-\frac{\beta^{c}\nu^{d}\eta}{2} \le P(x) \le 1. 
  \end{align}
      
  Let $Q(x) \coloneqq \tilde{Q}(x)P(x)$, we have
  \begin{align}
    \forall x \in[\nu,\beta]: &|Q(x)-\frac{\nu^{d}}{2}x^{-d}| \nonumber\\ \le& |Q(x)-\tilde{Q}(x)|+|\tilde{Q}(x)-\frac{\nu^{d}}{2}x^{-d}|\nonumber\\
    \le& |\tilde{Q}(x)||1-P(x)| + \frac{\beta^c\nu^d\eta}{2}\nonumber\\
    \le&|1-P(x)|+\frac{\beta^c\nu^d\eta}{2} \le \beta^{c}\nu^{d}\eta \nonumber
  \end{align}
  \begin{align}
      \label{eq:small_value_eqn1}
      \forall x \in[-1,1]: |Q(x)| & \le 1\nonumber\\
      \forall x \in [\beta,2\beta]: |Q(x)|
      & \le |\tilde{Q}(x)| \le \frac{\nu^dx^{-d}}{2}+\frac{\beta^c\nu^d\eta}{2} \nonumber\\
      \forall x \in [2\beta,1]: |Q(x)|                         & \le |P(x)|\le \frac{\beta^c\nu^d\eta}{2}.
  \end{align}

  Then, let $S(x) \coloneqq 2^{-c}\beta^{-c}\nu^{-d}x^{\lceil c\rceil}Q(x)$, which is an even or odd polynomial since $Q(x) = \tilde{Q}(x)P(x)$ and $\tilde{Q}(x),P(x)$ are even polynomials. We can infer that
  \begin{align}
    \label{eq:small_value_eqn2}
    \forall x \in [0,\nu]: S(x)\le & 2^{-c}\beta^{-c}\nu^{-d}x^{\lceil c \rceil} \nonumber \\= &  2^{-c}\beta^{-c}\nu^{-d}x^{d}x^{c} \nonumber \\ \le &  2^{-c}\beta^{-c}x^{c} = 2f(x),
  \end{align}
  where the first inequality comes from $|Q(x)|\le 1$, and
  \begin{align}
      \label{eq:small_value_eqn3}
      \forall x \in[\nu,\beta]: &|S(x)-f(x)| \nonumber \\
      = & |S(x)-2^{-c-1}\beta^{-c}x^c| \nonumber \\ =&2^{-c}\beta^{-c}\nu^{-d}x^{\lceil c \rceil}|Q(x)-\frac{\nu^dx^{-d}}{2}| \nonumber\\
      \le & 2^{-c}\beta^{-c}\nu^{-d}x^{\lceil c \rceil}(\beta^c\nu^d\eta)        \nonumber\\
      \le & 2^{-c}\eta \le \eta,
  \end{align}
  
  where the first inequality comes from \eq{small_value_eqn1}.

  To prove that $S(x)$ is bounded by $1$ on $[-1,1]$, since $S(x)$ is even or odd, we only need to prove $\forall x\in [0,1]: |S(x)|\le 1$.
  For $x\in[0,2\beta]$, we have
  \begin{align}
    \forall x \in [0,\nu]:|S(x)|     & \le 2f(x) \le 1   \nonumber                 \\
    \forall x\in [\nu,\beta]: |S(x)| & \le 2^{-c-1}\beta^{-c}x^{c} + \eta
    \le \frac{1}{2} + \eta
    \le 1                              ,
  \end{align}
  where the first inequality comes from \eq{small_value_eqn2} and the third inequality comes from \eq{small_value_eqn3}, and
  \begin{align}
    \forall x\in [\beta,2\beta]: & |S(x)| \nonumber \\
    \le & 2^{-c-1}\beta^{-c}\nu^{-d}x^{\lceil c \rceil}\left(\frac{\nu^dx^{-d}+\beta^c\nu^{d}\eta}{2}\right)\nonumber \\
    \le & 2^{-c-1}\frac{\beta^{-c}x^c+\eta}{2} \nonumber\\
    \le & 2^{-c-1}\frac{2^c+\eta}{2} \le 1,
  \end{align}
  where the first inequality comes from \eq{small_value_eqn1}.
  
  For $x\in[2\beta,1]$, we have
  \begin{align}
    \forall x \in[2\beta,1]: |S(x)| = & 2^{-c-1}\beta^{-c}\nu^{-d}x^{\lceil c \rceil}Q(x)\nonumber \\
    \le& 2^{-c-1}\beta^{-c}\nu^{-d} (\frac{\beta^c\nu^d\eta}{2}) \nonumber \\\le& \eta \le 1,
  \end{align}
  where the first inequality comes from \eq{small_value_eqn1}. Therefore, we can conclude that $\forall x\in [-1,1]: |S(x)|\le 1$. Together with \eq{small_value_eqn2}, \eq{small_value_eqn3} and that $S(x)$ is even or odd polynomial, we have that $S(x)$ with $\deg(S) = \lceil c \rceil + \deg(\tilde{Q}) + \deg{P} = \mathcal{O}\left(\frac{c}{\nu}\log(\frac{1}{\beta\nu\eta})\right)$ satisfies the conditions in this lemma.
\end{proof}

By carefully choosing the parameters in \lem{scaled_poly_approx_with_small_value_at_zero}, we can construct a polynomial which is similar to the polynomial $S$ in \eq{annealing_pre} as follows.

\begin{lemma}
  \label{lem:polynomial_large_alpha}
  For any probability distribution $\mathbf{p}$ on $[n]$, suppose that we are given $P,a,b$ such that $a P_{\alpha}(\mathbf{p}) \le P \le b P_{\alpha}(\mathbf{p})$. Let $p^* = \min(\frac{1}{a}P,1)^{\frac{1}{\alpha}}$. Then for any $\epsilon \in (0,1)$, $\alpha > 1$ and constants $d,d' > 0$, the polynomial $S$ in \lem{scaled_poly_approx_with_small_value_at_zero} with parameters $(c,\beta,\nu,\eta)$ to be $c := \alpha-1,\beta := \sqrt{p^*}, \nu := \left(d\frac{a(p^*)^{\alpha}\epsilon}{5bn}\right)^{\frac{1}{2\alpha}}, \eta := d'2^{-2\alpha}\frac{a}{b}p^*\epsilon$ has $\deg(S) = \widetilde{\mathcal{O}}\Bigl(\left(\frac{n}{\epsilon}\right)^{\frac{1}{2\alpha}} \frac{1}{\sqrt{p^*}}\Bigr)$, and satisfies \eq{qsvt_condition} in \thm{qsvt} and
  \begin{align}
  \label{eq:approximation_large_alpha}
      &\left|\sum_{i=1}^n p_iS(\sqrt{p_i})^2 -2^{-2\alpha}(p^*)^{1-\alpha}P_{\alpha}(\mathbf{p})\right| \nonumber \\ \le & (2d'+d)2^{-2\alpha}\frac{a}{b}p^*\epsilon.
  \end{align}
\end{lemma}

\begin{proof}
  \lem{scaled_poly_approx_with_small_value_at_zero} implies that $\deg(S)= \mathcal{O}\left(\frac{c}{\nu} \log(\frac{1}{\beta\nu\eta})\right) =\widetilde{\mathcal{O}}\bigl(\left(\frac{n}{\epsilon}\right)^{\frac{1}{2\alpha}} \frac{1}{\sqrt{p^*}}\bigr)$, $S$ is an even or odd polynomial, and $|S(x)|\le 1$ for all $x\in[-1,1]$. Therefore, $S$ satisfies \eq{qsvt_condition} in \thm{qsvt}.

  We now prove that $S$ satisfies \eq{approximation_large_alpha}. From the definition of $P_\alpha(\mathbf{p})$, we can infer that
    \begin{align}
      \label{eq:range-p}
      \forall i,p_i \le \left(\sum_{i = 1}^n p_i^{\alpha}\right)^{\frac{1}{\alpha}} \le \left(\frac{1}{a}P\right)^{\frac{1}{\alpha}}
    \end{align}
    and $p_i \le 1$, so $p_i\le \min((\frac{1}{a}P)^{\frac{1}{\alpha}},1) = p^* = \beta^2$.

  From \lem{scaled_poly_approx_with_small_value_at_zero}, we can infer that $S$ satisfies
  \begin{align}
    \forall x & \in [0,\nu]: |S(x)| \le  2^{1-\alpha}\beta^{1-\alpha}x^{\alpha-1}          \label{eq:S-less-2f}  
    \\
    \forall x & \in [\nu, \beta]: |2^{-\alpha}\beta^{1-\alpha}x^{\alpha-1}-S(x)| \le \eta \label{eq:error-less-eta}. 
  \end{align}
      
  For $i$ such that $\sqrt{p_i} \le \nu$, we have
  \begin{align}
      &\sum_{\sqrt{p_i} \le \nu} |p_iS(\sqrt{p_i})^2-2^{-2\alpha}\beta^{-2\alpha+2}p_i^{\alpha}|
        \nonumber \\ \le & \sum_{\sqrt{p_i} \le \nu}|p_iS(\sqrt{p_i})^2| + |2^{-2\alpha}\beta^{-2\alpha+2}p_i^{\alpha}|      \notag                          \\
        \le & \sum_{\sqrt{p_i} \le \nu} p_i2^{-2\alpha+2}\beta^{-2\alpha+2}p_i^{\alpha-1} + 2^{-2\alpha}\beta^{-2\alpha+2}p_i^{\alpha}               \notag   \\
        \le & \sum_{\sqrt{p_i} \le \nu} 2^{-2\alpha}\beta^{-2\alpha+2}5p_i^{\alpha}                                                 \notag    \\
        \le & \sum_{\sqrt{p_i} \le \nu} 2^{-2\alpha}\beta^{-2\alpha+2}5\nu^{2\alpha}                                                  \notag  \\
        \le & n2^{-2\alpha}\beta^{-2\alpha+2}5\nu^{2\alpha} \notag \\
        = &n2^{-2\alpha}(p^*)^{1-\alpha}5\left(d\frac{a(p^*)^{\alpha}}{5bn}\epsilon\right)\notag \\
        = &d2^{-2\alpha}\frac{a}{b}p^* \epsilon, \label{eq:large_alpha_little_delta}
    \end{align}
  where the second inequality comes from \eq{S-less-2f}.

  For $i$ such that $\sqrt{p_i} > \nu$, we have
  \begin{align}
      & \sum_{\sqrt{p_i} > \nu} |p_iS(\sqrt{p_i})^2-2^{-2\alpha}\beta^{-2\alpha+2}p_i^{\alpha}|
      \nonumber \\= & \sum_{\sqrt{p_i}> \nu} p_i|S(\sqrt{p_i})^2-2^{-2\alpha}\beta^{-2\alpha+2}(\sqrt{p_i})^{2c}| \notag
      \\ = & \sum_{\sqrt{p_i}> \nu} p_i|S(\sqrt{p_i})-2^{-\alpha}\beta^{-\alpha+1}(\sqrt{p_i})^{c}|\nonumber \\ &\cdot|S(\sqrt{p_i})+2^{-\alpha}\beta^{-\alpha+1}(\sqrt{p_i})^{c}| \notag
      \\ \le & 2\sum_{\sqrt{p_i}> \nu} p_i|S(\sqrt{p_i})-2^{-\alpha}\beta^{-\alpha+1}(\sqrt{p_i})^{c}| \notag
      \\ \le & 2\eta \sum_{\sqrt{p_i} > \nu} p_i \le 2\eta = 2d'2^{-2\alpha}\frac{a}{b}p^*\epsilon, \label{eq:large_alpha_large_delta}
    \end{align}
  where the first inequality comes from \eq{range-p} and \eq{error-less-eta}.

  From \eq{large_alpha_little_delta} and \eq{large_alpha_large_delta}, we can infer that
  \begin{align}
    \label{eq:large_alpha_eqn1}
    &\left|\sum_{i=1}^n p_iS(\sqrt{p_i})^2 -2^{-2\alpha}(p^*)^{1-\alpha}P_{\alpha}(\mathbf{p})\right| \nonumber \\ \le & (2d'+d)2^{-2\alpha}\frac{a}{b}p^*\epsilon,
  \end{align}
  which completes the proof.
\end{proof}

Therefore, from \lem{annealing_pre}, there exists an algorithm which can estimate $P_{\alpha}(\mathbf{p})$ to within multiplicative error $\mathcal{O}(\epsilon)$ using $\widetilde{\mathcal{O}}\left(\frac{\deg(S)}{\epsilon \sqrt{p^*}}\right) = \widetilde{\mathcal{O}}\left(\frac{n^{\frac{1}{2\alpha}}}{\epsilon^{1+\frac{1}{2\alpha}} p^*}\right) = \widetilde{\mathcal{O}}\left(\frac{n^{\frac{1}{2\alpha}}}{\epsilon^{1+\frac{1}{2\alpha}}P^{\frac{1}{\alpha}}}\right)$ given $P,a,b$ such that $aP_{\alpha}(\mathbf{p}) \le P \le bP_{\alpha}(\mathbf{p})$.

Then we use \algo{main_algo_improved} to replace \algo{main_algo} in \lem{annealing_pre} and apply \prop{correctness_final} to achieve a better query complexity upper bound.

\begin{lemma}
  \label{lem:large_alpha_given_P}
  For any $\alpha > 1$, there exists an algorithm $\mathcal{A}$ such that for any $\delta \in (0,1)$, $\epsilon\in(0,1)$ and probability distribution $\mathbf{p}$ on $[n]$, given $P,a,b$ such that $aP_{\alpha}(\mathbf{p}) \le P \le bP_{\alpha}(\mathbf{p})$ where $a,b$ are two constants, $\mathcal{A}$ can estimate $P_{\alpha}(\mathbf{p})$ to within multiplicative error $\epsilon$
  with success probability at least $1-\delta$ using $ \widetilde{\mathcal{O}}\left(\frac{n^{1-\frac{1}{2\alpha}}}{\epsilon} + \frac{\sqrt{n}}{\epsilon^{1+\frac{1}{2\alpha}}}\right)$ calls to $U_{\mathrm{pure}}$ and $U_{\mathrm{pure}}^{\dagger}$ in \defn{pure-state-preparation}.
\end{lemma}
\begin{proof}
We will first construct such an algorithm $\mathcal{A}$ using \prop{correctness_final}, prove its correctness, and then compute its query complexity.

\noindent
\textbf{Construction and correctness. }Let $p^*: = \min(\frac{1}{a}P,1)^{\frac{1}{\alpha}}$ and $\beta_0 = \sqrt{p^*}$, and we have
\begin{align}
  \forall i,p_i \le \left(\sum_{i = 1}^n p_i^{\alpha}\right)^{\frac{1}{\alpha}} \le \left(\frac{1}{a}P\right)^{\frac{1}{\alpha}}
\end{align}
and $p_i \le 1$, so $p_i\le \min((\frac{1}{a}P)^{\frac{1}{\alpha}},1) = p^* = \beta_0^2$.

    Let $\nu_0 = \left(\frac{1}{4}\frac{a(p^*)^{\alpha}\epsilon}{5bn}\right)^{\frac{1}{2\alpha}}$. Before constructing $S$ and $S_j$ in \prop{correctness_final}, we first define the number of stages of our variable-stopping-time quantum algorithm $m_0:=\lceil \log(\frac{\beta_0}{\nu_0})\rceil+1$, and $\nu_j = 2^{-j} \beta_0$ for $j = 1,\ldots, m_0-1$, $\nu_{m_0} = \nu_{m_0-1} = 2^{-m_0+1}\beta_0$.

    Let $L := \frac{a}{b}2^{-2\alpha-1}p^*$, and we will prove that $L$ is an lower bound of $\sum_{i=1}^n p_i S_0(\sqrt{p_i})^2$ later.
    
    Let $S_0$ be the polynomial $S$ in \lem{scaled_poly_approx_with_small_value_at_zero} with parameters $(c,\beta,\nu,\eta)$ to be $c := \alpha-1$, $\beta:= \beta_0$, $\nu:= \nu_{m_0}$, $\eta := \frac{L}{4}\epsilon$.
    
    Let $S_j$ for $j=1,\ldots, m_0$ be the polynomial $S$ in \lem{scaled_poly_approx_with_small_value_at_zero} with parameters $(c,\beta,\nu,\eta)$ to be $c := \alpha-1$, $\beta:= \beta_0$, $\nu:= \nu_j$, $\eta = \frac{L}{4}\epsilon$. From \lem{scaled_poly_approx_with_small_value_at_zero}, we have
    \begin{align}
    \label{eq:equation_S_j}
        |S_j(x)-2^{-\alpha}\beta_0^{-\alpha+1}x^{\alpha-1}| \le \frac{L}{4}\epsilon \ \forall x\in [\nu_j,\beta_0].
    \end{align}

    Now we set the parameters $(S, \beta,L,m,\{S_j\mid j\in[m]\})$ of \algo{main_algo_improved} to be $S:=S_0$, $\beta := \beta_0$, $L:=L$, $m:=m_0$, $S_j := S_j$ for $j = 1,\ldots,m_0$, and then prove that these parameters satisfy the conditions in \prop{correctness_final}.
    \begin{itemize}
        \item For $\beta_0$, we have shown that it is an upper bound of $\sqrt{p_i}$.
        \item For $L$, we need to prove it is a lower bound of $\sum_{i=1}^n p_i S_0(\sqrt{p_i})^2$. Note that $\nu_{m_0} = 2^{-m_0+1}\beta_0 = 2^{-\lceil \log(\frac{\beta_0}{\nu_0})\rceil}\beta_0 \le \nu_0 = \left(\frac{1}{4}\frac{a(p^*)^{\alpha}\epsilon}{5bn}\right)^{\frac{1}{2\alpha}}$. Let $d = \nu_{m_0}^{2\alpha}/\left(\frac{a(p^*)^{\alpha}\epsilon}{5bn}\right)$, and $d' = \frac{L}{4}\epsilon/(2^{-2\alpha}\frac{a}{b}p^*\epsilon)$, and then we have $d \le \left(\frac{1}{4}\frac{a(p^*)^{\alpha}\epsilon}{5bn}\right)/\left(\frac{a(p^*)^{\alpha}\epsilon}{5bn}\right) =  \frac{1}{4}$, $d' = \frac{1}{8}$. Therefore, the parameters $(c=\alpha-1,\beta=\beta_0,\nu=\nu_{m_0}, \eta = \frac{L}{4}\epsilon)$ of $S_0$ satisfy the conditions in \lem{polynomial_large_alpha} with constants $d$ and $d'$. From \lem{polynomial_large_alpha}, we have
        \begin{align}
        \label{eq:approximation_large_alpha_final}
            &\left|\sum_{i=1}^n p_iS_0(\sqrt{p_i})^2 -2^{-2\alpha}(p^*)^{1-\alpha}P_{\alpha}(\mathbf{p})\right| \nonumber \\ \le& (2d'+d)2^{-2\alpha}\frac{a}{b}p^*\epsilon \le \frac{1}{2}2^{-2\alpha}\frac{a}{b}p^*\epsilon = L\epsilon.
        \end{align}
    
        From \eq{approximation_large_alpha_final}, we have
        \begin{align}
            \sum_{i=1}^n p_iS_0(\sqrt{p_i})^2 &\ge  2^{-2\alpha}(p^*)^{1-\alpha}P_{\alpha}(\mathbf{p}) - L\epsilon\nonumber\\ &\ge \frac{P}{b}2^{-2\alpha}(p^*)^{1-\alpha} - L\epsilon \nonumber\\
            &\ge \frac{a}{b}2^{-2\alpha}p^* - L\epsilon = 2L- L\epsilon \ge L.
        \end{align}
    Therefore, $L$ is a lower bound of $\sum_{i=1}^n p_i S_0(\sqrt{p_i})^2$.
    \item For $S_j$, they are constructed by applying \lem{scaled_poly_approx_with_small_value_at_zero}, so they satisfy \eq{qsvt_condition} in \thm{qsvt}.
    Note that the parameters of $S_0$ in \lem{scaled_poly_approx_with_small_value_at_zero} is the same as the parameters of $S_{m_0}$, so we have $S_0 = S_{m_0}$. Then we can infer that for any  $x\in [\nu_j,\beta_0]$, 
    \begin{align}
        &|S_j(x)-S_0(x)|\nonumber \\ = & |S_j(x)-S_{m_0}(x)| \nonumber\\
        \le& |S_j(x)-2^{-\alpha}\beta_0^{-\alpha+1}x^{\alpha-1}| \nonumber \\ &+  |S_{m_0}(x)-2^{-\alpha}\beta_0^{-\alpha+1}x^{\alpha-1}|\nonumber \\
        \le& \frac{1}{4}L\epsilon +\frac{1}{4}L\epsilon \le L\epsilon  \label{eq:equation_S_j_2},
    \end{align}
    where the second inequality comes from \eq{equation_S_j} and $\nu_{m_0} < \nu_j$ for $j < m_0$.
    From \eq{equation_S_j_2}, we can infer that $\text{for any } x\in [\beta_0 2^{-j},\beta_0]:|S_j(x)-S_0(x)| \le L\epsilon$, which meets the requirements of \eq{S_j_condition}.
    \end{itemize}
    
    Therefore, the parameters we set are valid for \prop{correctness_final}, so \algo{main_algo_improved} with the same parameters and input $(\epsilon,\delta)$ can estimate $\sum_{i=1}^n p_i S_0(\sqrt{p_i})^2$ to within multiplicative error $\epsilon$ within success probability at least $1-\delta$. Denote the estimate by $\tilde{p}$. 
    
    Note that \begin{align}
        2^{-2\alpha}(p^*)^{1-\alpha}P_{\alpha}(\mathbf{p}) \nonumber  \ge & 2^{-2\alpha}(p^*)^{1-\alpha}\frac{P}{b} \nonumber \\ \ge& 2^{-2\alpha}(p^*)^{1-\alpha}\frac{a(p^*)^{\alpha}}{b}\nonumber \\ = & 2^{-2\alpha}\frac{a}{b}p^* =\frac{L}{2},
    \end{align}
    and because $\sum_{i=1}^n p_iS_0(\sqrt{p_i})^2$ is an approximation of $2^{-2\alpha}(p^*)^{1-\alpha}P_{\alpha}(\mathbf{p})$ within additive error $L\epsilon$ from \eq{approximation_large_alpha_final}, it is also an approximation of $2^{-2\alpha}(p^*)^{1-\alpha}P_{\alpha}(\mathbf{p})$ within multiplicative error $\frac{L\epsilon}{L/2} = 2\epsilon$. Therefore, $2^{2\alpha}(p^*)^{\alpha-1}\sum_{i=1}^n p_iS_0(\sqrt{p_i})^2$ is a $2\epsilon$ multiplicative approximation of $P_{\alpha}(\mathbf{p})$. Therefore, $2^{2\alpha}(p^*)^{\alpha-1}\tilde{p}$ is an $(1+2\epsilon)(1+\epsilon)-1\le 5\epsilon$ multiplicative approximation of $P_{\alpha}(\mathbf{p})$. We can rescale $\epsilon$ to $\frac{1}{5}\epsilon$ so that we can obtain an $\epsilon$-multiplicative approximation of $P_{\alpha}(\mathbf{p})$. 
    
    \vspace{3mm}
    \noindent
    \textbf{Complexity.}
    Now we compute the query complexity of the above algorithm. 
    First let us compute $t_j$ defined in \prop{correctness_final}. For $t_j$, we have
    \begin{align}
        t_j = &\frac{2^{j}}{\beta_0}\log(\frac{m_0}{\epsilon L}) + \sum_{k = 1}^{j}\deg(S_k)\nonumber\\
        =&\widetilde{\mathcal{O}}\left(\frac{2^j}{\beta_0} + \sum_{k = 1}^{j} \frac{2^k}{\beta_0}\right)= \widetilde{\mathcal{O}}\left(\frac{2^j}{\beta_0}\right),
    \end{align}
    where the second equation comes from $\deg(S_j) = \widetilde{\mathcal{O}}\left(\frac{1}{\nu_j}\right)$ given in \lem{scaled_poly_approx_with_small_value_at_zero}.

Let $\varphi_j = 2^{-j}\beta_0$ for $j = 1,\ldots,m_0-1$ and $\varphi_{m_0}=0$ following the definition in \prop{correctness_final}. Then from \prop{correctness_final}, the complexity of the algorithm we construct is
\begin{align}
  \label{eq:complexity_final_first}
    \widetilde{\mathcal{O}}\biggl(\frac{1}{\epsilon}\biggl(t_{m_0} +& \sqrt{\epsilon} \sum_{j=1}^{m_0} t_j \nonumber\\+& \frac{\sqrt{\sum_{j = 1}^{m_0}\sum_{i:\sqrt{p_i}\in[\varphi_{j},\varphi_{j-1})}p_i t_{j+1}^2}}{\sqrt{\sum_{i = 1}^np_iS_0(\sqrt{p_i})^2}} \biggr)\biggr) 
\end{align}

Let $Q_j= \{i: 2^{-j}\beta_0 \le \sqrt{p_i} \ge 2^{-j+1}\beta_0\}$, then we have 
\begin{align}
  \label{eq:averge_query_complexity}
  &\frac{\sqrt{\sum_{j = 1}^{m_0}\sum_{i:\sqrt{p_i}\in[\varphi_{j},\varphi_{j-1})}p_i t_{j+1}^2}}{\sqrt{\sum_{i = 1}^np_iS_0(\sqrt{p_i})^2}}\nonumber\\
  \le&\frac{1}{\sqrt{L}}\sqrt{\sum_{j=1}^{m_0-1}|Q_j|(2^{-j+1}\beta_0)^2(\frac{2^{j+1}}{\beta_0})^2 + \sum_{i\in Q_{m_{0}}}p_i\frac{1}{\varphi_{m_0-1}^2}} \nonumber\\
  = &\mathcal{O}\Bigl(\sqrt{\frac{n}{L}}\Bigr)
\end{align}
Subtitute it into \eq{complexity_final_first}, we get the total query complexity of the algorithm
\begin{align}
  &\widetilde{\mathcal{O}}\left(\frac{1}{\epsilon}\left(\frac{1}{\nu_{m_0}} + \sqrt{\frac{n}{L}}\right)\right) \nonumber\\
  =&\widetilde{\mathcal{O}}\left(\frac{1}{\epsilon}\left(\frac{n^{\frac{1}{2\alpha}}}{\sqrt{p^*}\epsilon^{\frac{1}{2\alpha}}} + \frac{\sqrt{n}}{\sqrt{p^*}}\right)\right) \nonumber\\
  =&\mathcal{O}\left(\frac{n^{\frac{1}{2\alpha}}}{(P_{\alpha}(\mathbf{p}))^{\frac{1}{2\alpha}}\epsilon^{1+\frac{1}{2\alpha}}}+\frac{\sqrt{n}}{(P_{\alpha}(\mathbf{p}))^{\frac{1}{2\alpha}}\epsilon}\right), \label{eq:case-dependent-large-alpha}
\end{align}
where the first equation comes from
\begin{align}\label{eq:nu-m-0}
\nu_{m_0} = &2^{-m_0+1}\beta_0 = 2^{-\lceil \log(\frac{\beta_0}{\nu_0})\rceil}\beta_0 = \Theta(\nu_0) \nonumber \\
= & \Theta\Bigl(\left(\frac{2^{2\alpha}(p^*)^{\alpha}\epsilon}{5n}\right)^{\frac{1}{2\alpha}}\Bigr) = \Theta\Bigl(\frac{\sqrt{p^*} \epsilon^{\frac{1}{2\alpha}}}{n^{\frac{1}{2\alpha}}}\Bigr)
\end{align}
and $L = \Theta(p^*)$, and the fourth equation comes from $p^* = \min(\frac{1}{a}P,1)^{\frac{1}{\alpha}}$ and $P = \Theta(P_{\alpha}(\mathbf{p}))$.

In the worst case that $P_{\alpha}(\mathbf{p}) = n^{1-\alpha}$, the complexity bound becomes $ \widetilde{\mathcal{O}}\left(\frac{n^{1-\frac{1}{2\alpha}}}{\epsilon} + \frac{\sqrt{n}}{\epsilon^{1+\frac{1}{2\alpha}}}\right)$.
\end{proof}

\begin{remark}
Note that \eq{case-dependent-large-alpha} above established a case-dependent bound of estimating $H_{\alpha}(\mathbf{p})$ given a rough estimation in advance. The requirement of the rough estimation can be removed during the analysis in the next subsection. 
\end{remark}

\subsection{Estimate $H_{\alpha}(\mathbf{p})$ by annealing}
\noindent
We now apply the annealing method in \prop{annealing} to remove the requirement of $P,a,b$ in \lem{large_alpha_given_P}.
\begin{theorem}\label{thm:main-alpha-large}
  For any $\alpha > 1$, there exists an algorithm $\mathcal{A}$ such that for any $\delta \in (0,1)$, $\epsilon\in(0,1)$ and probability distribution $\mathbf{p}$ on $[n]$, $\mathcal{A}$ can estimate $H_{\alpha}(\mathbf{p})$ to within additive error $\epsilon$
  with success probability at least $1-\delta$ using $ \widetilde{\mathcal{O}}\left(\frac{n^{1-\frac{1}{2\alpha}}}{\epsilon} + \frac{\sqrt{n}}{\epsilon^{1+\frac{1}{2\alpha}}}\right)$ calls to $U_{\mathrm{pure}}$ and $U_{\mathrm{pure}}^{\dagger}$ in \defn{pure-state-preparation}.
\end{theorem}
\begin{proof}
Let the parameters $l$ and function sequence $\{g_k\}$ be such that $l := \lceil \frac{\ln(\alpha)}{\ln(1+1/\ln(n))}\rceil$ and $g_k(x) := x^{\alpha(1+\frac{1}{\ln(n)})^{k-l}-1}$ for $k = 1,\ldots,l$. Note that $g_k(x)$ is monotonically increasing on $[0,1]$, $xg_k(x) = x^{\alpha(1+\frac{1}{\ln(n)})^{k-l}}$ is convex, $g_k(0) = 0$, and $g_k(1) = 1$ for $k = 1,\ldots,l$.

We now prove that \eq{annealing_condition} hold for $f_k(\mathbf{p}) = \sum_{i = 1}^n p_i g_k(p_i) = P_{\alpha(1+\frac{1}{\ln(n)})^{k-l}}(\mathbf{p})$.

First, we introduce the following lemma to connect the value of $f_k(\mathbf{p})$ to $f_{k+1}(\mathbf{p})$.

\begin{lemma}[{\cite[Lemma 5.3]{li2019entropy}}]
  For any distribution $\mathbf{p} = \left(p_{i}\right)_{i=1}^{n}$ and $0<\alpha_{1}<\alpha_{2}$, we have
\begin{align}
    \left(\sum_{i \in[n]} p_{i}^{\alpha_{2}}\right)^{\frac{\alpha_{1}}{\alpha_{2}}} \leq \sum_{i \in[n]} p_{i}^{\alpha_{1}} \leq n^{1-\frac{\alpha_{1}}{\alpha_{2}}}\left(\sum_{i \in[n]} p_{i}^{\alpha_{2}}\right)^{\frac{\alpha_{1}}{\alpha_{2}}}.
\end{align}
\end{lemma}

Specifically, for $f_{k+1} = P_{\alpha_2}(\mathbf{p})$ and $f_{k} = P_{\alpha_1}(\mathbf{p})$, we have  $\frac{\alpha_1}{\alpha_2} = \frac{1}{1+\frac{1}{\ln(n)}}$, and
\begin{equation}
  \label{eq:recursive_eqn}
  f_{k+1}^{\frac{1}{1+\frac{1}{\ln(n)}}} \le f_{k} \le n^{\frac{1}{1+\ln(n)}}f_{k+1}^{\frac{1}{1+\frac{1}{\ln(n)}}}\le ef_{k+1}^{\frac{1}{1+\frac{1}{\ln(n)}}}.
\end{equation}

Since  $\alpha (1+\frac{1}{\ln(n)})^{-l} = \alpha (1+\frac{1}{\ln(n)})^{\lceil \frac{\ln(\alpha)}{\ln(1+1/\ln(n))}\rceil} \le  1+\frac{1}{\ln(n)}$, we have $f_1(\mathbf{p}) = P_{\alpha (1+\frac{1}{\ln(n)})^{-l}}(\mathbf{p}) \ge n^{1-\alpha (1+\frac{1}{\ln(n)})^{-l}} \ge n^{1-(1+1/\ln(n))}  = \frac{1}{e}$. Then we have
\begin{align}
\label{eq:annealing_equation_1}
    \frac{\max_{\mathbf{q}\in\Delta^n}f_1(\mathbf{q})}{\min_{\mathbf{q}\in\Delta^n}f_1(\mathbf{q})} \le \frac{1}{1/e} = e.
\end{align}

For any distributions $\mathbf{q}$, $\mathbf{r}$ on $[n]$, and $k \in [l-1]$ such that $f_k(\mathbf{r})/f_k(\mathbf{q})\in[\frac{3}{4}, \frac{5}{4}]$, from \eq{recursive_eqn}, we have
\begin{align}
    f_{k+1}(\mathbf{q}) &\ge\left(\frac{1}{e}f_k(\mathbf{q})\right)^{1+\frac{1}{\ln n}}\ge \left(\frac{1}{e}\frac{4}{5}f_k(\mathbf{r})\right)^{1+\frac{1}{\ln n}}, \\
    f_{k+1}(\mathbf{q}) &\le f_k(\mathbf{q})^{1+1/\ln(n)}\le \left(\frac{4}{3}f_k(\mathbf{r})\right)^{1+1/\ln(n)}.
\end{align}

Therefore, we have
\begin{align}
\label{eq:annealing_equation_2}
&\max_{\mathbf{r} \in \Delta^n}\frac{\max_{\mathbf{q}\in \Delta^n, f_k(\mathbf{r})/f_k(\mathbf{q})\in [\frac{3}{4}, \frac{5}{4}]} f_{k+1}(\mathbf{q}) }{\min_{\mathbf{q}\in \Delta^n, f_k(\mathbf{r})/f_k(\mathbf{q})\in [\frac{3}{4}, \frac{5}{4}]} f_{k+1}(\mathbf{q})} \nonumber \\
\le & \max_{\mathbf{r} \in \Delta^n} \Bigl(\frac{\frac{4}{3}f_k(\mathbf{r})}{\frac{1}{e}\frac{4}{5}f_k(\mathbf{r})}\Bigr)^{1+1/\ln(n)} \\
= & \Bigl(\frac{5e}{3}\Bigr)^{1+1/\ln(n)} \le 4e^2
\end{align}
for all $k \in [l-1]$.

Setting $\epsilon_k := \frac{1}{4}$ for all $k \in[l-1]$, $\epsilon_l := \min(\frac{1}{2},\frac{(\alpha-1)\epsilon}{2})$, and $c:=4e^2$ in \prop{annealing}, from \eq{annealing_equation_1} and \eq{annealing_equation_2}, we can infer that \eq{annealing_condition} is satisfied.

For all $k\in[l]$, setting $\alpha:=\alpha(1+\frac{1}{\ln(n)})^{k-l}$ in \lem{large_alpha_given_P}, from  \lem{large_alpha_given_P}, there exists an algorithm $\mathcal{A}_k$ which can estimate $f_k(\mathbf{p})$ to within multiplicative error $\epsilon_k$ with success probability at least $1-\frac{\delta}{l}$ using $Q_k = \widetilde{\mathcal{O}}\left(\frac{n^{1-\frac{1}{2\alpha}}}{\epsilon_k} + \frac{\sqrt{n}}{\epsilon_k^{1+\frac{1}{2\alpha}}}\right)$ calls to $U_{\mathrm{pure}}$ and $U_{\mathrm{pure}}^{\dagger}$. These $\mathcal{A}_k$ satisfy the conditions in \prop{annealing}, so we can construct an algorithm to estimate $P_{\alpha}(\mathbf{p})$ to within multiplicative error $\epsilon_l = \Theta(\epsilon)$ using
\begin{align}
    &\sum_{k=1}^l Q_k \nonumber \\
    = & \sum_{k=1}^l \widetilde{\mathcal{O}}\biggl(\frac{n^{1-\frac{1}{2\alpha(1+\frac{1}{\ln n})^{k-l}}}}{\epsilon_k} + \frac{\sqrt{n}}{\epsilon_k^{1+\frac{1}{2\alpha(1+\frac{1}{\ln n})^{k-l}}}}\biggr) \nonumber\\
    = & \widetilde{\mathcal{O}}\biggl((l-1)\biggl(\frac{n^{1-\frac{1}{2\alpha}}}{1/4} + \frac{\sqrt{n}}{(\frac{1}{4})^{1+\frac{1}{2\alpha}}}\biggr) + \frac{n^{1-\frac{1}{2\alpha}}}{\epsilon} + \frac{\sqrt{n}}{\epsilon^{1+\frac{1}{2\alpha}}}\biggr)\nonumber\\
    = &\widetilde{\mathcal{O}}\biggl(\frac{n^{1-\frac{1}{2\alpha}}}{\epsilon} + \frac{\sqrt{n}}{\epsilon^{1+\frac{1}{2\alpha}}}\biggr)\label{eq:complexity-bound-large-alpha}
\end{align}
where the second equation comes from $n^{1-\frac{1}{2\alpha(1+\frac{1}{\ln(n)})^{k-l}}}\le n^{1-\frac{1}{2\alpha}}$ and $\epsilon_k = \frac{1}{4}$ for all $k\in[l-1]$, and the third equation comes from $l = \mathcal{O}(\ln(\alpha))$ and can be omitted in $\widetilde{\mathcal{O}}(\cdot)$.

Denote the estimate of $P_{\alpha}(\mathbf{p})$ by $\widetilde{P}_{\alpha}(\mathbf{p})$, we have $\frac{\widetilde{P}_{\alpha}(\mathbf{p})}{P_{\alpha}(\mathbf{p})}\in[1-\epsilon_l,1+\epsilon_l]$. Let $\widetilde{H}_\alpha(\mathbf{p})=\frac{1}{1-\alpha}\log(\widetilde{P}_{\alpha}(\mathbf{p}))$, then we have
\begin{align}
    |\widetilde{H}_\alpha(\mathbf{p})-H_\alpha(\mathbf{p})|
    = & \frac{1}{\alpha-1}\left|\log(\widetilde{P}_{\alpha}(\mathbf{p}))-\log(P_{\alpha}(\mathbf{p}))\right| \nonumber \\
    = & \frac{1}{\alpha-1}\Bigl|\log\Bigl(\frac{\widetilde{P}_\alpha(\mathbf{p})}{P_\alpha(\mathbf{p})}\Bigr)\Bigr| \nonumber \\
    \le & \frac{2}{\alpha-1}\epsilon_l \le \epsilon,
\end{align}
where the first inequality comes from $|\log(1+x)| \le 2|x|$ for all $x > -\frac{1}{2}$, so $\widetilde{H}_\alpha(\mathbf{p})$ is an estimate of $H_\alpha(\mathbf{p})$ within additive error $\epsilon$.
\end{proof}


\section{R{\'e}nyi Entropy Estimation ($0<\alpha<1$)}\label{sec:alpha-small}
\subsection{Upper bound}

\noindent
In order to approximate $H_{\alpha}(\mathbf{p}) $ within a given additive error $\epsilon$, we need to approximate $P_{\alpha}(\mathbf{p})$ within multiplicative error $\mathcal{O}(\epsilon)$. Note that $P_{\alpha}(\mathbf{p}) = \sum_{i=1}^n p_i ^{\alpha} = \sum_{i=1}^n p_i p_i^{\alpha-1}$ and $\alpha-1 < 1$ for $\alpha < 1$, so we first construct a series of polynomials $S$, such that for any constant $d > 0$, there exists polynomial $S$ such that $\sum_{i=1}^n p_iS(\sqrt{p_i})^2$ is an $\mathcal{O}(\epsilon)$ multiplicative approximation of $\frac{(d\epsilon/n)^{\frac{1}{\alpha}-1}}{4} P_{\alpha}(\mathbf{p})$. 

\begin{lemma}
\label{lem:small_alpha_approx}
  For any $\epsilon\in(0,1)$, $\alpha \in (0,1)$, and constants $d,d' > 0$, the odd polynomial $P'$ in \lem{poly-approx-nega-power} with parameters $(\delta, \varepsilon, c)$ to be $\delta := (\frac{d\epsilon}{n})^{{\frac{1}{2\alpha}}},c := 1-\alpha,\varepsilon := d'\epsilon\delta^{2c}$ has $\deg(S) = \widetilde{\mathcal{O}}\left(\frac{n^{\frac{1}{2\alpha}}}{\epsilon^{\frac{1}{2\alpha}}}\right)$, and satisfies \eq{qsvt_condition} in \thm{qsvt} and
  \begin{align}
      \sum_{i=1}^n \Bigl|p_iS(\sqrt{p_i})^2-\frac{\delta^{2c}}{4}p_i^{\alpha}\Bigr|\le \left(\frac{5}{4}d + 2d'\right)\epsilon \delta^{2c}.
  \end{align}
\end{lemma}
\begin{proof}
    \lem{poly-approx-nega-power} implies that $S$ is an odd polynomial and $|S(x)| \le 1$ for all $x\in[-1,1]$, so $S$ satisfies \eq{qsvt_condition} in \thm{qsvt}.
    From \lem{poly-approx-nega-power}, $S$ satisfies that
    \begin{align}
    \label{eq:approximation_small_alpha}
        &\forall x\in \left[\delta,1\right]: \Bigl|S(x)-\frac{\delta^{c}}{2}x^{\alpha-1}\Bigr| \le d'\epsilon\delta^{2c}.
    \end{align}
    From \lem{poly-approx-nega-power}, the degree of $S$ is $\deg(S) = \mathcal{O}\left(\frac{\max[1,c]}{\delta}\log(\frac{1}{\epsilon})\right) = \widetilde{\mathcal{O}}\left(\frac{n^{\frac{1}{2\alpha}}}{\epsilon^{\frac{1}{2\alpha}}}\right)$.

  For $i$ such that $\sqrt{p_i} \le \delta$, we have
    \begin{align}
      &\sum_{\sqrt{p_i} \le \delta} \Bigl|p_iS(\sqrt{p_i})^2-\frac{\delta^{2c}}{4}p_i^{\alpha}\Bigr| \nonumber \\
      \le & \sum_{\sqrt{p_i} \le \delta} \left(\Bigl|p_iS(\sqrt{p_i})^2\Bigr| + \Bigl|\frac{\delta^{2c}}{4}p_i^{\alpha}\Bigr|\right) \nonumber\\
      \le & \Bigl(\sum_{\sqrt{p_i} \le \delta}p_i + \frac{\delta^{2c}}{4}p_i^{\alpha}\Bigr)                \nonumber\\
      \le & \sum_{\sqrt{p_i} \le \delta} \left(\delta^{2} + \frac{\delta^{2}}{4}  \right)                   \nonumber \\
      = & \sum_{\sqrt{p_i} \le \delta} \frac{5}{4}\delta^{2\alpha}\delta^{2c} 
      \le  \frac{5}{4}d\epsilon\delta^{2c},
    \end{align}
    \label{eq:little_alpha_little_delta}
  where the second inequality comes from $|S(x)|\le 1$ for all $x\in[-1,1]$, the third inequality comes from $\sqrt{p_i} \le \delta$ and $c = 1-\alpha$,
  and the last inequality comes from $n \delta^{2\alpha} = d\epsilon$.

  For $i$ such that $\sqrt{p_i} > \delta$, we have

    \begin{align}
      &\sum_{\sqrt{p_i} > \delta} \Bigl|p_iS(\sqrt{p_i})^2-\frac{\delta^{2c}}{4}p_i^{\alpha}\Bigr| \nonumber \\
      = &\sum_{\sqrt{p_i} > \delta} p_i\Bigl|S(\sqrt{p_i})-\frac{\delta^{c}}{2}(\sqrt{p_i})^{\alpha-1}\Bigr|\nonumber \\
      \cdot &\Bigl|S(\sqrt{p_i})+\frac{\delta^{c}}{2}(\sqrt{p_i})^{\alpha-1}\Bigr|
      \nonumber\\ \le & 2d'\epsilon\delta^{2c}\sum_{\sqrt{p_i} > \delta} p_i \le  2d'\epsilon\delta^{2c},
    \end{align}
    \label{eq:little_alpha_large_delta}
  where the first inequality comes from \eq{approximation_small_alpha}, $|S(x)|\le 1$ for all $x\in[-1,1]$, and $\frac{\delta^{c}}{2}(\sqrt{p_i})^{\alpha-1}\le 1$ for all $i\in[n]$.

  Combining \eq{little_alpha_large_delta} and \eq{little_alpha_little_delta}, we have
  \begin{align}
      \sum_{i=1}^n \Bigl|p_iS(\sqrt{p_i})^2-\frac{\delta^{2c}}{4}p_i^{\alpha}\Bigr| \le (\frac{5}{4}d + 2d')\epsilon \delta^{2c},
  \end{align}
  which completes the proof.
\end{proof}

As a result, we can give an algorithm for estimating $H_{\alpha}(\mathbf{p})$ to within a given additive error $\epsilon$ with high probability as follows. 
\begin{theorem}\label{thm:main-alpha-small}
  For any $\alpha \in(0,1)$, there exists an algorithm $\mathcal{A}$ such that for any $\delta \in (0,1)$, $\epsilon\in(0,1)$ and probability distribution $\mathbf{p}$ on $[n]$, $\mathcal{A}$ can estimate $H_{\alpha}(\mathbf{p})$ to within additive error $\epsilon$
  with success probability at least $1-\delta$ using $ \widetilde{\mathcal{O}}\left(\frac{n^{\frac{1}{2\alpha}}}{\epsilon^{\frac{1}{2\alpha}+1}}\right)$ calls to $U_{\mathrm{pure}}$ and $U_{\mathrm{pure}}^{\dagger}$ in \defn{pure-state-preparation}. 
\end{theorem}

\begin{proof}
We will first construct such an algorithm $\mathcal{A}$ using \prop{correctness_final}, prove its correctness, and then compute its query complexity.

\noindent
\textbf{Construction and correctness.}
  Let $\epsilon_0 = \min(\frac{1}{2},\frac{(1-\alpha)\epsilon}{4})$, $\delta' = \left(\frac{\epsilon_0}{40n}\right)^{\frac{1}{2\alpha}}$. 
  Before constructing $S$ and $S_j$ in \prop{correctness_final}, we first define the number of stages of our variable-stopping-time quantum algorithm $m_0:=\lceil \log(\frac{1}{\delta'})\rceil+1$, and $\delta_j = 2^{-j}$ for $j = 1,\ldots, m_0-1$, $\delta_{m_0} = \delta_{m_0-1} = 2^{-m_0+1}$.

    Let $S_0$ be the polynomial $S$ in \lem{poly-approx-nega-power} with parameters $(\delta, \varepsilon, c)$ to be $\delta := \delta_{m_0},c := 1-\alpha,\varepsilon := \epsilon_{0}\delta'^{2c}$.
    
    Let $L := \frac{1}{8}\delta_{m_0}^{2c}$, and we will prove that $L$ is an lower bound of $\sum_{i=1}^n p_i S_0(\sqrt{p_i})^2$ later.

    Let $P_j$ for $j=1,\ldots, m_0$ be the odd polynomial $P'$ in \lem{poly-approx-nega-power} with parameters $(\delta, \varepsilon, c)$ to be $\delta := \delta_j,c := 1-\alpha,\varepsilon := \frac{1}{64}\delta_{m_0}^{2c}\epsilon_{0}$. From \lem{poly-approx-nega-power}, we have
    \begin{align}
    \label{eq:equation_P_j}
        \Bigl|P_j(x)-\frac{\delta_j^{c}}{2}x^{\alpha-1}\Bigr| \le \frac{1}{64}\delta_{m_0}^{2c}\epsilon_{0}\forall x\in [\delta_j,1],
    \end{align}
    and $\deg(P_j) = \mathcal{O}\left(\frac{\max [1, c]}{\delta} \log \left(\frac{1}{\varepsilon}\right)\right) = \widetilde{\mathcal{O}}(2^j)$. 

    Now we set the parameters $(S, \beta,L,m,\{S_j\mid j\in[m]\})$ of \algo{main_algo_improved} to be $S:=S_0$, $\beta := 1$, $L:=L$, $m:=m_0$, $S_j := \left(\frac{\delta_{m_0}}{\delta_j}\right)^{c}P_j$ for $j = 1,\ldots,m_0$, and then prove that these parameters satisfy the conditions in \prop{correctness_final}.
    \begin{itemize}[leftmargin=*]
    \item For $\beta$, we have $\sqrt{p_i} \le 1 =\beta$ for all $i\in[n]$.
        \item For $L$, we need to prove that is is a lower bound of $\sum_{i=1}^n p_i S_0(\sqrt{p_i})^2$.
  Note that $\delta_{m_0} = 2^{-m_0+1} = 2^{-\lceil \log(\frac{1}{\delta'})\rceil} \le \delta'\le \left(\frac{\epsilon_{0}}{40n}\right)^{\frac{1}{2\alpha}}$. Let the constants $d,d'$ in \lem{small_alpha_approx} be such that $d = \delta_{m_0}^{2\alpha}\frac{n}{\epsilon_{0}} \le \frac{1}{40}$, $d' = \frac{1}{64}$, then the parameters $(\delta = \delta_{m_0} = \left(\frac{d\epsilon_0}{n}\right)^{\frac{1}{2\alpha}}, c=1-\alpha,\varepsilon=\frac{1}{64}\epsilon_{0} \delta_{m_0}^{2c}= d'\epsilon_0\delta_{m_0}^{2c})$ of $S_0$ satisfy the conditions in \lem{small_alpha_approx}. From \lem{small_alpha_approx}, we have
  \begin{align}
    \label{eq:small_alpha_final_approx}
    \sum_{i=1}^n \Bigl|p_iS_0(\sqrt{p_i})^2-\frac{\delta_{m_0}^{2c}}{4}p_i^{\alpha}\Bigr| \le & (\frac{5}{4}d + 2d')\epsilon_{0} \delta_{m_0}^{2c} \nonumber \\ 
    \le & \frac{1}{16}\epsilon_{0} \delta_{m_0}^{2c} = \frac{1}{2}L\epsilon_{0}.
  \end{align}
From \eq{small_alpha_final_approx}, we have
    \begin{align}
        \sum_{i=1}^n p_iS_0(\sqrt{p_i})^2 &\ge  \frac{\delta_{m_0}^{2c}}{4}P_{\alpha}(\mathbf{p}) - \frac{1}{16} \epsilon_{0} \delta_{m_0}^{2c} \nonumber\\ 
        &\ge \frac{1}{4}\delta_{m_0}^{2c}-\frac{1}{16}\delta_{m_0}^{2c}\nonumber\\
        &\ge  \frac{1}{8}\delta_{m_0}^{2c} = L,
    \end{align}
    where the second inequality comes from $P^{\alpha}(\mathbf{p}) \ge 1$ for $\alpha\in(0,1)$, and $\epsilon_{0} \in (0,1)$.
    \item For $S_j$, they are odd polynomial and satisfy $|S_j(x)| \le 1$ for all $x\in[-1,1]$, which meet the requirements in \thm{qsvt}.
    Note that the parameters of $S_0$ in \lem{poly-approx-nega-power} is the same as the parameters of $S_{m_0}$, so we have $S_0 = S_{m_0}$. For any $x\in [\delta_j,1]$, $S_j$ satisfies
    \begin{align}
        &|S_j(x)-S_0(x)| \nonumber \\
        = & |S_j(x)-S_{m_0}(x)| \nonumber\\
        \le & \Bigl|\Bigl(\frac{\delta_{m_0}}{\delta_j}\Bigr)^{c}P_j(x)-\frac{\delta_{m_0}^{c}}{2}x^{\alpha-1}\Bigr| + \Bigl|P_{m_0}(x)-\frac{\delta_{m_0}^{c}}{2}x^{\alpha-1}\Bigr|\nonumber \\
        \le & \Bigl(\frac{\delta_{m_0}}{\delta_j}\Bigr)^{c}\Bigl|P_j(x)-\frac{\delta_{j}^{c}}{2}x^{\alpha-1}\Bigr| + \Bigl|P_{m_0}(x)-\frac{\delta_{m_0}^{c}}{2}x^{\alpha-1}\Bigr|\nonumber\\
        \le & \Bigl(1+\left(\frac{\delta_{m_0}}{\delta_j}\right)^{c}\Bigr)(\frac{1}{64}\delta^{2c}_{m_0}\epsilon_{0}) \le L\epsilon_{0} \label{eq:equation_S_j_2_small_alpha},
    \end{align}
    where the third inequality comes from \eq{equation_P_j} and $\delta_{m_0} < \delta_j$ for $j < m_0$.
    From \eq{equation_S_j_2_small_alpha}, we can infer that $\text{for any } x\in [ 2^{-j},1]:|S_j(x)-S_0(x)| \le L\epsilon_{0}$, which meets the requirements of \eq{S_j_condition}.
    \end{itemize}

    Therefore, the parameters we set are valid for \prop{correctness_final}, so \algo{main_algo_improved} with the same parameters and input $(\epsilon_{0},\delta)$ can estimate $\sum_{i=1}^n p_i S_0(\sqrt{p_i})^2$ to within multiplicative error $\epsilon_{0}$ within success probability at least $1-\delta$. Denote the estimate by $\tilde{p}$, we have
    \begin{align}
        &|4\delta^{-2c}_{m_0}\tilde{p}-P_{\alpha}(\mathbf{p})| \nonumber \\
        \le & |4\delta^{-2c}_{m_0}\tilde{p}-4\delta^{-2c}_{m_0}\sum_{i=1}^n p_i S_0(\sqrt{p_i})^2| \nonumber \\
        & + |4\delta^{-2c}_{m_0}\sum_{i=1}^n p_i S_0(\sqrt{p_i})^2-P_{\alpha}(\mathbf{p})|\nonumber\\
        \le & 4\delta^{-2c}_{m_0}(\sum_{i=1}^n p_i S_0(\sqrt{p_i})^2)\epsilon_{0} +2\delta^{-2c}_{m_0}L\epsilon_{0}\nonumber\\
        \le & 4\delta^{-2c}_{m_0}(\frac{\delta^{2c}_{m_0}}{4}P_{\alpha}(\mathbf{p})\epsilon_{0} + \frac{1}{2}L\epsilon_{0}^2)+\frac{1}{4}\epsilon_{0}\nonumber\\
        = &P_{\alpha}(\mathbf{p})\epsilon_{0} + \frac{1}{4}(\epsilon_{0}^2 + \epsilon_{0})\le 2P_{\alpha}(\mathbf{p})\epsilon_{0},
    \end{align}
    with success probability at least $1-\delta$, where the second inequality and the third inequality come from \eq{small_alpha_final_approx}, and the last inequality comes from $P_{\alpha} \ge 1$ for any $\alpha \in (0,1)$.  Let $\widetilde{P}_{\alpha}(\mathbf{p}) := 4\delta^{-2c}_{m_0}\tilde{p}$, and then we have $\frac{\widetilde{P}_{\alpha}(\mathbf{p})}{P_{\alpha}(\mathbf{p})}\in[1-2\epsilon_{0},1+2\epsilon_{0}]$. Let $\widetilde{H}_\alpha(\mathbf{p})=\frac{1}{1-\alpha}\log(\widetilde{P}_{\alpha}(\mathbf{p}))$, then we have
    \begin{align}
        |\widetilde{H}_\alpha(\mathbf{p})-H_\alpha(\mathbf{p})| = & \frac{1}{1-\alpha}\Bigl(\log(\widetilde{P}_{\alpha}(\mathbf{p}))-\log(P_{\alpha}(\mathbf{p}))\Bigr)\nonumber \\ 
        = & \frac{1}{1-\alpha}\Bigl|\log\Bigl(\frac{\widetilde{P}_\alpha(\mathbf{p})}{P_\alpha(\mathbf{p})}\Bigr)\Bigr|\nonumber \\ \le & \frac{4}{1-\alpha}\epsilon_{0} = \epsilon,
    \end{align}
    where the first inequality comes from $|\log(1+x)| \le 2|x|$ for all $x > -\frac{1}{2}$.
    
    \vspace{3mm}
    \noindent
    \textbf{Complexity.}
    Now we compute the query complexity of the above algorithm. 
    First, let us compute $t_j$ defined in \prop{correctness_final}. For $t_j$, we have
    \begin{align}
        t_j = & 2^j\log(\frac{m_0}{\epsilon L}) + \sum_{k = 1}^{j}\deg(S_k) \nonumber \\
        = & \widetilde{\mathcal{O}}\left(2^j + \sum_{k = 1}^{j} 2^k\right)= \widetilde{\mathcal{O}}\left(2^j\right),
    \end{align}
    where the second equation comes from $\deg(S_j) =
    \deg(P_j) = \widetilde{\mathcal{O}}\left(2^j\right)$. 

    Let $\varphi_j = 2^{-j}$ for $j = 1,\ldots,m_0-1$ and $\varphi_{m_0}=0$ following the definition in \prop{correctness_final}.
    From \prop{correctness_final}, the query complexity of the above algorithm is 
    \begin{align}
    &\widetilde{\mathcal{O}}\biggl(\frac{1}{\epsilon}\biggl(t_{m_0} + \sqrt{\epsilon} \sum_{j=1}^{m_0} t_j \nonumber \\ 
    & \qquad \quad+ \frac{\sqrt{\sum_{j = 1}^{m_0}\sum_{i:\sqrt{p_i}\in[\varphi_{j},\varphi_{j-1})}p_i t_{j+1}^2}}{\sqrt{\sum_{i = 1}^np_iS(\sqrt{p_i})^2}}\biggr)\biggr) \nonumber\\
    =&\widetilde{\mathcal{O}}\left(\frac{1}{\epsilon}\left(\frac{1}{\delta_{m_0}} + \frac{\sqrt{n + n}}{\delta_{m_0}^{c}}\right)\right) \nonumber\\
    =&\widetilde{\mathcal{O}}\left(\frac{1}{\epsilon}\left(\frac{n^{\frac{1}{2\alpha}}}{\epsilon^{\frac{1}{2\alpha}}} + \frac{\sqrt{n}}{\left(\frac{\epsilon}{n}\right)^{(1-\alpha)/2\alpha}}\right)\right) \nonumber\\
    =&\widetilde{\mathcal{O}}\left(\frac{n^{\frac{1}{2\alpha}}}{\epsilon^{\frac{1}{2\alpha} + 1}}+\frac{n^{\frac{1}{2\alpha}}}{\epsilon^{\frac{1}{2\alpha}+\frac{1}{2}}}\right) 
    = \widetilde{\mathcal{O}}\left(\frac{n^{\frac{1}{2\alpha}}}{\epsilon^{\frac{1}{2\alpha}+1}}\right), 
\end{align}
where the first equation can be derived in a similar way to \eq{averge_query_complexity}.
\end{proof}

\subsection{Lower bound}\label{sec:small-alpha-lower}
\noindent
The Hellinger distance between two discrete probability distributions $p$ and $q$ is defined as 
$\mathrm{d_{H}}(\mathbf{p},\mathbf{q}):=\sqrt{\sum_{i=1}^{n}(\sqrt{p_{i}}-\sqrt{q_{i}})^{2}/2}.
$
In \cite{belovs2019quantum}, they give a lower bound for the query complexity of distinguishing two distributions as follows. 
\begin{lemma}[{\cite[Claim 5]{belovs2019quantum}}]\label{lem:Hellinger}
Quantum query complexity of distinguishing probability distributions $\mathbf{p}$ and $\mathbf{q}$ with pure-state preparation oracle in \defn{pure-state-preparation} is $\Theta\bigl(1/\mathrm{d_H}(\mathbf{p},\mathbf{q})\bigr)$.
\end{lemma}
Then we can the give the following lower bound for estimating $H_{\alpha}(\mathbf{p})$ with pure-state preparation oracle $U_{\mathrm{pure}}$ and $U_{\mathrm{pure}}^{\dagger}$. 

\begin{theorem}\label{thm:lower}
    For any constant $\alpha\in (0,1)$, $n \ge 1+2^{1/(1-\alpha)}$, and $\epsilon\in(0,\frac{1}{2})$, any algorithm that can estimate $H_{\alpha}(\mathbf{p})$ to within additive error $\epsilon$ needs at least $\Omega\left(\frac{n^{1/2\alpha-1/2}}{\epsilon^{1/2\alpha}}\right)$ calls to $U_{\mathrm{pure}}$ and $U_{\mathrm{pure}}^{\dagger}$ in \defn{pure-state-preparation}.
\end{theorem}
\begin{proof}
For any $\epsilon\in (0,\frac{1}{2})$, $n > 1+2^{1/(1-\alpha)}$ and $\alpha \in (0,1)$, let $\delta = \left(\frac{4\epsilon}{(n-1)^{1-\alpha}}\right)^{\frac{1}{\alpha}} < 1$. 

Consider $\mathbf{p}=\left(1-\delta, \frac{\delta}{n-1}, \ldots, \frac{\delta}{n-1}\right)$ and $\mathbf{q}=(1,0, \ldots, 0)$. The Hellinger distance of $\mathbf{p}$ and $\mathbf{q}$ is 
\begin{align}
    \mathrm{d_{H}}(\mathbf{p}, \mathbf{q}) = & \sqrt{\frac{1}{2}\biggl((\sqrt{1-\delta}-1)^2+ (n-1)\biggl(\sqrt{\frac{\delta}{n-1}}\biggr)^2 \biggr)} \nonumber \\
    = & \sqrt{\frac{1}{2}\left(\Theta(\delta^2) + \delta\right)}
    = \Theta(\sqrt{\delta})
\end{align}
as $\delta \to 0$, where the second equation comes from $\sqrt{1-x} = 1-\Theta(x)$ as $x\to 0$. By Lemma \ref{lem:Hellinger}, we need $\Omega\left(\frac{1}{\sqrt{\delta}}\right)$ calls to $U_{\mathrm{pure}}$ and $U_{\mathrm{pure}}^{\dagger}$ to distinguish $\mathbf{p}$ and $\mathbf{q}$. 

Then we have
\begin{align}
    &|H_{\alpha}(\mathbf{q})-H_{\alpha}(\mathbf{p})| \nonumber \\
  = & \left|\frac{1}{1-\alpha}\log((1-\delta)^{\alpha} + \delta^{\alpha}(n-1)^{1-\alpha})-0\right| \nonumber\\
  \ge & \left|\frac{1}{1-\alpha}\log(1-\delta + \delta^{\alpha}(n-1)^{1-\alpha})\right| \nonumber\\
  = & \left| \frac{1}{1-\alpha} \log \left( 1 -\delta + 4\epsilon\right)\right|\nonumber\\
  \ge & \left| \frac{1}{1-\alpha} \log \left(1+ 2\epsilon\right)\right|\nonumber\\
  \ge & \frac{2\epsilon}{(1-\alpha)} \ge 2\epsilon,
\end{align}
where the first inequality is because $(1-x)^\alpha \ge 1-x$ as for $x\in(0,1)$ for any $\alpha\in(0,1)$, the second inequality is because $\delta \le \frac{4\epsilon}{(n-1)^{1-\alpha}}\le 2\epsilon$, and the third inequality is because $\log (1+x) \ge x$ for $x\in(0,1)$.

If we can estimate R{\'e}nyi entropy of $\mathbf{p}$ and $\mathbf{q}$ to within additive error $\epsilon$, we can distinguish distributions $\mathbf{p}$ and $\mathbf{q}$, which needs $\Omega(1/\sqrt{\delta})$ queries as proven above. Therefore, it requires $\Omega\left(\frac{1}{\sqrt{\delta}}\right) = \Omega\left(\frac{n^{1/2\alpha-1/2}}{\epsilon^{1/2\alpha}}\right)$ queries to $U_{\mathrm{pure}}$ and $U_{\mathrm{pure}}^{\dagger}$. 
\end{proof}

\textcolor{black}{We note that Acharya et al.~\cite{acharya2016estimating,acharya2019measuring} used the same distribution to prove lower bound of R{\'e}nyi entropy estimation in classical sampling model and quantum sampling model in \defn{q-sample-def}. This is because classically one need $\Theta(\frac{1}{d_{\mathrm{H}}(\mathbf{p},\mathbf{q})^2})$ samples to distinguish $\mathbf{p}$ and $\mathbf{q}$, so the hard instances in the quantum query model and the sampling model are the same. }

\subsection{More discussions about $\epsilon$ dependence}\label{sec:eps-dependence}
\noindent \textbf{$\epsilon$ dependence of estimating $H_{\alpha}(\mathbf{p})$ for $\alpha \in (0,1)$ in \cite{li2019entropy}.\quad}
Note that Belovs~\cite{belovs2019quantum} proved that the lower bound in \lem{Hellinger} of distinguishing probability distributions also holds with the oracle in \defn{discrete-quantum-query}. As a result, our lower bound $\Omega\left(\frac{n^{1/2\alpha-1/2}}{\epsilon^{1/2\alpha}}\right)$ in~\thm{lower} also holds with this oracle. However, a contradiction can be observed between the $\epsilon$ dependency of this lower bound and that of the upper bound in~\cite{li2019entropy} which uses $\widetilde{\mathcal{O}}\left(\frac{n^{1/\alpha-1/2}}{\epsilon^2}\right)$ calls to the oracles in \defn{discrete-quantum-query} and outputs an estimate of $H_{\alpha}(\mathbf{p})$ for $\mathbf{p}\in\Delta^n$ to within additive error $\epsilon$.

We suspect that there is an issue with Eq.~(V.46) in the journal version of~\cite{li2019entropy}. It follows the same proof as in Lemma 2, but in the proof of Lemma 2, the Taylor approximation of $\left(\sin \left(\left(\theta_i+\frac{l}{2^m}\right) \pi\right)\right)^{2(\alpha-1)}$ in Eq.~(V.8)   is not precise  for $\alpha < 1$ when $\theta_i + \frac{l}{2^m}$ is close to $0$ since $x^{2(\alpha-1)}$ diverges at $0$. Here we give a corrected analysis of the bias of the $\alpha$-R{\'e}nyi entropy estimator in~\cite{li2019entropy} when $\alpha \in (0,\frac{1}{2})$. Following the notation in~\cite{li2019entropy}, for $\alpha \in (0,\frac{1}{2})$ and each $i\in S_{j+1}$, in order to bound the Taylor approximation error, we need to treat $l\in(2^{m-1}\theta_i,2^m\theta_i)$ specially. Here we take $\theta_i = 2^j/2^m$ for simplicity, but the following equations hold for general $\theta_i$ when $i\in S_{j+1}$:
\begin{align}
    &p_i \mathbb{E}\left[\left|\tilde{p}_i^{\alpha-1}-p_i^{\alpha-1}\right|\right] \nonumber\\
    = &\mathcal{O}\biggl(\Bigl(\frac{2^j}{2^m}\pi\Bigr)^2\Bigl(\sum_{l= -(2^m-2^{j}),l\neq 0}^{2^{j-1}}+\sum_{l= 2^{j-1}}^{2^{j}}\Bigr)\biggr.\nonumber \\
    &\cdot\biggl.\left(\frac{1}{l^2}\left|\bigl(\sin\bigl(\bigl(\theta_i-\frac{l}{2^m}\bigr) \pi\bigr)\bigr)^{2(\alpha-1)}-(\sin (\theta_i \pi))^{2(\alpha-1)}\right|\right)\biggr) \nonumber\\
    = &\mathcal{O}\biggl(\Bigl(\frac{2^j}{2^m}\pi\Bigr)^2 \sum_{l=  -(2^m-2^{j}),l\neq 0}^{2^{j-1}}\frac{1}{l^2}\frac{|l|}{2^m}\Bigl(\frac{2^j}{2^m} \pi\Bigr)^{2 \alpha-3} \nonumber\\
    &\  +  \Bigl(\frac{2^j}{2^m}\pi\Bigr)^2 \sum_{l= 2^{j-1}}^{2^{j}-1}2^{-2j}\bigl(\sin\bigl(\bigl(\theta_i-\frac{l}{2^m}\bigr) \pi\bigr)\bigr)^{2\alpha-2}\biggr)\nonumber\\
    =&\mathcal{O}\biggl(\frac{m}{2^{2\alpha m}}2^{(2\alpha-1)j} + \sum_{r = 1}^{2^{j-1}}\frac{1}{2^{2m}}\Bigl(\frac{r}{2^m}\Bigr)^{2\alpha-2}\biggr)\nonumber\\
    =&\mathcal{O}\biggl(\frac{m}{2^{2\alpha m}}2^{(2\alpha-1)j} +\frac{m}{2^{2\alpha m}}\biggr),\label{eq:new_bias}
\end{align}
where the second equation comes from
\begin{align}
&\left|\left(\sin(\left(\theta_i-\frac{l}{2^m}\right)\pi)\right)^{2(\alpha-1)}-(\sin (\theta_i \pi))^{2(\alpha-1)}\right| \nonumber \\
\le & c\frac{|l|}{2^m}\left(\theta_i \pi\right)^{2 \alpha-3}
\end{align}
for $\frac{l}{2^m \theta_i} \le \frac{1}{2}$, and we replace $2^j-l$ with $r$ in the third equation. Note that the first term of \eq{new_bias} is the same as equation Eq.~(V.46) in the journal version of~\cite{li2019entropy}, but it is smaller than the second term, so we only need to set $m = \lceil \frac{1}{2\alpha}\log(\frac{\epsilon}{n}\log(\frac{\epsilon}{n}))\rceil$ so that
\begin{align}
\sum_{i=1}^{n}p_i \mathbb{E}\left[\left|\tilde{p}_i^{\alpha-1}-p_i^{\alpha-1}\right|\right] = \mathcal{O}\left(\frac{nm}{2^{2\alpha m}}\right)
\end{align}
is bounded by $\epsilon$. Therefore, the overall complexity of the algorithm in~\cite{li2019entropy} used to estimate $\alpha$-R{\'e}nyi entropy when $\alpha\in(0,\frac{1}{2})$ is 
\begin{align}
\widetilde{\mathcal{O}}\left(\frac{\sqrt{n^{\frac{1}{\alpha}-1}}}{\epsilon}\left(\frac{n}{\epsilon}\right)^{\frac{1}{2\alpha}}\right) = \widetilde{\mathcal{O}}\left(\frac{n^{\frac{1}{\alpha}-\frac{1}{2}}}{\epsilon^{1+\frac{1}{2\alpha}}}\right),
\end{align}
which has the same dependence on $\epsilon$ as that in  $\widetilde{\mathcal{O}}(\frac{n^{\frac{1}{2\alpha}}}{\epsilon^{1+\frac{1}{2\alpha}}})$ in our algorithm (and worse dependence in $n$ than that in our algorithm).
\\\\
\noindent \textbf{$\epsilon$ dependence of estimating $H_{\alpha}(\mathbf{p})$ for $\alpha \in (0,1)$ classically.\quad}
We also find that there might be an issue with the $\epsilon$ dependency of the classical upper bound $\mathcal{O}\left(\frac{n^{1 / \alpha}}{\epsilon^{1 / \alpha} \log n}\right)$ on estimating R{\'e}nyi entropy when $\alpha < 1$ in~\cite{acharya2016estimating}. Specifically, we suspect that the last two $o(1)$ terms in Eq.~(15) and Eq.~(18) of the arXiv version of~\cite{acharya2016estimating} are omitted, but according to Lemma 8, these two terms cannot be omitted unless they are $o(\delta)$. This might increase the order of $\epsilon$ in the current classical upper bound.

\textcolor{black}{Jiao et al.~\cite{jiao2015minimax} also gave a minimax rate-optimal estimator for $\alpha$-power sum $P_{\alpha}$ when $\alpha < 1$ in classical sampling model, since $P_{\alpha}(\mathbf{p}) \ge 1$ this is also an estimator for $\alpha$-R{\'e}nyi entropy. The sample complexity of their estimator is $\mathcal{O}\Bigl(\frac{n^{\frac{1}{\alpha}}}{\log n \epsilon^{\frac{1}{\alpha}}}\Bigr)$ for $\alpha\in(0,\frac{1}{2}]$ and $\mathcal{O}\Bigl(\frac{n^{\frac{1}{\alpha}}}{\log n \epsilon^{\frac{1}{\alpha}}}+\frac{n^{2-2\alpha}}{\epsilon^2}\Bigr)$ for $\alpha\in(\frac{1}{2},1)$. The query complexity of our algorithm in \thm{main} for $\alpha \in (0,1)$ is  $\widetilde{\mathcal{O}}\Bigl(\frac{n^{\frac{1}{2\alpha}}}{\epsilon^{\frac{1}{2\alpha}+1}}\Bigr)$ which is better with respect to both $n$ and $\epsilon$. }


\section{Applications}\label{sec:applications}
\subsection{Extension to quantum entropies}
\noindent
For the diagonal case of purified quantum query-access $U_p$ in \defn{purified-query-access}, we use the block-encoding in \eq{block-2} and denote it by
\begin{align}
    A=\widetilde{\Pi} U \Pi =\sum_{i=1}^{n} \sqrt{p_{i}}\left|\phi_{i}\right\rangle\langle \mathbf{0}|\otimes| i\rangle\langle \mathbf{0}|\otimes| i\rangle\langle i|.
\end{align}
The only difference this new oracle brings is that the state we obtain may have some garbage states added, so we need some ancilla registers to store them. For example, following the process in \sec{main_algo_standard} with purified quantum query-access oracle, we can get a quantum state $|\Psi_p\rangle$ such that
\begin{align}
    |\Psi_{p}\rangle =  &\sum_{i=1}^{n} \sqrt{p_{i}} S\left(\sqrt{p_{i}}\right)|i\rangle_A|i\rangle_B|+\rangle_Q|\phi^{(i)}_{\mathrm{garbage}}\rangle|1\rangle_F \nonumber \\
    &+ |\psi_{\mathrm{garbage}}\rangle|0\rangle_F,
\end{align}
where $|\phi^{(i)}_{\mathrm{garbage}}\rangle = |\phi_i\rangle|\phi_i\rangle$ is brought by the new oracle. Therefore, we can still use the amplitude estimate algorithm to estimate the amplitude of $|1\rangle_F$, which gives us an estimate of $\sum_{i=1}^n p_i S(\sqrt{p_i})^2$. The framework in \sec{annealing} and \sec{main_algo_2} also works well with purified quantum query-access oracle for the same reason. 

For the non-diagonal case of purified quantum query-access $U_{\rho}$ in \defn{purified-query-access}, there are two ways to encode information of $\rho$ by a unitary operator. The first way is to use the \textcolor{black}{projected unitary encoding} in \eq{block-3} \textcolor{black}{proposed by \cite{gilyen2019distributional}}
\begin{align}
  \label{eq:block-encoding-rho-1}
  \widetilde{\Pi}U\Pi=\sum_{i=1}^n \sqrt{\frac{p_i}{n}}\left|\phi_i^{\prime}\right\rangle\langle \mathbf{0}|\otimes| \mathbf{0}\rangle\langle \mathbf{0}|\otimes| \mathbf{0}\rangle\left\langle\psi_i\right|,
\end{align}
and the second is to use the block-encoding in \eq{block-4}
\begin{align}
  \label{eq:block-encoding-rho-2}
  (\langle\mathbf{0}|_{A,B}\otimes I_C)U(|\mathbf{0}\rangle_{A,B}\otimes I_C) &= \sum_{i=1}^n p_i|\psi_i\rangle\langle \psi_i|_C \nonumber\\&= \rho.
\end{align}
The second unitary block-encodes $\rho$ while the first unitary encodes the eigenvalues of $\sqrt{\rho/n}$. Algorithms using different encoding have different query complexities, and we can choose the encoding with a better query complexity. 

We prove the following theorems which give an algorithm to estimate the quantum R{\'e}nyi entropy of density operators 
\begin{align}
    H_{\alpha}(\rho) = \frac{1}{1-\alpha}\log(\Tr(\rho^{\alpha})),
\end{align}
with the purified quantum query-access oracle in \defn{purified-query-access}. 

\begin{corollary}
\label{cor:main-alpha-large-purified}
  For any $\alpha > 1$, there exists an algorithm $\mathcal{A}$ such that for any $\delta \in (0,1)$, $\epsilon\in(0,1)$ and density operator $\rho \in \mathbb{C}^{n\times n}$ , $\mathcal{A}$ can estimate $H_{\alpha}(\rho)$ to within additive error $\epsilon$
  with success probability at least $1-\delta$ using ${\widetilde{\mathcal{O}}\Bigl(\min\Bigl(\frac{n^{\frac{3}{2}-\frac{1}{2\alpha}}}{\epsilon} + \frac{n}{\epsilon^{1+\frac{1}{2\alpha}}}, \frac{n}{\epsilon^{\frac{1}{\alpha}+1}}\Bigr)\Bigr)}$ calls to $U_{\rho}$ and $U_{\rho}^{\dagger}$ in \defn{purified-query-access}.
\end{corollary}
\begin{proof}
We shall present two algorithms using different block-encodings. Taking the algorithm with smaller query complexity gives the claimed statement.

\textbf{Using the encoding in \eq{block-encoding-rho-1}.}
  Let $p_i$ be the eigenvalues of $\rho$. We use the \textcolor{black}{projected unitary encoding} in \eq{block-encoding-rho-1} which encodes the eigenvalues of $\sqrt{\rho/n}$  and follow the same process in \thm{main-alpha-large} to estimate $\Tr(\rho^\alpha)$ by QSVT and VTAE. The only difference is that we need to replace the polynomials $S_j(x)$ in \lem{large_alpha_given_P} with $S_j'(x)$ defined below. Note that $S_j$ for all $i= 0,\ldots,m_0$ in \lem{large_alpha_given_P} is constructed using \lem{scaled_poly_approx_with_small_value_at_zero}, so let $(c,\nu_j,\beta_j,\eta)$ be the parameters of $S_j$ in \lem{scaled_poly_approx_with_small_value_at_zero}. Let $S_j'$ be the polynomial constructed in \lem{scaled_poly_approx_with_small_value_at_zero} with the parameters $(c,\beta,\nu,\eta)$ to be $c := c$, $\nu:=\frac{\nu_j}{\sqrt{n}}$, $\beta:=\frac{\beta_j}{\sqrt{n}}$, $\eta:=\eta$. Then we can infer that $S_j'(\frac{\sqrt{p_i}}{\sqrt{n}})$ has the same behavior as $S_j(\sqrt{p_i})$ for all $p_i$ with $\sqrt{n}$ times larger degree, so $ \sum_{i=1}^n \frac{p_i}{n}S_0'(\sqrt{\frac{p_i}{n}})^2$
  is also an estimate of $P_{\alpha}(\mathbf{p})/n$ to within multiplicative error $\epsilon$. Therefore, we can obtain an estimate of $H_{\alpha}(\mathbf{p})$ to within additive error $\epsilon$ by rescaling the multiplicative error $\epsilon$ to $c\epsilon$ for some constant $c$. Following the same proof in \lem{large_alpha_given_P} and \thm{main-alpha-large} with $\nu_{m_0}' =  \nu_{m_0}/\sqrt{n}$ and $\sum_{i=1}^n \frac{p_i}{n}S_0'(\sqrt{\frac{p_i}{n}})^2 = \Theta(\frac{1}{n} \sum_{i=1}^n p_i S_0(\sqrt{p_i})^2)$, the query complexity becomes 
  \begin{align}
      \widetilde{\mathcal{O}}\left(\frac{n^{1+\frac{1}{2}-\frac{1}{2\alpha}}}{\epsilon} + \frac{\sqrt{n}\sqrt{n}}{\epsilon^{1+\frac{1}{2\alpha}}}\right) =\widetilde{\mathcal{O}}\left(\frac{n^{\frac{3}{2}-\frac{1}{2\alpha}}}{\epsilon} + \frac{n}{\epsilon^{1+\frac{1}{2\alpha}}}\right) .
  \end{align}

  \textbf{Using the block-encoding of $\rho$ in \eq{block-encoding-rho-2}.} This is a special case of \cor{low-rank-large-alpha} for $r = n$, so its query complexity is $\widetilde{\mathcal{O}}\left(\frac{n}{\epsilon^{\frac{1}{\alpha}+1}}\right)$.
\end{proof}

\begin{corollary}\label{cor:main-alpha-small-purified}
  For any $\alpha \in(0,1)$, there exists an algorithm $\mathcal{A}$ such that for any $\delta \in (0,1)$, $\epsilon\in(0,1)$ and density operator $\rho \in \mathbb{C}^{n\times n}$, $\mathcal{A}$ can estimate $H_{\alpha}(\rho)$ to within additive error $\epsilon$
  with success probability at least $1-\delta$ using $ \widetilde{\mathcal{O}}\left(\frac{n^{\frac{1}{2\alpha}+\frac{1}{2}}}{\epsilon^{\frac{1}{2\alpha}+1}}\right)$ calls to $U_{\rho}$ and $U_{\rho}^{\dagger}$ in \defn{purified-query-access}. 
\end{corollary}
\begin{proof}
    Let $p_i$ be the eigenvalues of $\rho$. The proof is essentially the same as that of \cor{main-alpha-large-purified}. We can construct $S_j'(x)$ with $\sqrt{n}$ times larger degree than $S_j$ in \thm{main-alpha-small} such that $S_j'(\sqrt{\frac{p_i}{n}})$ has the same behavior as $S_j(\sqrt{p_i})$ for all $p_i$. With the block encoding in \eq{block-encoding-rho-1}, we can follow the process in \thm{main-alpha-small} to give an algorithm estimating $H_{\alpha}(\rho)$ to within additive error $\epsilon$ with success probability at least $1-\delta$ using $ \widetilde{\mathcal{O}}\left(\frac{n^{\frac{1}{2\alpha}+\frac{1}{2}}}{\epsilon^{\frac{1}{2\alpha}+1}}\right)$ calls to $U_{\rho}$ and $U_{\rho}^{\dagger}$.

\end{proof}

\subsection{Low-rank cases}\label{sec:low-rank}
\noindent
For low-rank quantum distributions (density matrices) or classical distributions with at most $r$ elements with positive probability, we can adapt our algorithm to obtain better query complexity upper bound. 

\noindent
\textbf{Quantum distributions.}
If the rank of the density operator $\rho\in \mathbb{C}^{n\times n}$ is guaranteed to be $ r = o(n)$, we can apply our framework to estimate the R{\'e}nyi entropy of $\rho$ with $\poly(r) = o(n)$ calls to $U_{\rho}$ and $U_{\rho}^{\dagger}$. 

We will use the block-encoding in \eq{block-4}, which constructing a unitary operator $U$ such that
\begin{align}
  (\langle\mathbf{0}|_{A,B}\otimes I_C)U(|\mathbf{0}\rangle_{A,B}\otimes I_C) &= \sum_{i=1}^n p_i|\psi_i\rangle\langle \psi_i|_C \nonumber\\&= \rho,
\end{align}
with one call to $U_\rho$ and $U_\rho^{\dagger}$ respectively. 

For any polynomial $S$ satisfying \eq{qsvt_condition} in \thm{qsvt}, we can apply the singular value transformed unitary of $U$ to $\sum_{i=1}^n\sqrt{p_i}|\mathbf{0}\rangle_{A,B}|\psi_i\rangle_C|\phi_i\rangle_D$ which outputs \begin{align}
    \sum_{i=1}^n\sqrt{p_i}S(p_i)|\mathbf{0}\rangle_{A,B}|\psi_i\rangle_C|\phi_i\rangle_D + |\psi_{\perp}\rangle,
\end{align}
where $\|(\langle\mathbf{0}|_{A,B}\otimes I_{C,D})|\psi_{\perp}\rangle\| = 0$. 

Note that this process is similar to our application of QSVT in \sec{main_algo_standard} except that we have $S(p_i)$ now rather than $S(\sqrt{p_i})$ in \sec{main_algo_standard}.  Therefore, we can use the techniques in our framework to estimate R{\'e}nyi entropy with some minor changes to the transformation polynomials. 

\begin{corollary}
  \label{cor:low-rank-small-alpha}
For any $\alpha \in(0,1)$, there exists an algorithm $\mathcal{A}$ such that for any $\delta \in (0,1)$, $\epsilon\in(0,1)$, and rank-$r$ density operator $\rho\in \mathbb{C}^{n\times n}$, $\mathcal{A}$ can estimate $H_{\alpha}(\rho)$ to within additive error $\epsilon$
  with success probability at least $1-\delta$ using $\widetilde{\mathcal{O}}\left(\frac{r^{\frac{1}{\alpha}}}{\epsilon^{\frac{1}{\alpha}+1}}\right)$ calls to $U_{\rho}$ and $U_{\rho}^{\dagger}$ in \defn{purified-query-access}.
\end{corollary}
\begin{proof}
For any $0<\alpha<1$, we need to change all $n$ with $r$ and the parameters $(\delta,c)$ of $S_j$ for $j=0,\ldots,m_0$ in \thm{main-alpha-small} to $\delta' = \delta^2$, $c'= \frac{c}{2}$ so that the new polynomial $S_j'(p_i)$ has the same behavior as $S_j(\sqrt{p_i})$ for all $p_i$. Then following the proof of \thm{main-alpha-small}, we obtain an upper bound on the quantum query complexity of estimating $H_{\alpha}(\rho)$ to within additive error $\epsilon$ with purified quantum query-access to a rank-$r$ density operator $\rho$ as follows.
\begin{align}
  &\widetilde{\mathcal{O}}\left(\frac{1}{\epsilon}\left(\frac{1}{\delta_{m_0}'} + \frac{\sqrt{r + r}}{\delta_{m_0}'^{c'}}\right)\right) \nonumber \\
  = & \widetilde{\mathcal{O}}\left(\frac{1}{\epsilon}\left(\frac{r{\frac{1}{\alpha}}}{\epsilon^{\frac{1}{\alpha}}} + \frac{\sqrt{r}}{\left(\frac{\epsilon}{r}\right)^{\frac{1}{2}(1-\alpha)/\alpha}}\right)\right) \nonumber \\ 
  = & \widetilde{\mathcal{O}}\left(\frac{r^{\frac{1}{\alpha}}}{\epsilon^{\frac{1}{\alpha} + 1}}+\frac{r^{\frac{1}{2\alpha}}}{\epsilon^{\frac{1}{\alpha}+\frac{1}{2}}}\right) 
  = \widetilde{\mathcal{O}}\left(\frac{r^{\frac{1}{\alpha}}}{\epsilon^{\frac{1}{\alpha}+1}}\right).\label{eq:low-rank-small-alpha}
\end{align}
\end{proof}

\begin{corollary}
  \label{cor:low-rank-large-alpha}
For any $\alpha > 1$, there exists an algorithm $\mathcal{A}$ such that for any $\delta \in (0,1)$, $\epsilon\in(0,1)$, and rank-$r$ density operator $\rho\in \mathbb{C}^{n\times n}$, $\mathcal{A}$ can estimate $H_{\alpha}(\rho)$ to within additive error $\epsilon$
  with success probability at least $1-\delta$ using $\widetilde{\mathcal{O}}\left(\frac{r}{\epsilon^{1+\frac{1}{\alpha}}} \right)$ calls to $U_{\rho}$ and $U_{\rho}^{\dagger}$ in \defn{purified-query-access}.
\end{corollary}
\begin{proof}
For any $\alpha > 1$, we need to change all $n$ to $r$ and the parameters $(\nu,c,\beta)$ of $S_j$ for $j=0,\ldots,m_0$ in \lem{large_alpha_given_P} to $\nu' = \nu^2$, $c' = \frac{c}{2}$, $\beta' = \beta^2$ so that the new polynomial $S_j'(p_i)$ has the same behavior as $2^{\frac{c}{2}}S_j(\sqrt{p_i})$ for all $p_i$. The constant $2^{\frac{c}{2}}$ can be omitted in complexity analysis. Following the same proof in \lem{large_alpha_given_P} and \thm{main-alpha-large}, we can infer that the upper bound is 
\begin{align}
  &\widetilde{\mathcal{O}}\left(\frac{1}{\epsilon}\left(\frac{1}{\nu_{m_0}'} + \frac{\sqrt{r + r}}{\sqrt{L}}\right)\right) \nonumber \\
  = & \widetilde{\mathcal{O}}\left(\frac{1}{\epsilon}\left(\frac{r^{\frac{1}{\alpha}}}{p^*\epsilon^{\frac{1}{\alpha}}} + \frac{\sqrt{r}}{\sqrt{p^*}}\right)\right) \nonumber \\
  = & \widetilde{\mathcal{O}}\left(\frac{r^{\frac{1}{\alpha}}}{(P_{\alpha}(\mathbf{p}))^{\frac{1}{\alpha}}\epsilon^{1+\frac{1}{\alpha}}}+\frac{\sqrt{r}}{(P_{\alpha}(\mathbf{p}))^{\frac{1}{2\alpha}}\epsilon}\right),
\end{align}
which becomes $\widetilde{\mathcal{O}}\left(\frac{r}{\epsilon^{1+\frac{1}{\alpha}}} + \frac{r^{1-\frac{1}{2\alpha}}}{\epsilon}\right) = \widetilde{\mathcal{O}}\left(\frac{r}{\epsilon^{1+\frac{1}{\alpha}}} \right)$ in the worst case. 
\end{proof}

\vspace{3mm}
\noindent
\textbf{Classical distributions.}
For classical distributions, an analogy of low rank density operators is probability distributions $\mathbf{p}$ on $[n]$ such that there are at most $r$ elements $i$ whose probability $p_i > 0$. 
For such probability distributions, if we know $r$ in advance, we can directly obtain the upper bound in \thm{main-alpha-large} and \thm{main-alpha-small} replacing $n$ with $r$, since the proofs of these two theorem still hold if we replace all $n$ with $r$.

If we do not know $r$, we can also use the algorithm in the following corollary to estimate $H_{\alpha}(\mathbf{p})$ for any $\alpha > 1$.

\begin{corollary}
\label{cor:large-alpha-not-r}
  For any $\alpha > 1$, there exists an algorithm $\mathcal{A}$ such that for any $\delta \in (0,1)$, $\epsilon\in(0,1)$ and probability distribution $\mathbf{p}$ on $[n]$ with at most $r$ elements having $p_i>0$, $\mathcal{A}$ can estimate $H_{\alpha}(\mathbf{p})$ to within additive error $\epsilon$
  with success probability at least $1-\delta$ using $\widetilde{\mathcal{O}}\left(\frac{r^{1-\frac{1}{2\alpha}}}{\epsilon} + \frac{\sqrt{r}}{\epsilon^{1+\frac{1}{2\alpha-2}}}\right)$ calls to $U_{\mathrm{pure}}$ and $U_{\mathrm{pure}}^{\dagger}$ in \defn{pure-state-preparation}.
\end{corollary}
\begin{proof}
Changing the parameter $\nu$ to $\nu = \Theta((\epsilon P)^{\frac{1}{2\alpha-2}})$ in \lem{polynomial_large_alpha}, we can replace \eq{large_alpha_little_delta} with 
\begin{align}
  & \sum_{\sqrt{p_i} \le \nu} |p_iS(\sqrt{p_i})^2-2^{-2\alpha}\beta^{-2\alpha+2}p_i^{\alpha}| \nonumber \\
  \le & \sum_{\sqrt{p_i} \le \nu}|p_iS(\sqrt{p_i})^2| + |2^{-2\alpha}\beta^{-2\alpha+2}p_i^{\alpha}| \nonumber \\
  \le & \sum_{\sqrt{p_i} \le \nu} p_i2^{-2\alpha+2}\beta^{-2\alpha+2}p_i^{\alpha-1} + 2^{-2\alpha}\beta^{-2\alpha+2}p_i^{\alpha} \nonumber   \\
  \le & \sum_{\sqrt{p_i} \le \nu} 2^{-2\alpha}\beta^{-2\alpha+2}5p_i^{\alpha} \nonumber  \\
  \le & 2^{-2\alpha}\beta^{-2\alpha+2}5(\frac{1}{\nu^2}\nu^{2\alpha}) \nonumber \\
  \le & 2^{-2\alpha}\beta^{-2\alpha+2}5P\epsilon \nonumber \\
  =& \Theta((p^*)^{1-\alpha}(p^*)^{\alpha}\epsilon) = \Theta(p^*\epsilon),
\end{align}
  so \lem{polynomial_large_alpha} still holds with this $\nu$. Then the algorithm in \lem{large_alpha_given_P} setting $\nu_0$ therein to be $\Theta((\epsilon P)^{\frac{1}{2\alpha-2}})$ can also estimate $P_{\alpha}(\mathbf{p})$ within multiplicative error $\mathcal{O}(\epsilon)$ given a rough bound of $P_{\alpha}(\mathbf{p})$, and the query complexity becomes
  \begin{align}
    &\widetilde{\mathcal{O}}\biggl(\frac{1}{\epsilon}\biggl(t_{m_0} + \sqrt{\epsilon} \sum_{j=1}^{m_0} t_j \nonumber \\ 
    &\qquad\quad + \frac{\sqrt{\sum_{j = 1}^{m_0}\sum_{i:\sqrt{p_i}\in[\varphi_{j},\varphi_{j-1})}p_i t_{j+1}^2}}{\sqrt{\sum_{i = 1}^np_iS_0(\sqrt{p_i})^2}}\biggr)\biggr) \nonumber\\
    =&\widetilde{\mathcal{O}}\left(\frac{1}{\epsilon}\left(\frac{1}{\nu_{m_0}} + \frac{\sqrt{r+r}}{\sqrt{L}}\right)\right) \nonumber\\
    =&\widetilde{\mathcal{O}}\left(\frac{1}{\epsilon}\left(\frac{1}{\epsilon^{\frac{1}{2\alpha-2}}P_{\alpha}(\mathbf{p})^{\frac{1}{2\alpha-2}}} + \frac{\sqrt{r}}{\sqrt{p^*}}\right)\right) \nonumber\\
    =&\widetilde{\mathcal{O}}\left(\frac{1}{\epsilon^{\frac{1}{2\alpha-2}}P_{\alpha}(\mathbf{p})^{\frac{1}{2\alpha-2}}}+\frac{\sqrt{r}}{(P_{\alpha}(\mathbf{p}))^{\frac{1}{2\alpha}}\epsilon}\right), 
\end{align}
where the $\nu_{m_0}$ in the second line is bounded by \eq{nu-m-0}.
Then following the proof in \thm{main-alpha-large}, we can remove the requirement for $P,a,b$ with an $\mathcal{O}(\ln(\alpha))$ overhead in query complexity which can be absorbed into the $\widetilde{\mathcal{O}}$ notation. In the worst case when $P_{\alpha}(\mathbf{p}) = r^{1-\alpha}$, the query complexity becomes $\widetilde{\mathcal{O}}\left(\frac{r^{1-\frac{1}{2\alpha}}}{\epsilon} + \frac{\sqrt{r}}{\epsilon^{1+\frac{1}{2\alpha-2}}}\right)$. 
\end{proof}

\subsection{Quantum R{\'e}nyi divergence}
\noindent
For any rank-$r$ density operators $\rho,\sigma$, M{\"u}ller-Lennert et al.~\cite{muller2013quantum} defined a generalization of the $\alpha$-R{\'e}nyi divergence:
\begin{align}
D_\alpha(\rho \| \sigma):= \begin{cases}\frac{1}{\alpha-1} \log \left( \operatorname{Tr}\left[\left(\sigma^{\frac{1-\alpha}{2 \alpha}} \rho \sigma^{\frac{1-\alpha}{2 \alpha}}\right)^\alpha\right]\right) \\ \qquad\qquad\qquad\qquad\qquad \text { if } \mathrm{Tr}(\rho\sigma)\neq 0 \\ \infty \qquad\qquad\quad\qquad\qquad \text { else }\end{cases}
\end{align}for $\alpha \in(0, 1)$ and prove that it has some good properties. This quantum R{\'e}nyi entropy is also a generalization of fidelity since $D_{\frac{1}{2}}(\rho\|\sigma) = -2\log(F(\rho,\sigma))$. 

Given oracles $U_{\sigma},U_{\rho}$ to prepare purification of the mixed states $\sigma$ and $\rho$, the techniques used to estimate $\mathrm{Tr}(\rho^{\alpha})$ in low-rank cases can be directly applied to estimate $\mathrm{Tr}(\sigma^{\beta}\rho\sigma^{\beta})$ for any  $\beta > 0$. In fact, we can implement a unitary $\tilde{U}$ preparing the purification of $\sigma^{\beta}\rho\sigma^{\beta}$ with $U_{\rho}$ and $U_{\sigma}$ within $\delta^\beta + \epsilon$ additive error in spectral norm using $\mathcal{O}\left(\frac{1}{\delta}\log(\frac{1}{\epsilon})\right)$ calls to $U_{\sigma}$ and two calls to $U_{\rho}$ by choosing the transformation polynomial of $U_\sigma$ to be the one in \lem{scaled_poly_approx_with_small_value_at_zero} with parameters $c:=\beta,\beta:=1,\nu:=\delta,\eta:=\epsilon$. Using amplitude amplification, we can then implement a unitary $U'$ preparing the purification of $\frac{\sigma^{\beta}\rho\sigma^{\beta}}{\mathrm{Tr}(\sigma^{\beta}\rho\sigma^{\beta})}$ within $\epsilon$ error using $\mathcal{O}(\frac{1}{\sqrt{\mathrm{Tr}(\sigma^{\beta}\rho\sigma^{\beta})}}\log(\frac{1}{\epsilon}))$ calls to $\tilde{U}$. 

Therefore, we can apply our results to estimating quantum R{\'e}nyi divergence in the following two steps: 
\begin{enumerate}
    \item Construct a unitary $U’$ which is a block-encoding of $A = \frac{\sigma^{\frac{1-\alpha}{2 \alpha}} \rho \sigma^{\frac{1-\alpha}{2 \alpha}}}{\mathrm{Tr}(\sigma^{\frac{1-\alpha}{2 \alpha}} \rho \sigma^{\frac{1-\alpha}{2 \alpha}})}$  by $U_{\rho}$ and singular value transformed $U_{\sigma}$.
    
    \item Estimate $D_\alpha(\rho \| \sigma)=\frac{1}{\alpha-1} \log((\mathrm{Tr}(\sigma^{\frac{1-\alpha}{2 \alpha}} \rho \sigma^{\frac{1-\alpha}{2 \alpha}}))^{\alpha}\operatorname{Tr}(A^{\alpha}))$ with $U’$ using our techniques to estimate quantum $\alpha$-R{\'e}nyi entropy in \sec{low-rank}.
\end{enumerate}

\begin{corollary}
\textcolor{black}{For any $\epsilon\in (0,1),\alpha \in (0,1)$ and two density operators $\rho,\sigma$ with rank at most $r$, there is an algorithm $\mathcal{A}$ using }
\begin{align}
\widetilde{\mathcal{O}}\left(\frac{1}{\sqrt{\mathrm{Tr}(\sigma^{\beta}\rho\sigma^{\beta})}}\frac{r^{\frac{1}{\alpha}}}{\epsilon^{1+\frac{1}{\alpha}}}\right)
\end{align}
calls to $U_{\rho}$ in \defn{purified-query-access} and 
\begin{align}
\widetilde{\mathcal{O}}\left(\frac{1}{(\mathrm{Tr}((\sigma^{\beta}\rho\sigma^{\beta})^{\alpha}))^{\frac{1}{\alpha\beta}}}\frac{1}{\sqrt{\mathrm{Tr}(\sigma^{\beta}\rho\sigma^{\beta})}}\frac{r^{\frac{1}{\alpha}+\frac{1}{\alpha\beta}}}{\epsilon^{1+\frac{1}{\alpha}+\frac{1}{\alpha\beta}}}\right)
\end{align}
calls to $U_{\sigma}$ in \defn{purified-query-access} to estimate $D_{\alpha}(\rho\|\sigma)$ to error $\epsilon$ with high probability, where $\beta = (1-\alpha)/2\alpha$. 
\end{corollary}

\begin{proof}
  The error analysis is similar to that in Section 4.2 of \cite{wang2022new}. 
  According to \eq{low-rank-small-alpha}, we need $\widetilde{\mathcal{O}}\left(\frac{r^{\frac{1}{\alpha}}}{\epsilon^{1+\frac{1}{\alpha}}}\right)$ calls to $U'$ to obtain an estimate of $\operatorname{Tr}(A^{\alpha})$ within $\epsilon$ multiplicative error. With this estimate, we can then calculate an estimate of $D_\alpha(\rho \| \sigma)$ within $\epsilon$ additive error. 
  
  Each call to $U'$ uses $\tilde{U}$ $\mathcal{\mathcal{O}}
  \left(\frac{1}{\sqrt{\mathrm{Tr}(\sigma^{\beta}\rho\sigma^{\beta})}}\right)$ times since it amplifies $\sigma^{\beta}\rho\sigma^{\beta}$ to $\sigma^{\beta}\rho\sigma^{\beta}/\mathrm{Tr}(\sigma^{\beta}\rho\sigma^{\beta})$. 
  
  Each call to $\tilde{U}$ uses two calls to $U_{\rho}$, so the query number of $U_{\rho}$ is 
  \begin{align}
      \widetilde{\mathcal{O}}\left(\frac{1}{\sqrt{\mathrm{Tr}(\sigma^{\beta}\rho\sigma^{\beta})}}\frac{r^{\frac{1}{\alpha}}}{\epsilon^{1+\frac{1}{\alpha}}}\right).
  \end{align}

  To compute the number of queries to $U_\sigma$, we need to analyze the error induced by $\tilde{U}$. Using $\mathcal{O}\left(\frac{1}{\delta}\log(\frac{1}{\epsilon_1})\right)$ calls to $U_{\sigma}$, the error of $\tilde{U}$ in spectral norm can be bounded by $\delta^{\beta}+\epsilon_1$. Using the following \lem{power-trace-diff}, we can bound the additive error of our estimate of $\mathrm{Tr}((\sigma^{\beta}\rho\sigma^{\beta})^{\alpha})$ induced by $\tilde{U}$ by 
  \begin{align}
      \mathcal{O}\left(r\left(\frac{\delta^\beta+\epsilon_1}{\mathrm{Tr}(\sigma^{\beta}\rho\sigma^{\beta})}\right)^{\alpha}\mathrm{Tr}(\sigma^{\beta}\rho\sigma^{\beta})^{\alpha}\right) = \mathcal{O}\left(r(\delta^{\beta}+\epsilon_1)^{\alpha}\right).
  \end{align}
  \begin{lemma}[{\cite[Lemma 4.6]{wang2022new}}]
    \label{lem:power-trace-diff}
    Suppose that $A$ and $B$ are two positive semidefinite operators of rank $\leq r$, and $0<\alpha<1$. Then
  \begin{align}
  \left|\operatorname{tr}\left(A^\alpha\right)-\operatorname{tr}\left(B^\alpha\right)\right| \leq 5 r\|A-B\|^\alpha,
  \end{align}
  where $\|A-B\|$ is the spectral norm of $A-B$. 
  \end{lemma}
  
  Therefore, the final error induced by error of $\tilde{U}$ is 
  \begin{align}
      \mathcal{O}\left(\frac{r(\delta^{\beta}+\epsilon_1)^{\alpha}}{\mathrm{Tr}((\sigma^{\beta}\rho\sigma^{\beta})^{\alpha})}\right),
  \end{align}
  so we need to set $\delta := \left(\frac{\epsilon}{r}\right)^{\frac{1}{\alpha\beta}}(\mathrm{Tr}((\sigma^{\beta}\rho\sigma^{\beta})^{\alpha}))^{\frac{1}{\alpha\beta}}$ and $\epsilon_1:=\left(\frac{\epsilon}{r}\right)^{\frac{1}{\alpha}}(\mathrm{Tr}((\sigma^{\beta}\rho\sigma^{\beta})^{\alpha}))^{\frac{1}{\alpha}}$ so that the final error is bounded by $\epsilon$. 
  
  The query number of $U_\sigma$ in the final algorithm is 
  \begin{align}
      & \mathcal{O}\left(\frac{1}{\delta}\log(\frac{1}{\epsilon_1})\frac{1}{\sqrt{\mathrm{Tr}(\sigma^{\beta}\rho\sigma^{\beta})}}\frac{r^{\frac{1}{\alpha}}}{\epsilon^{1+\frac{1}{\alpha}}}\right) \nonumber \\
      = & \widetilde{\mathcal{O}}\left(\frac{1}{(\mathrm{Tr}((\sigma^{\beta}\rho\sigma^{\beta})^{\alpha}))^{\frac{1}{\alpha\beta}}}\frac{1}{\sqrt{\mathrm{Tr}(\sigma^{\beta}\rho\sigma^{\beta})}}\frac{r^{\frac{1}{\alpha}+\frac{1}{\alpha\beta}}}{\epsilon^{1+\frac{1}{\alpha}+\frac{1}{\alpha\beta}}}\right).
  \end{align}
\end{proof}

\textcolor{black}{This application is inspired by \cite{wang2022new} in which the author has also studied the estimation of quantum R{\'e}nyi divergence, but they consider the problem of estimating $\operatorname{Tr}\left(\sigma^{\frac{1-\alpha}{2 \alpha}} \rho \sigma^{\frac{1-\alpha}{2 \alpha}}\right)^\alpha = \textrm{exp}((\alpha-1)D_{\alpha}(\rho\|\sigma))$ within certain additive error and we use different polynomial approximations and also techniques beyond QSVT. }


\section*{Acknowledgements}
TL was supported by a startup fund from Peking University, and the Advanced Institute of Information Technology, Peking University.


\appendix
\section*{Proof of \prop{correctness_final}}
\label{append:correctness_final}
\begin{proof}
In the following proof, we use $\Lambda_k$ to denote a $m$-bit 01-string $\Lambda_k:=0^{k-1}10^{m-k}$ and when we say $|\phi\rangle$ is an approximation of $|\psi\rangle$ up to error $\mathcal{O}(L\epsilon)$ or $|\phi\rangle = |\psi\rangle + \mathcal{O}(L\epsilon)$, we mean $\||\psi\rangle-|\phi\rangle\| = \mathcal{O}(L\epsilon)$.

Let $\widetilde{\mathcal{A}}$ be \algo{improved_main_algo_sub} and $\mathcal{A} = \mathcal{A}_m\cdot \ldots \cdot \mathcal{A}_1$ be the unitary operation $\mathcal{A}$ in \algo{improved_main_algo_sub}. 

We first calculate the output of 
\begin{align}
    \widetilde{\mathcal{A}}|\mathbf{0}\rangle = \mathcal{A}(\sum_{i=1}^n\sqrt{p_i}|0\rangle_F|\mathbf{0}\rangle_C|\mathbf{0}\rangle_A|i\rangle_B|+\rangle_Q|\mathbf{0}\rangle_{P,I}). 
\end{align}

Since $\mathcal{A}$ is a linear operator, we only need to calculate 
\begin{align}
    \mathcal{A}(|0\rangle_F|\mathbf{0}\rangle_C|\mathbf{0}\rangle_A|i\rangle_B|+\rangle_Q|\mathbf{0}\rangle_{P,I}). 
\end{align}

We now describe the state in different stages of $\mathcal{A}$ when it is initialized to $|0\rangle_F|\mathbf{0}\rangle_C|\mathbf{0}\rangle_A|i\rangle_B|+\rangle_Q|\mathbf{0}\rangle_{P,I}$ before $\mathcal{A}$. 

Let $j\in \{1,\ldots,m\}$ be such that $\sqrt{p_i} \in [\varphi_j,\varphi_{j-1})$. We divide the $m$ stages of $\mathcal{A}$ into three parts: $\mathcal{A}_1$ to $\mathcal{A}_{j-1}$, $\mathcal{A}_j$ to $\mathcal{A}_{j+1}$, and $\mathcal{A}_{j+2}$ to $\mathcal{A}_{m}$ if exists. 

\vspace{3mm}
\noindent
\textbf{State after $\mathcal{A}_{k}$ for $k = 1,\ldots,j-1$.} The performance of $\mathcal{A}_{k}$ for all $k = 1,\ldots,j-1$ on $|0\rangle_F|\mathbf{0}\rangle_C|\mathbf{0}\rangle_A|i\rangle_B|+\rangle_Q$ is similar. Since $\mathcal{A}_k$ only change the first $k$ registers of $I$ and $P$, the state before $\mathcal{A}_k$ can be written as 
\begin{align}
\label{eq:first_component_k}
    &\alpha_0^{(i,k)}|0\rangle_F|\mathbf{0}\rangle_C|\mathbf{0}\rangle_A|i\rangle_B|+\rangle_Q|\gamma_1\rangle_{P_1,I_1}\nonumber \\
    &\cdots|\gamma_{k-1}\rangle_{P_{k-1},I_{k-1}}|\mathbf{0}\rangle_{P_k,I_k,\ldots,P_m,I_m}+ \alpha_1^{(i,k)}|\psi_{\mathrm{stopped}}\rangle,
\end{align}
where $\|(|\mathbf{0}\rangle\langle\mathbf{0}|_C\otimes I)|\psi_{\mathrm{stopped}}\rangle\| = 0 $ and $|\gamma\rangle_l$ is the state $|\gamma\rangle_{P,I}$ in \lem{separate_singular_value} produced by $W(\varphi_l,L\epsilon/m)$ in $\mathcal{A}_l$. State in registers $F,A,Q$ is $|0\rangle_F|\mathbf{0}\rangle_A|+\rangle_Q$ when state in register $C$ is $|\mathbf{0}\rangle$ since the step 3 of $\mathcal{A}_l$ is conditional on $C_l$ being $|1\rangle$ for all $l = 1,\ldots,k-1$ and only these operations can change the state in registers $F,A,Q$. 

Note that the last $m-k+1$ qubits of $C$ register must be $|0\rangle$, since $\mathcal{A}_l$ for $l = 1,\ldots,k-1$ do not change them. Then, we can infer that $|\psi_{\mathrm{stopped}}\rangle$ has no overlap with $|\mathbf{0}\rangle_{C_1,\ldots,C_{k-1}}$ since $\|(|\mathbf{0}\rangle\langle\mathbf{0}|_C\otimes I)|\psi_{\mathrm{stopped}}\rangle\| = 0 $. Therefore, $\mathcal{A}_k$ will not change $|\psi_{\mathrm{stopped}}\rangle$, and we only need to consider the first component in \eq{first_component_k}. 

After step 1 of $\mathcal{A}_k$, the the component in \eq{first_component_k} becomes
\begin{align}
    & \alpha_0^{(i,k)}|0\rangle_F|\mathbf{0}\rangle_C|\mathbf{0}\rangle_A|i\rangle_B|+\rangle_Q|\gamma_1\rangle_{P_1,I_1}\nonumber \\
    & \cdots|\gamma_{k-1}\rangle_{P_{k-1},I_{k-1}}|0\rangle_{P_k}(|\mathbf{0}\rangle|i\rangle)_{I_k}|\mathbf{0}\rangle_{P_k,I_k,\ldots,P_m,I_m}.  
\end{align}

After step 2 of $\mathcal{A}_k$, the first component in \eq{first_component_k} becomes 
\begin{align}
\label{eq:first_component_k+1}
    &\alpha_0^{(i,k)}\beta_0^{(i,k)}|0\rangle_F|\mathbf{0}\rangle_C|\mathbf{0}\rangle_A|i\rangle_B|+\rangle_Q|\gamma_1\rangle_{P_1,I_1}\nonumber \\
    &\cdots|\gamma_{k-1}\rangle_{P_{k-1},I_{k-1}}|\gamma_k\rangle_{P_k,I_k}|\mathbf{0}\rangle_{P_{k+1},I_{k+1},\ldots,P_m,I_m}\\&+ \alpha_0^{(i,k)}\beta_1^{(i,k)}|0\rangle_F|\Lambda_k\rangle_C|\mathbf{0}\rangle_A|i\rangle_B|+\rangle_Q|\gamma_1\rangle_{P_1,I_1}\nonumber \\
    &\cdots|\gamma_{k-1}\rangle_{P_{k-1},I_{k-1}}|+\rangle_{P_k}(|\mathbf{0}\rangle|i\rangle)_{I_k}|\mathbf{0}\rangle_{P_{k+1},I_{k+1},\ldots,P_m,I_m}, \notag
\end{align}
where $\Lambda_k$ is a $m$-bit 01-string $\Lambda_k:=0^{k-1}10^{m-k}$. 

Since the corresponding singular value of $(|\mathbf{0}\rangle|i\rangle)_{I_k}$ is $\sqrt{p_i}$ and $\sqrt{p_i}<\varphi_{j-1}\le\varphi_{k}$, from \lem{separate_singular_value}, we can infer that $|\beta_1^{(i,k)}|\le\frac{L\epsilon}{m}$ and $|\beta_0^{(i,k)}| = \sqrt{1-|\beta_1^{(i,k)}|^2}\ge \sqrt{1-\left(\frac{L\epsilon}{m}\right)^2}$.

After the step 3 of $\mathcal{A}_k$, the state in \eq{first_component_k+1} does not change. 

Notice that the component in \eq{first_component_k+1} is the component in \eq{first_component_k} for $k+1$. Since the $\alpha_0^{(1)} = 1$, we can infer that the amplitude of \eq{first_component_k+1} satisfies that $|\alpha_0^{(i,k)}\beta_0^{(i,k)}|\ge \sqrt{(1-\left(\frac{L\epsilon}{m}\right)^2)^k} \ge \sqrt{1-(L\epsilon)^2}$ by induction on $k$. 

Therefore, when $k = j-1$, we can infer that after $\mathcal{A}_{j-1}$, the state becomes
\begin{align}
    \label{eq:state_before_j}
    |0\rangle_F|\mathbf{0}\rangle_C|\mathbf{0}\rangle_A|i\rangle_B|+\rangle_Q|\gamma_1\rangle_{P_1,I_1}\nonumber \\ \cdots|\gamma_{j-1}\rangle_{P_{j-1},I_{j-1}}|\mathbf{0}\rangle_{P_{j},I_j,\ldots,P_m,I_m}
\end{align}
up to error $\mathcal{O}(L\epsilon)$.  

Before continuing to the next part, we first consider a special case when $j = m$. In $\mathcal{A}_m$, step 1 and step 2 will map the state in \eq{state_before_j} to 
\begin{align}
    &|0\rangle_F|\Lambda_m\rangle_C|\mathbf{0}\rangle_A|i\rangle_B|+\rangle_Q|\gamma_1\rangle_{P_1,I_1}\nonumber \\ & \cdots|\gamma_{m-1}\rangle_{P_{m-1},I_{m-1}}|0\rangle_{P_m}(|\mathbf{0}\rangle|i\rangle)_{I_m}. 
\end{align}

After step 3 of $\mathcal{A}_m$, the state in \eq{state_before_j} will be 
\begin{align}
    &S_m(\sqrt{p_i})|1\rangle_F|\Lambda_m\rangle_C|i\rangle_A|i\rangle_B|+\rangle_Q|\gamma_1\rangle_{P_1,I_1}\nonumber \\ & \cdots|\gamma_{m-1}\rangle_{P_{m-1},I_{m-1}}|0\rangle_{P_m}(|\mathbf{0}\rangle|i\rangle)_{I_m} \\
    &+ \sqrt{1-S_m(\sqrt{p_i})^2}|0\rangle_F|\Lambda_m\rangle_C|\psi_\mathrm{garbage}^{(m)}\rangle_{A,B,Q,P,I}. 
\end{align}

Since $\sqrt{p_i}\in[\varphi_m,\varphi_{m-2})$ and $S_m$ satisfies the condition in \eq{S_j_condition}, we can infer that
\begin{equation}
\label{eq:final_m}
    \begin{aligned}
    &(|1\rangle\langle1|_F\otimes I)\mathcal{A}(|0\rangle_F|\mathbf{0}\rangle_C|\mathbf{0}\rangle_A|i\rangle_B|+\rangle_Q|\mathbf{0}\rangle_{P,I}) \\
    =& S_m(\sqrt{p_i})|1\rangle_F|\Lambda_m\rangle_C|i\rangle_A|i\rangle_B|+\rangle_Q|\gamma_1\rangle_{P_1,I_1}\nonumber \\ & \cdots|\gamma_{m-1}\rangle_{P_{m-1},I_{m-1}}|0\rangle_{P_m}(|\mathbf{0}\rangle|i\rangle)_{I_m} \\
    \approx& S(\sqrt{p_i})|1\rangle_F|\Lambda_m\rangle_C|i\rangle_A|i\rangle_B|+\rangle_Q|\gamma_1\rangle_{P_1,I_1}\nonumber \\ & \cdots|\gamma_{m-1}\rangle_{P_{m-1},I_{m-1}}|0\rangle_{P_m}(|\mathbf{0}\rangle|i\rangle)_{I_m}\\
    =:&S(\sqrt{p_i})|\Phi_m\rangle,
\end{aligned}
\end{equation}
where the approximation error of the second equation is $\mathcal{O}(L\epsilon)$. 

In the following part, we assume $j  < m$. 

\vspace{3mm}
\noindent
\textbf{State after $\mathcal{A}_j$. }
Since $\mathcal{A}_j$ is linear, we only consider applying $\mathcal{A}_j$ to 
\begin{align}
    &|0\rangle_F|\mathbf{0}\rangle_C|\mathbf{0}\rangle_A|i\rangle_B|+\rangle_Q|\gamma_1\rangle_{P_1,I_1}\nonumber \\ & \cdots|\gamma_{j}\rangle_{P_{j-1},I_{j-1}}|\mathbf{0}\rangle_{P_{j},I_j,\ldots,P_m,I_m},
\end{align}
and the result is also an $\mathcal{O}(L\epsilon)$-approximation of the state after $\mathcal{A}_j$.

After step 1 and step 2 of $\mathcal{A}_j$, we will have
\begin{align}
    \label{eq:first_component_j}
    &\beta_0^{(i,j)}|0\rangle_F|\mathbf{0}\rangle_C|\mathbf{0}\rangle_A|i\rangle_B|+\rangle_Q|\gamma_1\rangle_{P_1,I_1}\nonumber \\ & \cdots|\gamma_{j-1}\rangle_{P_{j-1},I_{j-1}}|\gamma_j\rangle_{P_j,I_j}|\mathbf{0}\rangle_{P_{j+1},I_{j+1},\ldots,P_m,I_m}\\
    \label{eq:second_component_j}
    & + \beta_1^{(i,j)}|0\rangle_F|\Lambda_j\rangle_C|\mathbf{0}\rangle_A|i\rangle_B|+\rangle_Q|\gamma_1\rangle_{P_1,I_1}\nonumber \\ & \cdots|\gamma_{j-1}\rangle_{P_{j-1},I_{j-1}}|+\rangle_{P_j}(|\mathbf{0}\rangle|i\rangle)_{I_j}|\mathbf{0}\rangle_{P_{j+1},I_{j+1},\ldots,P_m,I_m}. 
\end{align}

Since step 3 of $\mathcal{A}_j$ is conditional on register $C_j$ being $|0\rangle$, so the component in \eq{first_component_j} does not change, and the component in \eq{second_component_j} becomes
\begin{align}
    &\beta_1^{(i,j)}S_j(\sqrt{p_i})|1\rangle_F|\Lambda_j\rangle_C|i\rangle_A|i\rangle_B|+\rangle_Q|\gamma_1\rangle_{P_1,I_1}\nonumber \\ & \cdots|\gamma_{j-1}\rangle_{P_{j-1},I_{j-1}}|+\rangle_{P_j}(|\mathbf{0}\rangle|i\rangle)_{I_j}|\mathbf{0}\rangle_{P_{j+1},I_{j+1},\ldots,P_m,I_m}\notag\\
    &+\beta^{(i,j)}_1\sqrt{1-S_j(\sqrt{p_i})^2}|0\rangle_F|\Lambda_j\rangle_C|\psi_\mathrm{garbage}^{(i,j)}\rangle_{A,B,Q,P,I}\label{eq:first_component_in_j}
\end{align}

Since $S_j$ satisfies the condition in \eq{S_j_condition} and $\sqrt{p_i}\in[\varphi_j,\varphi_{j-2})$, we can infer that the component in \eq{first_component_in_j} is an approximation of 
\begin{align}
\hspace{-3mm}&\beta_1^{(i,j)}S(\sqrt{p_i})|1\rangle_F|\Lambda_j\rangle_C|i\rangle_A|i\rangle_B|+\rangle_Q|\gamma_1\rangle_{P_1,I_1}\nonumber \\ & \cdots|\gamma_{j-1}\rangle_{P_{j-1},I_{j-1}}|+\rangle_{P_j}(|\mathbf{0}\rangle|i\rangle)_{I_j}|\mathbf{0}\rangle_{P_{j+1},I_{j+1},\ldots,P_m,I_m}
\end{align}
up to error $\mathcal{O}(L\epsilon)$. 

Therefore, the state after $\mathcal{A}_j$ is 
\begin{align}
    &\beta_0^{(i,j)}|0\rangle_F|\mathbf{0}\rangle_C|\mathbf{0}\rangle_A|i\rangle_B|+\rangle_Q|\gamma_1\rangle_{P_1,I_1}\nonumber \\ & \cdots|\gamma_j\rangle_{P_j,I_j}|\mathbf{0}\rangle_{P_{j+1},I_{j+1},\ldots,P_m,I_m}\notag\\
    +&\beta_1^{(i,j)}S(\sqrt{p_i})|1\rangle_F|\Lambda_j\rangle_C|i\rangle_A|i\rangle_B|+\rangle_Q|\gamma_1\rangle_{P_1,I_1}\nonumber \\ & \cdots|+\rangle_{P_j}(|\mathbf{0}\rangle|i\rangle)_{I_j}|\mathbf{0}\rangle_{P_{j+1},I_{j+1},\ldots,P_m,I_m}\notag\\
    +&\beta^{(i,j)}_1\sqrt{1-S_j(\sqrt{p_i})^2}|0\rangle_F|\Lambda_j\rangle_C|\psi_\mathrm{garbage}^{(i,j)}\rangle_{A,B,Q,P,I} \label{eq:first_component_after_j}
\end{align}
up to error $\mathcal{O}(L\epsilon)$. 

\vspace{3mm}
\noindent
\textbf{State after $\mathcal{A}_{j+1}$. }
Since the first two steps of $\mathcal{A}_{j+1}$ are conditional on the first $j$ qubits of register $C$ being $|\mathbf{0}\rangle$ and the third step is conditional on the $(j+1)$-th qubit of register $C$ being $|1\rangle$, $\mathcal{A}_{j+1}$ will only change the component in \eq{first_component_after_j}. Therefore, we only consider the result applying $\mathcal{A}_{j+1}$ to the component \eq{first_component_after_j}.  

After step 1 and step 2 of $\mathcal{A}_{j+1}$, the component in \eq{first_component_after_j} becomes
\begin{align}
    &\beta_0^{(i,j)}\beta_0^{(i,j+1)}|0\rangle_F|\mathbf{0}\rangle_C|\mathbf{0}\rangle_A|i\rangle_B|+\rangle_Q|\gamma_1\rangle_{P_1,I_1}\nonumber \\ & \cdots|\gamma_{j}\rangle_{P_{j},I_{j}}|\gamma_{j+1}\rangle_{P_{j+1},I_{j+1}}|\mathbf{0}\rangle_{P_{j+2},I_{j+2},\ldots,P_m,I_m}\\
    +&\beta_0^{(i,j)}\beta_1^{(i,j+1)}|0\rangle_F|\Lambda_{j+1}\rangle_C|\mathbf{0}\rangle_A|i\rangle_B|+\rangle_Q|\gamma_1\rangle_{P_1,I_1}\nonumber \\ & \cdots|\gamma_{j}\rangle_{P_{j},I_{j}}|+\rangle_{P_{j+1}}(|\mathbf{0}\rangle|i\rangle)_{I_{j+1}}|\mathbf{0}\rangle_{P_{j+2},I_{j+2},\ldots,P_m,I_m}. 
\end{align}

Since $\sqrt{p_i}\ge \varphi_{j} = 2\varphi_{j+1}$, from \lem{separate_singular_value}, we have $|\beta^{(i,j+1)}_0| \le L\epsilon$. Then, we can infer that the component in \eq{first_component_after_j} after step 1 and step 2 of $\mathcal{A}_{j+1}$ is an approximation of 
\begin{align}
    &\beta_0^{(i,j)}|0\rangle_F|\Lambda_{j+1}\rangle_C|\mathbf{0}\rangle_A|i\rangle_B|+\rangle_Q|\gamma_1\rangle_{P_1,I_1}\nonumber \\ & \cdots|\gamma_{j}\rangle_{P_{j},I_{j}}|+\rangle_{P_{j+1}}(|\mathbf{0}\rangle|i\rangle)_{I_{j+1}}|\mathbf{0}\rangle_{P_{j+2},I_{j+2},\ldots,P_m,I_m}
\end{align}
up to error $\mathcal{O}(L\epsilon)$. 

Then, after step 3 of $\mathcal{A}_{j+1}$, the component in \eq{first_component_after_j} becomes
\begin{align}
    &\beta_0^{(i,j)}S_{j+1}(\sqrt{p_i})|1\rangle_F|\Lambda_{j+1}\rangle_C|i\rangle_A|i\rangle_B|+\rangle_Q|\gamma_1\rangle_{P_1,I_1}\nonumber \\ & \cdots|\gamma_{j}\rangle_{P_{j},I_{j}}|+\rangle_{P_{j+1}}(|\mathbf{0}\rangle|i\rangle)_{I_{j+1}}|\mathbf{0}\rangle_{P_{j+2},I_{j+2},\ldots,P_m,I_m}\notag\\
    &+\beta^{(i,j)}_0\sqrt{1-S_{j+1}(\sqrt{p_i})^2}|0\rangle_F|\Lambda_{j+1}\rangle_C|\psi_\mathrm{garbage}^{(i,j+1)}\rangle_{A,B,Q,P,I}
    \label{eq:first_component_after_j+1}
\end{align}
up to error $\mathcal{O}(L\epsilon)$. 

Since $S_{j+1}$ satisfies the condition in \eq{S_j_condition} and $\sqrt{p_i}\in[\varphi_{j+1},\varphi_{j-1})$, we can infer that the component in \eq{first_component_after_j+1} is an approximation of 
\begin{align}
\hspace{-3mm}&\beta_0^{(i,j)}S(\sqrt{p_i})|1\rangle_F|\Lambda_{j+1}\rangle_C|i\rangle_A|i\rangle_B|+\rangle_Q|\gamma_1\rangle_{P_1,I_1}\nonumber \\ & \cdots|\gamma_{j}\rangle_{P_{j},I_{j}}|+\rangle_{P_{j+1}}(|\mathbf{0}\rangle|i\rangle)_{I_{j+1}}|\mathbf{0}\rangle_{P_{j+2},I_{j+2},\ldots,P_m,I_m}
\end{align}
up to error $\mathcal{O}(L\epsilon)$. 

In conclusion, the state after $\mathcal{A}_{j+1}$ is 
\begin{align}
    &\beta_0^{(i,j)}S(\sqrt{p_i})|1\rangle_F|\Lambda_{j+1}\rangle_C|i\rangle_A|i\rangle_B|+\rangle_Q|\gamma_1\rangle_{P_1,I_1}\nonumber \\ & \cdots|\gamma_{j}\rangle_{P_{j},I_{j}}|+\rangle_{P_{j+1}}(|\mathbf{0}\rangle|i\rangle)_{I_{j+1}}|\mathbf{0}\rangle_{P_{j+2},I_{j+2},\ldots,P_m,I_m}\notag \\
    +&\beta_1^{(i,j)}S(\sqrt{p_i})|1\rangle_F|\Lambda_j\rangle_C|i\rangle_A|i\rangle_B|+\rangle_Q|\gamma_1\rangle_{P_1,I_1}\nonumber \\ & \cdots|+\rangle_{P_j}(|\mathbf{0}\rangle|i\rangle)_{I_j}|\mathbf{0}\rangle_{P_{j+1},I_{j+1},\ldots,P_m,I_m}\notag \\
    +&\beta^{(i,j)}_0\sqrt{1-S(\sqrt{p_i})^2}|0\rangle_F|\Lambda_{j+1}\rangle_C|\psi_\mathrm{garbage}^{(i,j+1)}\rangle_{A,B,Q,P,I}\notag\\
    +&\beta^{(i,j)}_1\sqrt{1-S(\sqrt{p_i})^2}|0\rangle_F|\Lambda_j\rangle_C|\psi_\mathrm{garbage}^{(i,j)}\rangle_{A,B,Q,P,I}     \label{eq:state_after_j+1}
\end{align}
up to error $\mathcal{O}(L\epsilon)$. 

\vspace{3mm}
\noindent
\textbf{State after $\mathcal{A}$. }
Since the state in \eq{state_after_j+1} has no overlap with $|\mathbf{0}\rangle_{C_1,\ldots,C_{j+1}}$, $\mathcal{A}_k$ for all $k = j+2,\ldots,m$ do not change it. 

Therefore, we can infer that state after $\mathcal{A}$, $\mathcal{A}(|0\rangle_F|\mathbf{0}\rangle_C|\mathbf{0}\rangle_A|i\rangle_B|+\rangle_Q|\mathbf{0}\rangle_{P,I})$, satisfies 
\begin{align}
    &(|1\rangle\langle1|_F\otimes I)\mathcal{A}(|0\rangle_F|\mathbf{0}\rangle_C|\mathbf{0}\rangle_A|i\rangle_B|+\rangle_Q|\mathbf{0}\rangle_{P,I})\notag\\
    \approx &\beta_0^{(i,j)}S(\sqrt{p_i})|1\rangle_F|\Lambda_{j+1}\rangle_C|i\rangle_A|i\rangle_B|+\rangle_Q|\gamma_1\rangle_{P_1,I_1}\nonumber \\ & \cdots|\gamma_{j}\rangle_{P_{j},I_{j}}|+\rangle_{P_{j+1}}(|\mathbf{0}\rangle|i\rangle)_{I_{j+1}}|\mathbf{0}\rangle_{P_{j+2},I_{j+2},\ldots,P_m,I_m}\notag\\
    &+\beta_1^{(i,j)}S(\sqrt{p_i})|1\rangle_F|\Lambda_j\rangle_C|i\rangle_A|i\rangle_B|+\rangle_Q|\gamma_1\rangle_{P_1,I_1}\nonumber \\ & \cdots|+\rangle_{P_j}(|\mathbf{0}\rangle|i\rangle)_{I_j}|\mathbf{0}\rangle_{P_{j+1},I_{j+1},\ldots,P_m,I_m}\notag\\
     =: &S(\sqrt{p_i})|1\rangle_F|\Phi_i\rangle, \label{eq:A_on_singular_vector}
\end{align}
where the approximation error of the first equation is $\mathcal{O}(L\epsilon)$ and $|\Phi_i\rangle$ is a normalized state. 

\vspace{3mm}
\noindent
\textbf{Correctness. }
From \eq{A_on_singular_vector} and \eq{final_m}, we can infer that 
\begin{align}
    &(|1\rangle\langle1|_F\otimes I)\mathcal{A}(\sum_{i=1}^n\sqrt{p_i}|0\rangle_F|\mathbf{0}\rangle_C|\mathbf{0}\rangle_A|i\rangle_B|+\rangle_Q|\mathbf{0}\rangle_{P,I}) \nonumber \\
    \approx & \sum_{i = 1}^n\sqrt{p_i}S(\sqrt{p_i})|1\rangle_F|\Phi_i\rangle,
\end{align}
where the approximation error is $\mathcal{O}(L\epsilon)$. 

Using variable-time amplitude estimation, we can estimate 
\begin{align}
& \|(|1\rangle\langle1|_F\otimes I)\mathcal{A}(|0\rangle_F|\mathbf{0}\rangle_C|\mathbf{0}\rangle_A|\psi_p\rangle_B|+\rangle_Q|\mathbf{0}\rangle_{P,I})\|^2 \nonumber \\ = & \sum_{i = 1}^n p_i S(\sqrt{p_i})^2 + \mathcal{O}(L\epsilon)\nonumber \\
= & \sum_{i=1}^n p_i S(\sqrt{p_i})^2(1 + \mathcal{O}(\epsilon))
\end{align}
within multiplicative error $\epsilon$ with success probability at least $1-\delta$.

Since $L \le \sum_{i = 1}^n p_i S(\sqrt{p_i})^2 $, we have $\sum_{i = 1}^n p_i S(\sqrt{p_i})^2 + \mathcal{O}(L\epsilon) = \sum_{i=1}^n p_i S(\sqrt{p_i})^2(1 + \mathcal{O}(\epsilon))$.
Therefore, the output of \algo{main_algo_improved} is an estimate of $\sum_{i = 1}^n p_i S(\sqrt{p_i})^2 $ within multiplicative error $\mathcal{O}(\epsilon)$. 

By rescaling $\epsilon$ to $c\epsilon$ with a small constant $c$, we can estimate $\sum_{i = 1}^n p_i S(\sqrt{p_i})^2 $ to within multiplicative error $\epsilon$ with the same query complexity. 

\vspace{1mm}
\noindent
\textbf{Complexity. }
Our \algo{main_algo_improved} is a direct use of variable-time amplitude estimation in \thm{vtae}, so in order to get its query complexity, we only need to calculate $t_j,p_{\mathrm{stop}=t_j},T_{\mathrm{avg}}$, and $p_{\mathrm{succ}}$ of $\widetilde{\mathcal{A}}$. 
We will calculate these parameters for $\mathcal{A}$ for simplicity since $\widetilde{\mathcal{A}}$ only use $U_p$ one more time than $\mathcal{A}$. 

\vspace{1mm}
\noindent
\textbf{Calculate $p_{\mathrm{succ}}$. }
In the previous paragraph, we have proved that $p_{\textrm{succ}} = \sum_{i=1}^n p_i S(\sqrt{p_i})^2(1 + \mathcal{O}(\epsilon))$. 

\vspace{3mm}
\noindent
\textbf{Calculate $t_j$. }
The query complexity of $\mathcal{A}_j$ for $j < m$ is the sum of query complexity of $W(\varphi_j,L\epsilon/m)$ and $U_{S_j}^{(SV)}$, which is $\mathcal{O}\left(\frac{1}{\varphi_j}\log(\frac{m}{\epsilon L}) + \deg(S_j)\right)$ for $j < m$ while the query complexity of $\mathcal{A}_m$ is $\mathcal{O}(\deg(S_m))$.

Then the sum of the query complexity of the first $j$ stages of $\mathcal{A}$ for $j < m$ is 
\begin{align}
    t_j &= \mathcal{O}\biggl(\sum_{k = 1}^{j}\biggl(\frac{1}{\varphi_k}\log(\frac{m}{\epsilon L}) + \deg(S_k)\biggr)\biggr) \nonumber\\
    &= \mathcal{O}\biggl(\sum_{k = 1}^{j}\biggl(\frac{2^{k}}{\beta}\log(\frac{m}{\epsilon L}) + \deg(S_k)\biggr)\biggr)\nonumber\\
    &= \mathcal{O}\biggl(\frac{2^{j}}{\beta}\log(\frac{m}{\epsilon L}) + \sum_{k = 1}^{j}\deg(S_k)\biggr),
\end{align}
and $t_m = \mathcal{O}\left(\frac{2^{m-1}}{\beta}\log(\frac{m}{\epsilon L}) + \sum_{k=1}^{m}\deg(S_k)\right)$.

\vspace{1mm}
\noindent
\textbf{Calculate $p_{\mathrm{stop}=t_j}$. }
Note that 
\begin{align}
    p_{\mathrm{stop}=t_j} = & \|(|\Lambda_j\rangle\langle \Lambda_j|_C\otimes I)\mathcal{A}\nonumber \\
    &\cdot(\sum_{i=1}^n \sqrt{p_i}|0\rangle_F|\mathbf{0}\rangle_C|\mathbf{0}\rangle_A|i\rangle_B|+\rangle_Q|\mathbf{0}\rangle_{P,I})\|^2 \nonumber\\
    = & \sum_{i=1}^n p_i\|(|\Lambda_j\rangle\langle \Lambda_j|_C\otimes I)\mathcal{A}\nonumber \\
    &\cdot(|0\rangle_F|\mathbf{0}\rangle_C|\mathbf{0}\rangle_A|i\rangle_B|+\rangle_Q|\mathbf{0}\rangle_{P,I})\|^2
\end{align}

To simplify the writting of formulas, we define $Q_j= \{i: \sqrt{p_i}\in[\varphi_j,\varphi_{j-1})\}$ for $j=1,\ldots,m$, and $Q_0 = \emptyset$ in the following proof.

Then from \eq{state_after_j+1}, we can infer that for $i$ such that $\sqrt{p_i} \not\in [\varphi_j,\varphi_{j-2})$, they contribute at most $\mathcal{O}((L\epsilon)^2)$ to $p_{\mathrm{stop}=t_j}$, so we only need to consider contribute of $i$ such that $\sqrt{p_i} \in [\varphi_j,\varphi_{j-2})$. Therefore, we can infer that $p_{\mathrm{stop}=t_j}$ equals
\begin{align}
\label{eq:p_stop_j}
    \sum_{i\in Q_j} p_i |\beta_1^{(i,j)}|^2 + \sum_{i\in Q_{j-1}} p_i |\beta_0^{(i,j-1)}|^2+\mathcal{O}(L\epsilon),
\end{align}
where $\beta_0^{(i,0)} = \beta_0^{(i,m)} := 0$ and $\beta_1^{(i,m)} := 1$.

\vspace{1mm}
\noindent
\textbf{Calculate $T_{\mathrm{avg}}$. }
Let $t_{m+1} =t_m$. Then we have 
\begin{align}
    & T_{\mathrm{avg}}^2 \nonumber \\
    = & \sum_{j=1}^m p_{\mathrm{stop}=t_j} t_j^2\nonumber\\
    \le & \sum_{j=1}^m\Bigl(\sum_{i\in Q_j} p_i |\beta_1^{(i,j)}|^2 + \sum_{i\in Q_{j-1}} p_i |\beta_0^{(i,j-1)}|^2 + \mathcal{O}(L\epsilon )\Bigr)t_j^2 \nonumber\\
    \le & \sum_{j=1}^m\Bigl(\sum_{i\in Q_j} p_i |\beta_1^{(i,j)}|^2t_{j+1}^2+ \sum_{i\in Q_{j-1}} p_i |\beta_0^{(i,j-1)}|^2t_j^2\Bigr)\nonumber \\
    & + \mathcal{O}(L\epsilon \sum_{j=1}^m t_j^2)\nonumber\\
    = & \sum_{j =1}^m\sum_{i\in Q_j} p_i (|\beta_1^{(i,j)}|^2 + |\beta_0^{(i,j)}|^2)t_{j+1}^2+ \mathcal{O}(L\epsilon \sum_{j=1}^m t_j^2)\nonumber\\
    = & \sum_{j =1}^m\sum_{i\in Q_j} p_i t_{j+1}^2+ \mathcal{O}(L\epsilon \sum_{j=1}^m t_j^2).
\end{align}

From \thm{vtae}, we can infer the query complexity of \algo{main_algo_improved} is 
\begin{align}
    &\widetilde{\mathcal{O}}\left(t_m+\frac{T_{\mathrm{avg}}}{\sqrt{p_{\mathrm{succ}}}}\right)\nonumber\\
    =& \widetilde{\mathcal{O}}\left(t_m + \frac{\sqrt{\sum_{j =1}^m\sum_{i\in Q_j} p_i t_{j+1}^2}+\sqrt{L\epsilon} \sum_{j=1}^m t_j}{\sqrt{\sum_{i=1}^n p_i S(\sqrt{p_i})^2}}\right)\nonumber\\
    =& \widetilde{\mathcal{O}}\left(t_m + \sqrt{\epsilon} \sum_{j=1}^m t_j + \frac{\sqrt{\sum_{j =1}^m\sum_{i\in Q_j} p_i t_{j+1}^2}}{\sqrt{\sum_{i=1}^n p_i S(\sqrt{p_i})^2}}\right),
\end{align}
where the second inequality comes from $L\le \sum_{i=1}^n p_i S(\sqrt{p_i})^2$. 
\end{proof}

\section*{Other proofs}
\label{append:other_proof}
\subsection{Proof of \lem{rough_amplitude_estimation}}
\begin{proof}
Let $\Pi = |\mathbf{0}\rangle\langle \mathbf{0}|_{\mathcal{H}_F}$, $\widetilde{\Pi} = |1\rangle\langle1|_{\mathcal{H}_F}\otimes I_{\mathcal{H}_W}$, then $\widetilde{\Pi}\mathcal{A}\Pi$ has only one singular value $\sqrt{p_{\mathrm{succ}}}$. Then using \lem{separate_singular_value}, we can determine whether $\sqrt{p_{\mathrm{succ}}}$ is larger than $2\varphi$ or smaller than $\varphi$ for a given $\varphi\in(0,1)$ with success probability at least $1-\delta$ using $\log(\frac{1}{\delta})\frac{1}{\varphi}$ calls to $\mathcal{A}$ and $\mathcal{A}^{\dagger}$. 

Then setting $\varphi = 1,\frac{1}{2},\ldots,\frac{1}{2^{\lceil \log(\frac{1}{L}) \rceil}}$ sequentially, we can determine whether $\sqrt{p_{\mathrm{succ}}}\ge 2\varphi$ with success probability $1-\frac{\delta}{\log(\frac{1}{L})}$ using $\mathcal{O}\left(\log(\frac{\log(\frac{1}{L})}{\delta})\frac{1}{\varphi}\right)$ calls to $\mathcal{A}$ and $\mathcal{A}^{\dagger}$, if so, stop and output $2\varphi$.

Then with success probability at least $1-\delta$, the algorithm will stop at $\varphi = \frac{1}{2^{\left\lceil \log(\frac{1}{\sqrt{p_{\mathrm{succ}}}})\right\rceil+1}} = \Theta(\sqrt{p_{\mathrm{succ}}})$, so the output is in $[\frac{1}{2}\sqrt{p_{\mathrm{succ}}}, 2\sqrt{p_{\mathrm{succ}}}]$, and the total calls to $\mathcal{A}$ and $\mathcal{A}^{\dagger}$ is
\begin{align}
    \mathcal{O}\left(\log(\frac{\log(\frac{1}{L})}{\delta})\log(\frac{1}{\sqrt{p_{\mathrm{succ}}}})\frac{1}{\sqrt{p_{\mathrm{succ}}}}\right). 
\end{align}
\end{proof}
\subsection{Proof of \lem{bound_of_annealing}}
\begin{proof}
    Let $\gamma_n := (xg(x))^{-1}(\frac{1}{a}g(\frac{1}{n}))$. 
    For any probability distribution $\mathbf{p} = (p_i)_{i=1}^n$, let $p_M = \max_{i\in[n]} p_i$ and $i_M = \mathrm{argmax}_{i\in[n]} p_i$.  
    
    Since $xg(x)$ is a convex function, we have
    \begin{align}
        \frac{\sqrt{\sum_{i=1}^n p_i g(p_i)^2}}{\sum_{i=1}^n p_i g(p_i)} 
        & \le \frac{\sqrt{g(p_M)}}{\sqrt{\sum_{i=1}^n p_i g(p_i)}} \notag\\
        & =  \frac{\sqrt{g(p_M)}}{\sqrt{p_M g(p_M) + \sum_{i\neq i_M} p_i g(p_i)}} \notag\\
        & \le \frac{\sqrt{g(p_M)}}{\sqrt{p_M g(p_M) + (n-1) \frac{1-p_M}{n-1}g(\frac{1-p_M}{n-1})}}\notag \\
        & = \frac{1}{\sqrt{p_M + (1-p_M) \frac{g(\frac{1-p_M}{n-1})}{g(p_M)}}}, \label{eq:p_M}
    \end{align}
    where the second inequality comes from $xg(x)$ is a convex function on $[0,1]$ and Jensen's inequality. 
    
    Since $g(x)$ and $(xg(x))^{-1}$ are monotonically functions on $[0,1]$ and $g(0) = (xg(x))^{-1}(0) = 0$, we have $\lim_{x\to 0} g(x) =\lim_{x\to 0} (xg(x))^{-1} =  0$. Thus we have $\lim_{n \to \infty } \gamma_n = \lim_{n \to \infty }(xg(x))^{-1}(\frac{1}{a}g(\frac{1}{n})) = 0 $. 
    
    If $p_M \ge \gamma_n$, from \eq{p_M}, we have $\frac{\sqrt{\sum_{i=1}^n p_i g(p_i)^2}}{\sum_{i=1}^n p_i g(p_i)} \le \frac{1}{\sqrt{\gamma_n}}$. 
    
    If $p_M < \gamma_n$, we have 
    \begin{align}
        &\lim_{n \to \infty}\left(\frac{1}{\gamma_n}\left( p_M +  (1-p_M) \frac{g(\frac{1-p_M}{n-1})}{g(p_M)} \right)\right)\nonumber \\
        > & \lim_{n \to \infty} (1-\gamma_n)\frac{g(\frac{1-\gamma_n}{n-1})}{\gamma_n g(\gamma_n)} 
        = \lim_{n \to \infty} \frac{g(\frac{1}{n})}{g(\gamma_n)\gamma_n}
        = 2,
    \end{align}
    where the first equation comes from convex function is continuous and $\lim_{n \to \infty} \gamma_n = 0$, and the second equation comes from $\gamma_n g(\gamma_n) = \frac{1}{a}g(\frac{1}{n})$. Then we can infer that for sufficiently large $n$, 
    \begin{align}
    \label{eq:gamma_n_O}
        \frac{\sqrt{\sum_{i=1}^n p_i g(p_i)^2}}{\sum_{i=1}^n p_i g(p_i)} \le \frac{1}{\sqrt{p_M + (1-p_M) \frac{g(\frac{1-p_M}{n-1})}{g(p_M)}}} < \frac{1}{\sqrt{\gamma_n}}.
    \end{align}
    
    Therefore, we can infer that $\frac{\sqrt{\sum_{i=1}^n p_i g(p_i)^2}}{\sum_{i=1}^n p_i g(p_i)} = \mathcal{O}\left(\frac{1}{\sqrt{\gamma_n}}\right)$ as $n\to \infty$.
    
    If we choose $p_1^{(n)} = \gamma_n$ and $p_i^{(n)} = \frac{1-\gamma_n}{n-1}$ for $i = 2,\ldots,n$, we have 
    \begin{align}
        \label{eq:gamma_n_Omega}
        &\lim_{n \to \infty}\left(\sqrt{\gamma_n}\frac{\sqrt{\sum_{i=1}^n p_i^{(n)} g(p_i^{(n)})^2}}{\sum_{i=1}^n p_i^{(n)} g(p_i^{(n)})}\right) \nonumber \\
        = & \lim_{n \to \infty} \left( \sqrt{\gamma_n}\frac{\sqrt{\gamma_n g(\gamma_n)^2 + (1-\gamma_n)g(\frac{1-\gamma_n}{n-1})^2}}{\gamma_n g(\gamma_n) + (1-\gamma_n)g(\frac{1-\gamma_n}{n-1})}\right) \nonumber\\
        = & \lim_{n \to \infty}\left(\sqrt{\gamma_n}\frac{\sqrt{\gamma_n g(\gamma_n)^2 + g(\frac{1}{n})^2}}{\gamma_n g(\gamma_n) + g(\frac{1}{n})}\right) \nonumber\\
        = & \lim_{n \to \infty}\left(\sqrt{\gamma_n} \frac{\sqrt{\frac{1}{4\gamma_n} g(\frac{1}{n})^2+ g(\frac{1}{n})^2}}{\frac{3}{2}g(\frac{1}{n})}\right)\nonumber\\
        = & \frac{1}{3},
    \end{align}
    where the fourth equation comes from $ \lim_{n \to \infty}\gamma_n = 0$.
    
    From \eq{gamma_n_O} and \eq{gamma_n_Omega}, we can infer that $\max_{\mathbf{p} = (p_i)_{i=1}^n}\frac{\sqrt{\sum_{i=1}^n p_i g(p_i)^2}}{\sum_{i=1}^n p_i g(p_i)} = \Theta\left(\frac{1}{\sqrt{\gamma_n}}\right)$ as $n\to\infty$, which completes the proof. 
\end{proof}

\subsection{Proof of \lem{separate_singular_value}}
\begin{proof}
We first apply Hadamard gate $H$ to register $P$ and obtain
\begin{align}
    |0\rangle_C|+\rangle_P|\psi_i\rangle_I. 
\end{align}

Setting the parameters $(\delta',\epsilon',t)$ in \lem{rectangle_approximation} to $\delta' := \frac{1}{2}\varphi,t := \frac{3}{2}\varphi, \epsilon' := \frac{\epsilon^2}{2}$, we can construct an even polynomial $S:=P'$ in \lem{rectangle_approximation} with $\deg(S) = O\left(\log(\frac{1}{\epsilon})/\varphi\right)$ such that 
\begin{equation}
    \label{eq:property_rectangle}
    \begin{aligned}
    \forall x \in [-1,-2\varphi]\cup[2\varphi,1]:  &S(x)\in[0,\frac{\epsilon^2}{2}] \text{, and}\\
    \forall x \in[-\varphi,\varphi]: &S(x)\in[1-\frac{\epsilon^2}{2},1].
\end{aligned}
\end{equation}

Apply $S^{(SV)}(\widetilde{\Pi}U{\Pi})$ to register $I$ using register $P$ as ancilla register. Then we will have the state $|\Phi\rangle$ in register $(P,I)$ such that 
\begin{align}
    (\langle +|_P \otimes \Pi)|\Phi\rangle 
    = & S^{(SV)}(\widetilde{\Pi}U{\Pi})|\psi_i\rangle_I \nonumber \\
    = & (\sum_{i = 1}^{d}S(\sigma_i)|\psi_i\rangle \langle \psi_i|)|\psi_i\rangle_I = S(\sigma_i)|\psi_i\rangle_I,
\end{align}
where the first equation comes from $\Pi|\psi_i\rangle = |\psi_i\rangle$.
Therefore, we can infer that the state $|\Phi\rangle$ satisfies
\begin{align}
    |\Phi\rangle = S(\sigma_i)|+\rangle_P|\psi_i\rangle_I + \sqrt{1-S(\sigma_i)^2}|\gamma\rangle_{P,I},
\end{align}
where $(\langle +|_P\otimes \Pi)|\gamma\rangle_{P,C} = 0$, since $\Pi$ is an orthogonal projection. 

Next, we apply $\text{C}_{|+\rangle\langle+|\otimes\Pi}\text{NOT} = |+\rangle\langle+|\otimes\Pi\otimes X + (I-|+\rangle\langle+|\otimes\Pi)\otimes I$ to register $(P,I,C)$ and obtain
\begin{align}
    &\text{C}_{|+\rangle\langle+|\otimes\Pi}\text{NOT}|\Phi\rangle_{P,I}|0\rangle_C\nonumber\\
    =& (|+\rangle\langle+|\otimes\Pi\otimes X)S(\sigma_i)|+\rangle_P|\psi_i\rangle_I|0\rangle_C + \nonumber \\
    &(I-|+\rangle\langle+|\otimes\Pi)\sqrt{1-S(\sigma_i)^2}|\gamma\rangle_{P,I}|\gamma\rangle_C\nonumber\\
    =& S(\sigma_i)|+\rangle_P|\psi_i\rangle_I|1\rangle_C + \sqrt{1-S(\sigma_i)^2}|\gamma\rangle_{P,I}|0\rangle_C,
\end{align}
where the first equation comes from $(\langle +|_P\otimes \Pi)|\gamma\rangle_{P,C} = 0$ and $|+\rangle\langle+|\otimes\Pi$ is an orthogonal projector. 

From \eq{property_rectangle}, we can infer that
\begin{align}
    \forall \sigma_i \in [0,\varphi]:  \sqrt{1-S(\sigma_i)^2} &\le \sqrt{1-\left(1-\frac{\epsilon^2}{2}\right)^2} \nonumber \\
    &\le \sqrt{1-(1-\epsilon^2)} = \epsilon \\
    \forall \sigma_i \in [2\varphi,1]:  S(\sigma_i) &\le \frac{\epsilon^2}{2}\le \epsilon,
\end{align}
which completes the proof. 
\end{proof}

\bibliographystyle{IEEEtran}
\bibliography{Bib-svt-entropy}

\vspace{11pt}

\begin{IEEEbiographynophoto}{Xinzhao Wang} received the B.S. degree in computer science from Peking University, China, in 2022. He is currently pursuing a Ph.D. degree in computer science at Peking University. His research interest lies in quantum computing and quantum information.
\end{IEEEbiographynophoto}

\begin{IEEEbiographynophoto}{Shengyu Zhang} is a Distinguished Scientist at Tencent, and Director of Tencent Quantum Lab. He received his Ph.D. in computer science at Princeton University in 2006. He then worked in California Institute of Technology and The Chinese University of Hong Kong before joining Tencent in 2018. His research interest lies in quantum computing theory, algorithm designing, computational complexity, foundation of machine learning, and AI for sciences. 
\end{IEEEbiographynophoto}

\begin{IEEEbiographynophoto}{Tongyang Li} received the B.E. degree in computer science and the B.S. degree in mathematics from Tsinghua University, China, in 2015. He received the M.S. and Ph.D. degree in computer science from University of Maryland, USA, in 2018 and 2020, respectively. 

From 2020 to 2021, he was a Postdoctoral Associate at the Center for Theoretical Physics, Massachusetts Institute of Technology, USA. Since 2021, he has been an Assistant Professor at Center on Frontiers of Computing Studies, School of Computer Science, Peking University, China. His research focuses on quantum algorithms, including topics such as quantum algorithms for machine learning and optimization, quantum query complexity, quantum simulation, quantum walks, etc. 

Dr. Tongyang Li was a recipient of the IBM Ph.D. Fellowship, the NSF QISE-NET Triplet Award, the Lanczos Fellowship, and Outstanding Reviewer Awards for ICML 2020 and ICML 2022.
\end{IEEEbiographynophoto}

\vfill

\end{document}